\documentclass[a4paper,11pt, fleqn]{article}
\usepackage[a4paper, total={6.5in, 10in}]{geometry}
\usepackage{algorithmicx}
\usepackage{amsmath}
\usepackage{amssymb}
\usepackage{upgreek}
\usepackage{multirow}
\usepackage{dsfont}
\usepackage{amsthm}
\usepackage{amstext}
\usepackage{hyperref}
\usepackage{cleveref}
\usepackage{enumerate}
\usepackage{array}
\usepackage{caption}
\usepackage{graphicx}
\usepackage[round]{natbib}
\usepackage[title]{appendix}
\usepackage{authblk}
\usepackage[ruled,vlined]{algorithm2e}

\usepackage{enumitem}

\DeclareMathOperator*{\argmin}{arg\,min}

\newcommand{\abs}[1]{\left\lvert#1\right\rvert}
\newcommand{\norm}[1]{\left\lVert#1\right\rVert}
\newcommand{\E}{\mathbb{E}}
\renewcommand{\P}{\mathbb{P}}
\newcommand{\smaxN}{\bar{s}_r}
\newcommand{\boundCT}{\frac{h}{\delta_T}}
\newcommand{\setEP}[1]{\mathcal{E}_{#1}} % Notation for empirical process sets
\newcommand{\setCC}{\mathcal{CC}_T} % Notation for covariance closeness sets
\newcommand{\setLL}{\mathcal{L}_T} % Notation for LLN sets (Ass 7)

\newcommand{\setILcons}{\mathcal{P}_{T,las}} % The "initial lasso consistency" set
\newcommand{\setNWcons}{\mathcal{P}_{T,nw}} % The "nodewise consistency" set
\newcommand{\setEPvuj}[1]{\setEP{#1, uv}^{(j)}}

\newcommand{\setEPl}[1]{\setEP{T}^l}
\newcommand{\setUVcons}{\mathcal{P}_{T,uv}}
\newcommand{\setTail}[2]{\mathcal{E}_{T, {#1}} (#2)}

\newcommand{\bzero}{\boldsymbol{0}}
\newcommand{\by}{\boldsymbol{y}}
\newcommand{\bx}{\boldsymbol{x}}
\newcommand{\bq}{\boldsymbol{q}}
\newcommand{\bu}{\boldsymbol{u}}
\newcommand{\bv}{\boldsymbol{v}}

\newcommand{\bb}{\boldsymbol{b}}
\newcommand{\bB}{\boldsymbol{B}}

\newcommand{\bz}{\boldsymbol{z}}
\newcommand{\bs}{\boldsymbol{s}}
\newcommand{\br}{\boldsymbol{r}}
\newcommand{\bX}{\boldsymbol{X}}

\newcommand{\bR}{\boldsymbol{R}}
\newcommand{\bw}{\boldsymbol{w}}

\newcommand{\bg}{\boldsymbol{g}}

\newcommand{\bI}{\boldsymbol{I}}
\newcommand{\bm}{\boldsymbol{m}}

\newcommand{\bbeta}{\boldsymbol{\beta}}

\newcommand{\bXi}{\boldsymbol{\Xi}}

\newcommand{\bgamma}{\boldsymbol{\gamma}}
\newcommand{\bPhi}{\boldsymbol{\Phi}}
\newcommand{\bphi}{\boldsymbol{\phi}}

\newcommand{\bSigma}{\boldsymbol{\Sigma}}
\newcommand{\bOmega}{\boldsymbol{\Omega}}
\newcommand{\bTheta}{\boldsymbol{\Theta}}
\newcommand{\bUpsilon}{\boldsymbol{\Upsilon}}
\newcommand{\bGamma}{\boldsymbol{\Gamma}}

\newcommand{\diag}{\text{diag}}
\newcommand{\ba}{\boldsymbol{a}}

\newcommand{\bA}{\boldsymbol{A}}

\newcommand{\bbf}{\boldsymbol{f}}

\newcommand{\bnu}{\boldsymbol{\nu}}

\newcommand{\bLambda}{\boldsymbol{\Lambda}}

\newcommand{\Ntwo}{\left\lceil{N/2}\right\rceil}
\newcommand{\mcH}{\mathcal{H}}

\title{Lasso Inference for High-Dimensional Time Series}
\author[a]{Robert Adamek\footnote{Present Address: Department of Economics and Business Economics, Aarhus University, Fuglesangs All\'e 4, 8210 Aarhus V, Denmark.}}
\author[a]{Stephan Smeekes\thanks{Correspondence to: S.~Smeekes at Maastricht University, Department of Quantitative Economics, P.O. Box 616, 6200 MD Maastricht, The Netherlands. E-mail addresses: 
\href{mailto:r.adamek@econ.au.dk}{r.adamek@econ.au.dk} (R.~Adamek), \href{mailto:s.smeekes@maastrichtuniversity.nl}{s.smeekes@maastrichtuniversity.nl} (S. Smeekes), \href{mailto:i.wilms@maastrichtuniversity.nl}{i.wilms@maastrichtuniversity.nl} (I.~Wilms)}}
\author[a]{Ines Wilms}
\affil[a]{Department of Quantitative Economics,
Maastricht University, The Netherlands}
\date{\today}

\newtheorem{definition}{Definition}
\newtheorem{theorem}{Theorem}
\newtheorem{corollary}{Corollary}
\newtheorem{lemma}{Lemma}
\crefname{lemma}{Lemma}{Lemmas}
\crefname{corollary}{Corollary}{Corollaries}
\crefname{theorem}{Theorem}{Theorems}
\crefname{definition}{Definition}{Definitions}
\crefname{example}{Example}{Examples}
\crefname{remark}{Remark}{Remarks}
\crefname{assumption}{Assumption}{Assumptions}

\theoremstyle{definition}
\newtheorem{assumption}{Assumption}
\newtheorem*{assumption*}{Assumption}
\newtheorem{remark}{Remark}
\newtheoremstyle{example}{3pt}{3pt}{}{\parindent}{}{.}{}{}
\newtheorem{example}{Example}

\usepackage[nodisplayskipstretch]{setspace}
\begin{document}
\onehalfspacing
\maketitle

\begin{abstract}
In this paper we develop valid inference for high-dimensional time series. We extend the desparsified lasso to a time series setting under Near-Epoch Dependence (NED) assumptions allowing for non-Gaussian, serially correlated and heteroskedastic processes, where the number of regressors can possibly grow faster than the time dimension. We first derive an error bound 
under weak sparsity, which, coupled with the NED assumption, means this inequality can also be applied to the (inherently misspecified) nodewise regressions performed in the desparsified lasso. This allows us to establish the uniform asymptotic normality of the desparsified lasso under general conditions, including for inference on parameters of increasing dimensions. 
Additionally, we show consistency of a long-run variance estimator, thus providing a complete set of tools for performing inference in high-dimensional linear time series models. Finally, we perform a simulation exercise to demonstrate the small sample properties of the desparsified lasso in common time series settings.
\bigskip
    
    
\end{abstract}

\section{Introduction}
    
In this paper we propose methods for performing uniformly valid inference on high-dimensional time series regression models.
Specifically, we establish the uniform asymptotic normality of the desparsified lasso method \citep{vandeGeer14} under very general conditions, thereby allowing for inference in high-dimensional time series settings that encompass many econometric applications.
That is, we establish validity for potentially misspecified time series models, where the regressors and errors may exhibit serial dependence, heteroskedasticity and fat tails. In addition, as part of our analysis we derive new error bounds for the lasso \citep{Tibshirani96}, on which the desparsified lasso is based.

Although traditionally approaches to high-dimensionality in econometric time series have been dominated by factor models \citep[cf.]{BaiNg08survey,StockWatson11}, shrinkage methods have rapidly been gaining ground. Unlike factor models where dimensionality is reduced by assuming common structures underlying regressors, shrinkage methods assume a certain structure on the parameter  vector. Typically, sparsity is assumed, where only a small, unknown subset of the variables is thought to have ``significantly non-zero'' coefficients, and all the other variables have negligible -- or even exactly zero -- coefficients. The most prominent among shrinkage methods exploiting sparsity is the lasso proposed by \cite{Tibshirani96}, which adds a penalty on the absolute value of the parameters to the least squares objective function. This penalty ensures that many of the coefficients will be set to zero and thus variable selection is performed, an attractive feature that helps to make the results of a high-dimensional analysis interpretable. Due to this feature, the lasso and its many extensions are now standard tools for high-dimensional analysis \citep[see e.g.,][for reviews]{Hesterberg08,Vidaurre13,Hastie15}.

Much effort has been devoted to establish error bounds for lasso-based methods to guarantee consistency for prediction (e.g., \citealp{greenshtein2004persistence, buehlmann2006boosting}) and estimation of a high-dimensional parameter (e.g., \citealp{bunea2007sparsity, zhang2008sparsity, bickel2009simultaneous, meinshausen2009lasso, Huang08}). While most of these advances have been made in frameworks with independent and identically distributed (IID) data, early extensions of lasso-based methods to the time series case can be found in \cite{Wang07}, \cite{Hsu08}. These authors, however, only consider the case where the number of variables is smaller than the sample size. Various papers (e.g., \citealp{NardiRinaldo11}; \citealp{KockCallot2015} and \citealp{BasuMichailidis15}) let the number of variables increase with the sample size, but often require restrictive assumptions (for instance Gaussianity) on the error process when investigating theoretical properties of lasso-based estimators in time series models. 
    
Exceptions are \cite{MedeirosMendes16}, \cite{WuWu2016}, \cite{masini2019regularized}, and \cite{wong2020lasso}. 
\cite{MedeirosMendes16}
consider the adaptive lasso for sparse, high-dimensional time series models and show that it is model selection consistent and has the oracle property, even when the errors are non-Gaussian and conditionally heteroskedastic. 
\cite{WuWu2016} consider high-dimensional linear models with dependent non-Gaussian errors and/or regressors and provide asymptotic theory for the lasso with deterministic design. To this end,
they adopt the functional dependence framework of \cite{Wu05}.
\cite{masini2019regularized} focus on weakly sparse high-dimensional vector autoregressions for a class of potentially heteroskedastic and serially dependent errors, which encompass many multivariate volatility models. The authors derive finite sample estimation error bounds for the parameter vector and establish consistency properties of lasso estimation.
\cite{wong2020lasso} derive nonasymptotic inequalities for estimation error and prediction error of the lasso
without assuming any specific parametric form of the DGP. 
The authors assume the series to be either $\alpha$-mixing Gaussian processes or $\beta$-mixing processes with sub-Weibull marginal distributions thereby accommodating settings with
heavy-tailed non-Gaussian errors.

While one of the attractive feature of lasso-type methods is their ability to perform variable selection, this also causes serious issues when performing inference on the estimated parameters. In particular, performing inference on a (data-driven) selected model, while ignoring the selection, causes the inference to be invalid. This has been discussed by, among others, \cite{LeebPoetscher05} in the general context of model selection and \cite{LeebPoetscher08} for shrinkage estimators. As a consequence, recent statistical literature has seen a surge in the development of so-called \emph{post-selection inference} methods that circumvent the problem induced by model selection; see for example the literature on selective inference \citep[cf.][]{FST15,LSST16} and simultaneous inference \citep{PoSI13,bachoc2020}.

In the context of lasso-type estimation, methods have been developed based on the idea of orthogonalizing the estimation of the parameter of interest to the estimation (and potential incorrect selection) of the other parameters. \cite{BCH14,CHS15ARE} propose a \emph{post-double-selection} approach that uses a Frisch-Waugh partialling out strategy to achieve this orthogonalization by selecting important covariates in initial selection steps on both the dependent variable and the variable of interest, and show this approach yields uniformly valid and standard normal inference for independent data. In a related approach, \cite{JavanmardMontanari14}; \cite{vandeGeer14} and \cite{ZhangZhang14} 
introduce debiased or desparsified versions of the lasso that achieve uniform validity based on similar principles for IID Gaussian data.
Extensions to the time series case include \cite{chernozhukov2021timeandspace} who provide desparsified simultaneous inference on the parameters in a high-dimensional regression model allowing for temporal and cross-sectional dependency in covariates and error processes, 
\cite{Krampe18} who introduce bootstrap-based inference for autoregressive time series models based on the desparsification idea, \cite{HMS19} who use the post-double-selection procedure of \cite{BCH14} for constructing uniformly valid Granger causality test in high-dimensional VAR models, and \cite{ghysels20} who use a debiased sparse group lasso  for inference on a low dimensional group of parameters.

In this paper, we contribute to the literature on shrinkage methods for high-dimensional time series models by providing novel theoretical results for both point estimation and inference via the desparsified lasso.
We consider a very general time series-framework where the regressors and errors terms are allowed to be 
non-Gaussian, serially correlated and heteroskedastic, and the number of variables can grow faster than the time dimension. Moreover, our assumptions allow for both correctly specified and misspecified models, thus providing results relevant for structural interpretations if the overall model is specified correctly, but not limited to this.

We derive error bounds for the lasso in high-dimensional, linear time series models under mixingale assumptions and a weak sparsity assumption on the parameter vector. Our  setting generalizes the one from \cite{MedeirosMendes16}, who require a martingale difference sequence (m.d.s.) assumption -- and hence correct specification -- on the error process. Moreover, we relax the traditional sparsity assumption to allow for weak sparsity, thereby recognizing that the true parameters  are likely not exactly zero.
The error bounds are used to establish estimation and prediction consistency even when the number of parameters grows faster than the sample size. 
    
We extend the error bounds to the \emph{nodewise regressions} performed in the desparsified lasso, where each regressor (on which inference is performed) is regressed on all other regressors. Note that, contrary to the setting with independence over time, these nodewise regressions are inherently misspecified in dynamic models with temporal dependence. As such our error bounds are specifically derived under potential misspecification.
We then establish the asymptotic normality of the desparsified lasso under general conditions.  As such, we ensure uniformly valid inference over the class of weakly sparse models. 
This result is accompanied by a consistent estimator for the long run variance, thereby providing a complete set of tools for performing inference in high-dimensional, linear time series models. As such, our theoretical results accommodate various financial and macro-economic applications encountered by applied researchers.

The remainder of this paper is structured as follows. 
\Cref{sec:Model} introduces the time series setting and assumptions thereof. 
In \Cref{sec:Lasso}, 
we derive an error bound for the lasso (\Cref{cor:separateResults}) that forms the basis for the nodewise regressions performed for the desparsfied lasso. 
In \Cref{sec:desparsifiedLasso}, we establish the theory that allows for uniform inference with the desparsified lasso.
\Cref{sec:simulations} contains a simulation study examining the small sample performance of the desparsified lasso, and  
\Cref{sec:conclusion} concludes. 
The main proofs and preliminary lemmas needed for \Cref{sec:Lasso}
are contained in Appendix \ref{app:A1}, while Appendix \ref{app:A2} contains the results and proofs on \Cref{sec:desparsifiedLasso}. Appendix C contains supplementary material.

A word on notation. For any $N$ dimensional vector $\boldsymbol{x}$, $\left\Vert \boldsymbol{x}\right\Vert_r=\left(\sum\limits_{i=1}^{N}\left\vert x_i\right\vert^r\right)^{1/r}$ 
denotes the $L_r$-norm, with the familiar convention that $\norm{\boldsymbol{x}}_0 = \sum_{i} 1 (\abs{x_i}>0)$ and $\left\Vert \boldsymbol{x}\right\Vert_{\infty}=\max\limits_{i}\left\vert
x_i\right\vert$. For a matrix $\bA$, we let $\norm{\bA}_r = \max_{\norm{\bx}_r = 1} \norm{\bA \bx}_r$ for any $r \in [0, \infty]$ and $\norm{\bA}_{\max}=\max\limits_{i,j}\left\vert a_{i,j}\right\vert$.
We use $\overset{p}{\to}$ and $\overset{d}{\to}$ to denote convergence in probability and distribution respectively. Depending on the context, $\sim$ denotes equivalence in order of magnitude of sequences, or equivalence in distribution. 
We frequently make use of arbitrary positive finite constants $C$ (or its sub-indexed version $C_i$) whose values may change from line to line throughout the paper, but they are always independent of the time and cross-sectional dimension.
Similarly, generic sequences converging to zero as $T\to\infty$ are denoted by  $\eta_T$  (or its sub-indexed version $\eta_{T,i}$).
We say a sequence $\eta_T$ is of size $-x$ if $\eta_T=O\left(T^{-x-\varepsilon}\right)$ for some $\varepsilon>0$.

\section{The High-Dimensional Linear Model}\label{sec:Model} 
Consider the linear model
\begin{equation}\label{eq:DGP}
y_t=\boldsymbol{x}_t'\boldsymbol{\beta}^0+u_t, \qquad t = 1,\ldots, T,\ 
\end{equation}
where $\boldsymbol{x}_t=\left(x_{1,t},\dots, x_{N,t}\right)'$ is a $N\times 1$ vector of explanatory variables, $\boldsymbol\beta^0$ is a $N\times 1$ parameter vector  and $u_t$ is an error term. Throughout the paper, we examine the high-dimensional time series model where $N$ can be larger than $T$. 
    
We impose the following assumptions on the processes $\{\boldsymbol{x}_t\}$ and $\{u_t\}$.

\begin{assumption}\label{ass:dgp}
Let $\boldsymbol{z}_t = (\boldsymbol{x}_t^\prime, u_t)^\prime$,
and let there exist some constants $\bar m>m>2$, and $d\geq \max\{1,(\bar m/m-1)/(\bar m-2)\}$ such that
\begin{enumerate}[label=(\roman*)]
\item\label{ass:dgpStationary}  Let $\E\left[\bz_t\right]=\bzero$, $\E\left[\boldsymbol{x}_t u_t\right]=\boldsymbol{0}$, and $\max\limits_{1\leq j\leq N+1,\ 1\leq t\leq T}E\abs{z_{j,t}}^{2\bar m} \leq C$.
\item\label{ass:dgpNED}Let $\boldsymbol{s}_{T,t}$ denote a $k(T)$-dimensional triangular array that is $\alpha$-mixing of size $-d/(1/m-1/\bar{m})$ with $\sigma\text{-field}$ $\mathcal{F}^{\boldsymbol{s}}_t:=\sigma\left\lbrace\boldsymbol{s}_{T,t},\boldsymbol{s}_{T,t-1},\dots\right\rbrace$ such that $\boldsymbol{z}_t$ is $\mathcal{F}^{\boldsymbol{s}}_t$-measurable. The process $\left\lbrace z_{j,t}\right\rbrace$ is $L_{2m}$-near-epoch-dependent (NED) of size 
$-d$ on $\boldsymbol{s}_{T,t}$ with positive bounded NED constants, uniformly over $j=1,\ldots,N + 1$.
\end{enumerate}
\end{assumption}
\cref{ass:dgp}\ref{ass:dgpStationary} ensures that the error terms are contemporaneously uncorrelated with each of the regressors, and that the process has finite and constant unconditional moments. 
One can think of $\boldsymbol{s}_{T,t}$ in \cref{ass:dgp}\ref{ass:dgpNED} 
as an underlying shock process driving the regressors and errors in $\boldsymbol{z}_t$, where we assume $\boldsymbol{z}_t$ to depend almost entirely on the ``near epoch'' of $s_{T,t}$.\footnote{Since $\boldsymbol{z}_t$ grows asymptotically in dimension, it is natural to let the dimension of $\boldsymbol{s}_{T,t}$ grow with $T$, though this is not theoretically required. Although, like $\boldsymbol{s}_{T,t}$, technically our stochastic process $\boldsymbol{z}_t$ is a triangular array due to dimension $N$ increasing with $T$, in the remainder of the paper we suppress the dependence on $T$ for notational convenience.} 

Near epoch dependence of $\bz_t$ can be interpreted as $\bz_t$ being ``approximately'' mixing, in the sense that it can be well-approximated by a mixing process. The NED framework in \cref{ass:dgp} therefore allows for very general forms of dependence that are often encountered in econometrics applications including, but not limited to,
strong mixing processes \citep{McLeish75}, linear processes including ARMA models, various types of stochastic volatility and GARCH  models \citep{hansen1991garch},
and nonlinear 
processes \citep{davidson2002establishing}. Moreover, NED holds in cases where mixing has well-known failures for common processes, such as the AR(1) process discussed in \cite{andrews1984non}. These properties have made NED a very popular tool for modelling dependence in econometrics \citep[Sections 14, 17]{Davidson02}.\footnote{To make the paper self-contained, we include  formal definitions on NED and mixingales in Appendix \ref{sec:definitions}.}

To our knowledge, our paper is the first to utilize the NED framework for establishing uniformly valid high-dimensional inference.
\cite{wong2020lasso} consider time series models with $\beta$-mixing errors, which has the advantage of allowing for general forms of dynamic misspecification resulting in serially correlated error terms, but, as discussed above, rules out several relevant data generating processes, and is in addition typically difficult to verify. Alternative approaches that avoid mixing assumptions are found in \cite{ghysels20}, who consider $\tau-$dependence, as well as \cite{WuWu2016} and \cite{chernozhukov2021timeandspace}, who use functional dependence for modeling the dependence allowed in regressors and innovations. Finally, \cite{masini2019regularized}
use an m.d.s.~assumption on the innovations in combination
with sub-Weibull tails and a mixingale assumption on the conditional covariance matrix. The m.d.s.~assumption of \cite{MedeirosMendes16} and \cite{masini2019regularized} however does not allow for dynamic misspecification of the full model.
Importantly, the NED assumption on $u_t$ does allow for misspecified models as well, in which case we view $\bbeta_0$ as the coefficients of the pseudo-true model when restricting the class of models to those linear in $\bx_t$. In particular, it allows one to view \eqref{eq:DGP} as simply the linear projection of $y_t$ on \emph{all} the variables in $\boldsymbol{x}_t$, with $\boldsymbol{\beta}^0$ in that case representing the corresponding best linear projection coefficients. In such a case $\E\left[u_t\right]=0$ and $\E\left[u_t  x_{j,t}\right]=0$ hold by construction, and the additional conditions of \cref{ass:dgp} can be shown to hold under weak further assumptions. On the other hand, $u_t$ is not likely to be an m.d.s.~in that case. As will be explained later, allowing for misspecified dynamics is crucial for developing the theory for the nodewise regressions underlying the desparsified lasso. 

It is important to note that we do not consider $\bbeta^0$ as the projection coefficients of the (lasso) selected model, but only of the full, pseudo-true, model. 
Our approach simply allows for the possibility of the full model being misspecified, for instance if the econometrician has missed relevant confounders in the initial dataset. This does not imply a ``failure'' of our lasso inference method, but rather a failure of the econometrician in setting up the initial model.\footnote{Of course, the misspecification may be intentional, as even in dynamically misspecified models, the parameter of interest can still have a structural meaning. One example is the local projections of \cite{jorda2005estimation}, where $h$-step ahead predictive regressions with generally serially correlated error terms are performed.} 
Allowing for such misspecification 
is crucial for the nodewise regressions we consider in Section \ref{sec:desparsifiedLasso} 
which are simply projections of one explanatory variable on all the others, and therefore inherently misspecified.

We further elaborate on misspecification in \Cref{ex:misspecified}, after we present two examples of correctly specified common econometric time series DGPs. 

\begin{remark}
The NED-order $m$ and sequence size $-d$ play a key role in later theorems where they enter the asymptotic rates. 
In \cref{ass:dgp}\ref{ass:dgpStationary},  we require $\bz_t$ to have $\bar m$ moments, with $\bar m$ being slightly larger than $m$. The more moments, the tighter the error bounds and the weaker conditions on the tuning parameter are,
but a high $\bar m$ implies stronger restrictions on the model (see e.g.,\ the GARCH parameters in the to be discussed Example \ref{ex:exampleARDLplusGARCH}).
Additionally, there is a tradeoff between the thickness of the tails allowed for and the amount of dependence -- measured through the mixing rate in \cref{ass:dgp}\ref{ass:dgpNED}.
Under strong dependence, fewer moments are needed;
the reduction from $\bar m$ to $m$ then reflects the price one needs to pay for allowing more dependence through a smaller mixing rate.
\end{remark}

\begin{example}[ARDL model with GARCH errors]\label{ex:exampleARDLplusGARCH}
Consider the autoregressive distributed lag (ARDL) model with GARCH errors 
\begin{equation*}\label{exampleDGP2}\begin{split}
    & y_t=\sum\limits_{i=1}^p \rho_i y_{t-i}+\sum\limits_{i=0}^q\boldsymbol{\theta}_i' \boldsymbol{w}_{t-i}+u_t=\boldsymbol{x}_t'\boldsymbol{\beta}^0+u_t, \\ & u_t=\sqrt{h_t}\varepsilon_t, \qquad \varepsilon_t \sim IID(0,1),\\
    & h_t=\pi_0+\pi_1 h_{t-1}+\pi_2u^2_{t-1},
\end{split}\end{equation*}
where the roots of the lag polynomial $\rho(z) = 1-\sum\limits_{i=1}^{p}\rho_i z^{i}$ are outside the unit circle. Take $\varepsilon_t$, $\pi_1$ and $\pi_2$ such that $\E\left[\ln(\pi_1 \varepsilon_t^2 + \pi_2)\right] <0$, then $u_t$ is a strictly stationary geometrically $\beta$-mixing process \citep[][Theorem 3.4]{FrancqZakoian10}, and additionally such that $\E\left[\abs{u_t}^{2\bar m}\right] < \infty$ for some $\bar m\in \mathds{N}$ (the number of moments depends on $\pi_1$, $\pi_2$ and the moments of $\epsilon_t$, \citealp[cf.][Example 2.3]{FrancqZakoian10}). Also assume that the vector of  exogenous variables $\boldsymbol{w}_t$ is stationary and geometrically $\beta$-mixing as well with finite $2\bar m$ moments. Given the invertibility of the lag polynomial, we may then write $y_t = \rho^{-1} (L) v_t$, where $v_t = \sum_{i=0}^q \boldsymbol{\theta}_i^\prime \boldsymbol{w}_{t-i} + u_t$ and the inverse lag polynomial $\rho^{-1}(z)$ has geometrically decaying coefficients. Then it follows directly that $y_t$ is NED on $v_t$, where $v_t$ is strong mixing of size $-\infty$ as its components are geometrically $\beta$-mixing, and the sum inherits the mixing properties. Furthermore, if $\norm{\theta_i}_1 \leq C$ for all $i=0, \ldots, q$, it follows directly from Minkowski that $E \abs{v_t}^{2\bar m} \leq C$ and consequently $E\abs{y_t}^{2\bar m} \leq C$. Then $y_t$ is NED of size $-\infty$ on $(\boldsymbol{w}_t, u_t)$, and consequently $\boldsymbol{z}_t = (y_{t-1}, \boldsymbol{w}_t, u_t)$ as well.
\end{example}

\begin{example}[Equation-by-equation VAR]\label{ex:eqbyeqVAR}
    Consider the vector autoregressive model
    \begin{equation*}
        \boldsymbol{y}_t=\sum\limits_{i=1}^{p}\boldsymbol{\Phi}_i\boldsymbol{y}_{t-i}+\boldsymbol{u}_t,
    \end{equation*}
    where $\boldsymbol{y}_t$ is a $K\times1$ vector of dependent variables, $\E\abs{u_t}^{2\bar m}\leq C$ , and the $K\times K$ matrices $\boldsymbol{\Phi}_i$ satisfy appropriate stationarity and $2\bar m$-th order summability conditions. 
    The equivalent equation-by-equation representation is
    \begin{equation*}
        y_{k,t}=\sum\limits_{i=1}^p\left[\Phi_{k,1,i},\dots,\Phi_{k,K,i}\right]\boldsymbol{y}_{t-i}+u_{k,t}=\left[\boldsymbol{y}'_{t-1},\dots,\boldsymbol{y}'_{t-p}\right]\boldsymbol{\beta}_k+u_{k,t},\qquad k\in(1,\dots,K).
    \end{equation*}
    Assuming a well-specified model with $\E\left[\boldsymbol{u_t}\vert\boldsymbol{y}_{t-1},\dots,\boldsymbol{y}_{t-p}\right]=\boldsymbol{0}$, the conditions of \cref{ass:dgp} are then satisfied trivially.
\end{example}

\cref{ex:exampleARDLplusGARCH,ex:eqbyeqVAR} demonstrate that \cref{ass:dgp} is sufficiently general to include common time series models in econometrics. While these examples are equally well covered by other commonly used assumptions such as the martingale difference sequence (m.d.s) framework chosen in \cite{MedeirosMendes16} or \cite{masini2019regularized}, we opt for the more general NED framework, as it additionally covers many relevant cases -- in particular for our nodewise regressions -- where properties such as m.d.s.\ fail. The following examples provide simple illustrations of these cases.

\begin{example}[Misspecified AR model]\label{ex:misspecified}
Consider an autoregressive (AR) model of order 2
    \begin{equation*}
        y_t=\rho_1y_{t-1}+\rho_2y_{t-2}+v_t,\qquad v_t\sim IID(0,1),
    \end{equation*}
    where $E\vert v_t\vert^{2\bar m}\leq C$ and the roots of
    $1-\rho_1L-\rho_2L^2$ are outside the unit circle. 
    Define the misspecified model
$y_t=\tilde\rho y_{t-1}+u_t$,
where $\tilde\rho=\argmin\limits_{\rho}\E\left[(y_t-\rho y_{t-1})^2\right]=\frac{\E\left[y_t y_{t-1}\right]}{\E\left[y_{t-1}^2\right]}=\frac{\rho_1}{1-\rho_2}$ and $u_t$ is autocorrelated. An m.d.s.~assumption would be inappropriate in this case, as
\begin{equation*}
\E\left[u_t\vert \sigma\left\lbrace y_{t-1},y_{t-2},\dots\right\rbrace\right]=\E\left[y_t-\tilde\rho y_{t-1}\vert \sigma\left\lbrace y_{t-1},y_{t-2},\dots\right\rbrace\right] = -\frac{\rho_1\rho_2}{1-\rho_2}y_{t-1}+\rho_2y_{t-2}\neq 0. 
\end{equation*}
However, it can be shown that $(y_{t-1}, u_t)'$ satisfies \cref{ass:dgp}\ref{ass:dgpNED} by considering the moving average representation of $y_t$ and by extension, of $u_t=y_{t}-\tilde\rho y_{t-1}$. As the coefficients are geometrically decaying, $u_t$ is clearly NED on $v_t$ and \cref{ass:dgp}\ref{ass:dgpNED} is satisfied.
\end{example}

The key condition to apply the lasso successfully is that the parameter vector $\boldsymbol{\beta}_0$ is (at least approximately) sparse. We formulate this in \Cref{ass:sparsity} below.
\begin{assumption}
\label{ass:sparsity}
For some $0\leq r<1$ and sparsity level $s_r$, define the $N$-dimensional sparse compact parameter space
\begin{equation*}
\boldsymbol{B}_N(r, s_r) :=\left\lbrace \boldsymbol{\beta}\in \mathds{R}^N: \norm{\boldsymbol{\beta}}_r^r \leq s_r, \; \norm{\boldsymbol{\beta}}_{\infty} \leq C, \, \exists C < \infty \right\rbrace ,
\end{equation*}
and assume that $\boldsymbol{\beta}^0\in\boldsymbol{B}_N(r, s_r)$.
\end{assumption}
    
\cref{ass:sparsity} implies that 
${\boldsymbol{\beta}}^0$ is sparse with the degree of sparsity governed by both $r$ and $s_r$. Without further assumptions on $r$ and $s_r$, \cref{ass:sparsity} is not binding, but as will be seen later, the allowed rates will interact with other DGP parameters creating binding conditions. \cref{ass:sparsity} generalizes the common assumption of exact sparsity taking $r=0$ (see e.g., \citealp{MedeirosMendes16}; \citealp{vandeGeer14}; \citealp{chernozhukov2021timeandspace}; \citealp{ghysels20}), which assumes that there are only a few (at most $s_0$) non-zero components in $\boldsymbol\beta^0$, to weak sparsity (see e.g., \citealp{vandeGeer19}). 
This allows us to have many non-zero elements in the parameter vector, as long as they are sufficiently small. 
It follows directly from the formulation in \Cref{ass:sparsity} that, given the compactness of the parameter space, exact sparsity of order $s_0$ implies weak sparsity with $r >0$ of the same order (up to a fixed constant). In general, the smaller $r$ is, the more restrictive the assumption.
The relaxation
to weak sparsity is 
straightforward and follows from elementary inequalities (see e.g., Section 2.10 of \citealp{vandeGeer2016book} and
the proof of \cref{lma:errorBoundonSets}).

\begin{example}[Infinite order AR]\label{ex:infiniteOrderAR}
Consider an infinite order autoregressive model
\begin{equation*}
y_t = \sum_{j=1}^\infty \rho_j y_{t-j} + \varepsilon_t,
\end{equation*}
where $\varepsilon_t$ is a stationary m.d.s.~with sufficient moments existing, and the lag polynomial 
$1 - \sum_{j=1}^\infty\rho_j L^j$
is invertible and satisfies the summability condition $\sum_{j=1}^\infty j^{a} \abs{\rho_j} < \infty$ for some $a\geq 0$. One might consider fitting an autoregressive approximation of order $P$ to $y_t$,
\begin{equation*}
y_t = \sum_{j=1}^P \beta_j y_{t-j} + u_t,
\end{equation*}
as it is well known that if $P$ is sufficiently large, the best linear predictors $\beta_j$ will be close to the true coefficients $\rho_j$ \citep[see e.g.,][Lemma 2.2]{KPP11}. To relate the summability condition above to the weak sparsity condition, note that by H\"older's inequality we have that
\begin{equation*}
\begin{split}
\norm{\boldsymbol{\beta}}_r^r = \sum_{j=1}^P \left(j^a \abs{\beta_j} \right)^r j^{-ar} \leq \left(\sum_{j=1}^P j^a \abs{\beta_j} \right)^r \left( \sum_{j=1}^P j^{-\frac{ar}{1-r}} \right)^{1-r} \leq C \max\{P^{1-(a+1)r}, 1\}.
\end{split}
\end{equation*}
The constant comes from bounding the first term by the convergence of $\beta_j$ to $\rho_j$ plus the summability of the latter, while the second term involving $P$ follows from Lemma 5.1 of \cite{PhillipsSolo92}.\footnote{As the same lemma shows, one should in fact treat the case $r=1/(a+1)$ separately, in which a bound of order $\left(\ln P\right)^{\frac{a}{a+1}}$ holds.} As such, summability conditions on lag polynomials imply weak sparsity conditions, where the strength of the summability condition (measured through $a$) and the required strictness of the sparsity (measured through $r$) determine the order $s_r$ of the sparsity. Therefore, weak sparsity -- unlike exact sparsity -- can accommodate sparse sieve estimation of infinite-order, appropriately summable, processes, providing an alternative to least-squares estimation of lower order approximations. For VAR models we can apply the same reasoning, with the addition that appropriate row sparsity is needed for the coefficients in the row of interest of the VAR if the number of series increases with the sample size.
\end{example}
	
For $\lambda\geq0$, define the weak sparsity index set
\begin{equation} \label{eq:weaksparsity}
S_\lambda:=\left\lbrace j:\abs{\beta_j^0}>\lambda \right\rbrace \quad \text{with cardinality } \vert S_\lambda\vert,
\end{equation}
and complement set $S^c_\lambda=\left\lbrace1,\dots,N\right\rbrace\setminus S_\lambda$. With an appropriate choice of $\lambda$, this set contains all `sufficiently large' coefficients; for $\lambda=0$ it contains all non-zero parameters. We need this set in the following condition, which formulates the standard compatibility conditions needed for lasso consistency \citep[see e.g.,][Chapter 6]{SHD11}.

\begin{assumption}
\label{ass:compatibility}
Let $\boldsymbol\Sigma:=\frac{1}{T}\sum\limits_{t=1}^T\E\left[\boldsymbol{x}_t\boldsymbol{x}'_t\right]$.
For a general index set $S$ with cardinality $\vert S\vert$, define the compatibility constant
\begin{equation*}
	\phi_{\boldsymbol{\Sigma}}^2(S):=\min\limits_{\left\lbrace \boldsymbol{z}\in{\mathds{R}}^{N}\setminus \boldsymbol{0}:\Vert \boldsymbol{z}_{S^c}\Vert_1\leq 3\Vert \boldsymbol{z}_{S}\Vert_1\right\rbrace}\left\lbrace\frac{\vert S\vert \boldsymbol{z}'{\boldsymbol{\Sigma}} \boldsymbol{z}}{\Vert \boldsymbol{z}_{S}\Vert^2_1}\right\rbrace.
	\end{equation*}
	Assume that $\phi_{{\boldsymbol{\Sigma}}}^2(S_\lambda)\geq 1/C$, which implies that
	\begin{equation*}
	\Vert \boldsymbol{z}_{S_\lambda}\Vert^2_1\leq\frac{\vert S_\lambda\vert \boldsymbol{z}'{\boldsymbol{\Sigma}} \boldsymbol{z}}{\phi_{{\boldsymbol{\Sigma}}}^2(S_\lambda)} \leq C\vert S_\lambda\vert \boldsymbol{z}'{\boldsymbol{\Sigma}} \boldsymbol{z},
	\end{equation*}
	for all $\boldsymbol{z}$ satisfying $\Vert \boldsymbol{z}_{S^c_\lambda}\Vert_1\leq 3\Vert \boldsymbol{z}_{S_\lambda}\Vert_1\neq0$.
\end{assumption}

The compatibility constant in \Cref{ass:compatibility} is an upper bound on the minimum eigenvalue of ${\boldsymbol{\Sigma}}$, so this condition is considerably weaker than assuming ${\boldsymbol{\Sigma}}$ to be positive definite. 
We formulate the compatibility condition in \Cref{ass:compatibility} on the population covariance matrix 
rather than directly on the sample covariance matrix $\hat{\boldsymbol{\Sigma}}:=\bX^\prime\bX/T$, see e.g., the restricted eigenvalue condition in \cite{MedeirosMendes16} or Assumption (A2) in \cite{chernozhukov2021timeandspace}. Verifying this assumption on the population covariance matrix is generally more straightforward than directly on the sample covariance matrix.\footnote{Though note that \cite{BasuMichailidis15} 
show in their Proposition 3.1 that the restricted eigenvalue condition holds with high probability under general time series conditions when $\bx_t$ is a stable process with full-rank spectral density and $T$ is sufficiently large. Their Proposition 4.2 includes a stable VAR process as an example.}

Finally, note that the compatibility assumption for the weak sparsity index set $S_\lambda$ is weaker than (and implied by) its equivalent for $S_0$, see Lemma 6.19 in \cite{SHD11}, and that the strictness of this assumption depends on the choice of the tuning parameter $\lambda$. 

\section{Error Bound and Consistency for the Lasso}\label{sec:Lasso}
In this section, we derive a new error bound for the lasso in a high-dimensional time series model.
The lasso estimator \citep{Tibshirani96} of the parameter vector $\boldsymbol\beta^0$ in Model \eqref{eq:DGP} is given by
\begin{equation}\label{eq:LassoDef}
\hat{\boldsymbol{\beta}}:=\argmin_{{\boldsymbol{\beta}}\in\mathds{R}^N} \left\lbrace \frac{\Vert \boldsymbol{y} -\boldsymbol{X}{\boldsymbol{\beta}}\Vert_2^2}{T}+2\lambda\Vert{\boldsymbol{\beta}}\Vert_1 \right\rbrace,
\end{equation}	
where $\boldsymbol{y}=(y_1, \ldots, y_T)^\prime$ is the $T\times 1$ response vector, $\boldsymbol{X}=\left(\boldsymbol{x}_1,\dots,\boldsymbol{x}_T\right)'$ the $T\times N$ design matrix and $\lambda>0$ a tuning parameter. Optimization problem \eqref{eq:LassoDef} adds a penalty term to the least squares objective  to penalize parameters that are different from zero.

When deriving this error bound, one typically requires that $\lambda$ is chosen sufficiently large to exceed the empirical process $\max\limits_{j}\abs{\frac{1}{T}\sum_{t=1}^Tx_{j,t}u_t}$ with high probability. 
To this end, we define the set $\setEP{T}(z):=\left\lbrace\max\limits_{j\leq N,l\leq T}\abs{\sum\limits_{t=1}^{l}u_t x_{j,t}}\leq z\right\rbrace$, and establish the conditions under which $\P\left(\setEP{T}(T\lambda/4)\right)\to 1$. 
In addition, since we formulate the compatibility condition in \Cref{ass:compatibility} on the population covariance matrix, we need to show that $\bSigma$ and $\hat\bSigma$ are sufficiently close under the DGP assumptions. To this end, we define the set $\setCC(S):=\left\lbrace\norm{\hat\bSigma-\bSigma}_{\max}\leq C/\abs{S}\right\rbrace$, and show that
 $\P\left(\setCC(S_{\lambda})\right)\to1$.
\cref{thm:ourContribution} then presents both results.

\begin{theorem}\label{thm:ourContribution} 
Let \cref{ass:dgp,ass:sparsity,ass:compatibility} hold, and assume that
\begin{equation} \label{eq:condition_lambda}
\begin{split}
0<r<1:&\quad\lambda\geq  C\ln(\ln( T))^{\frac{d+m-1}{r(dm+m-1)}}\left[s_r\left(\frac{N^{\left(\frac{2}{d}+\frac{2}{m-1}\right)}}{\sqrt{T}}\right)^{\frac{1}{\left(\frac{1}{d}+\frac{m}{m-1}\right)}}\right]^{\frac{1}{r}}\\
r=0:&\quad s_0\leq C \ln(\ln( T))^{-\frac{d+m-1}{dm+m-1}}\left[\frac{\sqrt{T}}{N^{\left(\frac{2}{d}+\frac{2}{m-1}\right)}}\right]^{\frac{1}{\left(\frac{1}{d}+\frac{m}{m-1}\right)}},\\
&\quad \lambda\geq C{ \ln(\ln( T))}^{1/m}\frac{N^{1/m}}{\sqrt{T}}
\end{split}
\end{equation}
When $N, T$ are sufficiently large, $\P\left(\setEP{T}(T\lambda/4)\cap\setCC(S_\lambda)\right)\geq 1-C \ln(\ln( T))^{-1}$.
\end{theorem}

\cref{thm:ourContribution} thus establishes that
the sets $\setEP{T}(T\lambda/4)$ and $\setCC(S_\lambda)$ hold with high probability.
Each set has a condition under which its probability converges to 1, which follow from \cref{lma:CovarianceCloseness,lma:empiricalProcess} respectively.
For the set $\setEP{T}(T\lambda/4)$, the condition $\lambda\geq{ C\ln(\ln( T))}^{1/m}\frac{N^{1/m}}{\sqrt{T}}$ is required.
The $\ln(\ln( T))$ appearing throughout the theorem 
is chosen arbitrarily as a sequence which grows slowly as $T\to\infty$; we only need some sequence tending to infinity sufficiently slowly. The details can be found in the proof of \cref{thm:ourContribution}. For the set $\setCC(S_\lambda)$, we need to distinguish the cases $0<r<1$ and $r=0$ due to the way the size of the sparsity index set in \cref{eq:weaksparsity} is bounded.
For $0<r<1$, a lower bound on $\lambda$ is imposed which is stricter than the one for the empirical process, hence only that bounds appears in \cref{thm:ourContribution}.
For $r=0$, the conditions do not depend on $\lambda$ hence both bounds appear in \cref{thm:ourContribution}.

\cref{thm:ourContribution} directly yields an error bound for the lasso in high-dimensional time series models by standard arguments in the literature, see e.g.,~Chapter 2 of \cite{vandeGeer2016book}. The proofs of Lemmas A.6 and A.7 in the Supplementary Appendix C.1 provide details.

\begin{corollary}\label{cor:separateResults}
Under \cref{ass:dgp,ass:sparsity,ass:compatibility} and the conditions of \cref{thm:ourContribution}, when $N,T$ are sufficiently large, 
the following
holds with probability at least $1-C\ln \ln T^{-1}$:
\begin{enumerate}[label=(\roman*)]
\item\label{item:prediction} 
$\quad \frac{1}{T} \norm{\boldsymbol{X}(\hat{\boldsymbol{\beta}}-{\boldsymbol{\beta}}^0)}_{2}^2
\leq C\lambda^{2-r}s_r,
$
\item\label{item:estimation} 
$\quad\norm{\hat{\boldsymbol{\beta}}-{\boldsymbol{\beta}}^0}_1
\leq C\lambda^{1-r}s_r.
$
\end{enumerate}
\end{corollary}

Under the additional assumption that $\lambda^{1-r}s_{r}\to0$, these error bounds directly establish prediction and estimation consistency. 
The bounds in \cref{thm:ourContribution} thereby put implicit limits on the divergence rate of $N$, and $s_{r}$ relative to $T$. In particular, the term offsetting the divergence in $N$, and $s_{r}$ is of polynomial order in $T$. The order of the polynomial, and therefore the restriction on the growth of $N$ and $s_{r}$, is determined by the moments $m$ and dependence parameter $d$; the higher the number of moments $m$ and the larger the dependence parameter $d$, the fewer restrictions one has on the allowed polynomial growth of $N$ and $s_r$. In the limit, if $m$ and $d$ tend to infinity (all moments exist and the data are mixing), the order of the polynomial restriction on $N$ tends to infinity, thereby approaching exponential growth. A similar trade off between the allowed growth of $N$ and the existence of moments was found in \cite{MedeirosMendes16}. In Example C.1 we study in greater detail how the different rates interact, thereby providing an overview of the restrictions under different scenarios.

While \cref{cor:separateResults}
is a useful result in its own right, it is vital to derive the theoretical results for  the desparsified lasso, which we turn to next. 

\section{Uniformly Valid Inference via the Desparsified Lasso}\label{sec:desparsifiedLasso}
We use the desparsified lasso to perform uniformly valid inference in general high-dimensional time series settings.
After briefly reviewing the desparsified lasso, we formulate the assumptions needed in \Cref{assumptions_DL}. The asymptotic theory is then derived in \Cref{Inference_DL} for inference on low-dimensional parameters of interest, and \Cref{Inference_HD} for inference on a high-dimensional parameters. 

The desparsified lasso \citep{vandeGeer14} is defined as 
\begin{equation}\label{eq:bhatDef}
	\hat{\boldsymbol{b}}:=\hat{{\boldsymbol{\beta}}}+\frac{\hat{{\boldsymbol{\Theta}}}\boldsymbol{X}'({\boldsymbol{y}}-\boldsymbol{X}\hat{{\boldsymbol{\beta}}})}{T},
\end{equation}	
where $\hat{\boldsymbol{\beta}}$ is the lasso estimator from \cref{eq:LassoDef} and $\hat{{\boldsymbol{\Theta}}}:=\hat{\boldsymbol{\Upsilon}}^{-2}\hat{\boldsymbol{\Gamma}}$ is a reasonable approximation for the inverse of 
$\hat\bSigma$. 
By de-sparsifying the initial lasso, the bias in the lasso estimator is removed and uniformly valid inference can be obtained.
The matrix $\hat{\boldsymbol{\Gamma}}$ is constructed using nodewise regressions; regressing each column of $\boldsymbol{X}$ on all other explanatory variables using the lasso. Let the lasso estimates of the $j=1,\dots,N$ nodewise regressions be
	\begin{equation}\label{eq:NodewiseLassoDef}
	\hat{\boldsymbol{\gamma}}_j:=\argmin_{\boldsymbol{\gamma}_j\in\mathds{R}^{N-1}} \left\lbrace \frac{\Vert \boldsymbol{x}_j-\boldsymbol{X}_{-j}\boldsymbol{\gamma}_j\Vert_2^2}{T}+2\lambda_j\Vert\boldsymbol{\gamma}_j\Vert_1 \right\rbrace,
	\end{equation}
where the $T\times(N-1)$ matrix $\boldsymbol{X}_{-j}$ is $\boldsymbol{X}$ with its $j$th column removed. Their components are given by $\hat{\boldsymbol{\gamma}}_j=\left\lbrace\hat\gamma_{j,k}:k=\{1,\dots,N\}\setminus j\right\rbrace$. 
Stacking these estimated parameter vectors row-wise with ones on the diagonal gives the matrix
\begin{equation*}\label{ChatDef} 
\hat{\boldsymbol{\Gamma}}:=\begin{bmatrix}
1      & -\hat{\gamma}_{1,2} & \dots & -\hat{\gamma}_{1,N} \\
-\hat{\gamma}_{2,1}      & 1 & \dots & -\hat{\gamma}_{2,N} \\
\vdots & \vdots &\ddots & \vdots\\
-\hat{\gamma}_{N,1}      & -\hat{\gamma}_{N,2}  & \dots & 1
\end{bmatrix}.
\end{equation*}
We then take $\hat{\boldsymbol{\Upsilon}}^{-2}:=\text{diag}\left(1/\hat{\tau}_1^2,\dots,1/\hat{\tau}_N^{2}\right)$, where
$\hat{\tau}_j^2:=\frac{1}{T} \norm{ \boldsymbol{x}_j-\boldsymbol{X}_{-j} \hat{\boldsymbol{\gamma}}_j}_2^2 + 2\lambda_j \norm{\hat{\boldsymbol{\gamma}}_j}_1$.
	
We use the index set $H\subseteq \left\lbrace 1,\dots,N\right\rbrace$ with cardinality $h=\abs{H}$ to denote the set of variables whose coefficients we wish to perform inference on. In this case computational gains can be obtained with respect to the nodewise regressions, as we only need to obtain the sub-vector of the  desparsified lasso  corresponding to
$\hat{\boldsymbol{b}}_H:=\hat{\boldsymbol{\beta}}_H + \hat{\boldsymbol{\Theta}}_H \boldsymbol{X}({\boldsymbol{y}}-\boldsymbol{X}\hat{{\boldsymbol{\beta}}})$,
with the subscript $H$ indicating that we only take the respective rows of $\hat{\boldsymbol{\beta}}$ and $\hat{\boldsymbol{\Theta}}$. To compute $ \hat{\boldsymbol{\Theta}}_H$, one only needs to compute  $h$ nodewise regressions instead of $N$, which can be a considerable reduction for small $h$ relative to large $N$.

\subsection{Assumptions}\label{assumptions_DL}
Consider the population nodewise regressions defined by the linear projections 
\begin{equation}\label{eq:populationGammaj}
	x_{j,t}=\boldsymbol{x}'_{-j,t}{\boldsymbol{\gamma}}^0_j+v_{j,t} \qquad \boldsymbol{\gamma}^0_j:= \argmin_{\boldsymbol{\gamma}} \left\lbrace \E\left[\frac{1}{T}\sum\limits_{t=1}^T\left(x_{j,t}-\boldsymbol{x}_{-j,t}'\boldsymbol{\gamma}\right)^2\right]\right\rbrace, \qquad j=1,\ldots, N,
\end{equation}	
with $\tau_j^2:=\frac{1}{T}\sum\limits_{t=1}^T\E\left[v_{j,t}^2\right]$. 
Note that by construction, it holds that $\E\left[v_{j,t}\right]=0,\ \forall t, j$ and $\E\left[v_{j,t}x_{k,t}\right]=0,\ \forall t, k\neq j$. 
We first present 
\cref{ass:statandvmoments,ass:nodewise},
which allow us to extend \cref{cor:separateResults} to the nodewise lasso regressions. 

\begin{assumption}\label{ass:statandvmoments}
Let $\max\limits_{1\leq j\leq N,\ 1\leq t\leq T} \E \abs{v_{j,t}}^{2\bar m} \leq C$.
\end{assumption}

\begin{assumption}\label{ass:nodewise}\hfill
\begin{enumerate}[label=(\roman*)]
\item\label{ass:nodewiseSparsity} 
For some $0\leq r<1$ and sparsity levels $s_{r}^{(j)}$, let $\gamma_j^0\in\boldsymbol{B}_{N-1}(r,s_{r}^{(j)})$, $\forall j\in H$.
\item\label{ass:nodewiseCompatibility} Let $\max\limits_{1\leq j\leq N}\sigma_{j,j}\leq C$ and $\Lambda_{\min} \geq 1/C$, where $\Lambda_{\min}$ is the smallest eigenvalue of $\boldsymbol{\Sigma}$.
\end{enumerate}
\end{assumption}

\cref{ass:statandvmoments}  requires the errors $v_{j,t}$ from the nodewise linear projections to have bounded moments of an order greater than fourth.
By the properties of NED processes, we use \cref{ass:dgp,ass:statandvmoments} to establish mixingale properties of the products $v_{j,t}u_t=:w_{j,t}$ and $w_{j,t}w_{k,t-l}$ in \cref{ass:CLT}, which are used extensively in the derivation of the desparsified lasso's asymptotic distribution. 

\cref{ass:nodewise}\ref{ass:nodewiseSparsity}, similar to \cref{ass:sparsity}, requires weak sparsity of the nodewise regressions, not exact sparsity. 
The latter could be problematic, as it would imply many of the regressors to be uncorrelated. In contrast, weak sparsity is a plausible alternative, see e.g.,\ \cref{ex:infiniteOrderAR}.
Importantly, the weak sparsity of the nodewise regressions is fully determined by the model and hence should be verified. Below, we provide concrete examples 
where the weak sparsity assumption holds.

\cref{ass:nodewise}\ref{ass:nodewiseCompatibility} requires the population covariance matrix to be positive definite, with its smallest eigenvalue bounded away from zero, and to have finite variances. \cref{ass:nodewise}\ref{ass:nodewiseCompatibility} implies the compatibility condition and thus replaces \Cref{ass:compatibility} in \Cref{sec:Lasso}, with $\Lambda_{min}$ fulfilling the role of $\phi_{{\boldsymbol{\Sigma}}}^2$. It also implies that the explanatory variables, including the irrelevant ones, cannot be linear combinations of each other even as we let the number of variables tends to infinity. Although this is a considerable strengthening of \cref{ass:compatibility}, it is important to realize this assumption is still made on the population matrix instead of the sample version, and may therefore still hold in fairly general, high-dimensional models. For example, \cite{BasuMichailidis15} provide a lower bound for $\Lambda_{\min}$ in VAR models on their Proposition 2.3, which can be shown to be bounded away from zero under realistic conditions, see also
\citeauthor{masini2019regularized} (\citeyear{masini2019regularized}, p.\ 6). 
Similarly, this assumption can be shown to hold in factor models under minimal assumptions on the idiosyncratic errors (see \Cref{ex:SparseFactor} below).

\begin{example}(Sparse factor model) \label{ex:SparseFactor}
Consider the factor model
\begin{equation*}\begin{split}
y_t&=\boldsymbol{\beta^0}'\boldsymbol{x}_{t}+u_t,\ u_t\sim IID(0,1)\\
\boldsymbol{x}_t&=\underset{N\times k}{\bLambda}\underset{k\times 1}{\bbf_t}+\boldsymbol{\nu}_t,\ \boldsymbol{\nu}_t \sim IID(\boldsymbol{0},\bSigma_{\bnu}),\qquad
\bbf_t \sim IID(\bzero,\bSigma_{\bbf}),
\end{split}\end{equation*}
where $\bLambda$ has bounded elements, $\bSigma_{\bbf}$ and $\bSigma_{\bnu}$ are positive definite with bounded eigenvalues, and $\bnu_t$ and $\bbf_t$ are uncorrelated. In this DGP,
\begin{equation*}
   \bSigma=\bLambda\bSigma_{\bbf}\bLambda^\prime+\bSigma_{\bnu}\Longrightarrow\bTheta=\bSigma_{\bnu}^{-1}-\bSigma_{\bnu}^{-1}\bLambda\left(\bSigma_{\bbf}^{-1}+\bLambda^\prime\bSigma_{\bnu}^{-1}\bLambda\right)^{-1}\bLambda^{\prime}\bSigma_{\bnu}^{-1}.
\end{equation*}
As shown in Supplementary Appendix C.4, the sparsity of the nodewise regression parameters can be bounded as
\begin{equation*}
\max\limits_{j}\norm{\bgamma_j^0}_r^r\leq \norm{\bSigma_{\bnu}^{-1}}_r^r\left(1+C \norm{\bSigma_{\bnu}^{-1}}_r^r \norm{\bLambda}_r^r k^{2-r/2} N^{-ar} \right),
\end{equation*}
where $N^a$ is the rate at which the $k$-th largest eigenvalue of $\bSigma$ diverges. This result allows for weak factor models where $a<1$, which have been proposed for providing a theoretical explanation for the often observed empirical phenomenon where the separation between the eigenvalues of the Gram matrix is not as large as the strong factor model with $a=1$ implies  \citep[cf.][]{DeMol2008forecasting,Onatski2012asymptotics,uematsu2022estimation,uematsu2022inference}.

The bound of the nodewise regressions further depends on the number of factors, the sparsity of the factor loadings and the sparsity of $\bSigma_{\bnu}^{-1}$. Sparse factor loadings are intimately linked to weak factor models, and may provide accurate descriptions of the data in various economic and financial applications, see  \citet{uematsu2022estimation,uematsu2022inference} and Supplementary Appendix C.4 for details.

Sparsity in $\bSigma_{\bnu}^{-1}$ holds when the idiosyncratic components are not too strongly cross-sectionally dependent, which is a standard assumption in factor models. It occurs for instance for
block diagonal structures of $\bSigma_{\bnu}$, in which case
$\norm{\bSigma_{\bnu}^{-1}}_r^r\leq Cb$  where $b$ is the size of the largest $b\times b$ block matrix  with $b^2$ nonzero elements, or for Toeplitz structures ${\sigma_{\bnu}}_{i,j}=\rho^{\abs{i-j}}, \abs{\rho}<1$, in which case  $\norm{\bSigma_{\bnu}^{-1}}_r^r\leq C$. 
Note that to satisfy the minimum eigenvalue condition  (\Cref{ass:nodewise}\ref{ass:nodewiseCompatibility}), we only need the minimum eigenvalue of $\bSigma_{\bnu}$ to be bounded away from 0.
\end{example}

\begin{example}[Sparse VAR(1)]\label{ex:sparseVAR1}
    Consider a stationary VAR(1) model for $\bz_t=(y_t,\bx_t^\prime)^\prime$
    \begin{equation*}
        \bz_t=\bPhi \bz_{t-1}+\bu_t,\ \E\bu_t\bu_t^\prime:=\bOmega,\ \E\bu_t\bu_{t-l}^\prime=\bzero,\ \forall l\neq0, 
    \end{equation*}
    with our regression of interest being the first line of the VAR, that is $y_t=\bphi_1\bz_{t-1}+u_{1,t}$,
    where $\bphi_j$ is the $j$th row of $\bPhi$. Under this DGP, the nodewise regression parameters $\bgamma_j^0$ are determined entirely by $\bPhi$ and $\bOmega$, and we now consider two cases for which we derive explicit results in Supplementary Appendix C.4.
\begin{enumerate}
    \item[(a)]  Let $\bPhi$ be symmetric and block diagonal with largest block of size $b$. Assume that $\bPhi$ has eigenvalues strictly between 0 and 1, and $\norm{\bPhi}_{\max}\leq C$.  Furthermore, let $\bOmega=\bI$. 
    Then the nonzero entries of $\bgamma_j^0$ follow the block structure of $\bPhi$, such that $\max\limits_{j}\norm{\bgamma^0_j}_{0}\leq Cb$.
    \item[(b)] Let $\bPhi=\phi \bI$ with $\abs{\phi}<1$, and let $\bOmega$ have a Toeplitz structure $\omega_{i,j}=\rho^{\abs{i-j}},\ \abs{\rho}<1$. 
    Then $\bgamma_j^0$ is only weakly sparse, in the sense that it contains no zeroes, but its entries follow a geometrically decaying pattern, meaning that
    $\max\limits_{j}\norm{\bgamma_j^0}_r^r\leq C$. 
\end{enumerate}
More generally, sparsity of $\bgamma_j^0$ requires that the autoregressive coefficient matrix $\bPhi$ and the error covariance matrix $\bOmega$ are row- and column-sparse in such a way that matrix multiplication preserves this sparsity. For case (a), we may  relax the assumption on $\bOmega$ to block-diagonality, provided the block structure is similar to that of $\bPhi$. For case (b), the result holds even when we let $\bPhi$ have a similar Toeplitz structure as $\bOmega$, as we numerically investigate in Supplementary Appendix C.4. To verify the minimum eigenvalue condition in \Cref{ass:nodewise}\ref{ass:nodewiseCompatibility}, we may apply the bound derived in \cite[p.~6]{masini2019regularized}, which gives $\bLambda_{\min}\geq \bLambda_{\min}(\bOmega) \left[1+\left(\norm{\bPhi}_1+\norm{\bPhi}_{\infty}\right)/2\right]^2$, where $\bLambda_{\min}(\bOmega)$ is the smallest eigenvalue of $\bOmega$.
\end{example}
\begin{remark}
Alternative approaches exist that circumvent the need to directly impose weak sparsity assumptions on the nodewise regressions. \cite{Krampe18} use the desparsified lasso for inference in the context of stationary VARs with IID errors, but do not use nodewise regressions to build an estimator of $\bTheta$ as we do. Instead, they use the VAR model structure to derive an estimator based on regularized estimates of the VAR coefficients and the error covariances. Such an approach requires knowledge of the full model underlying the covariates to provide an analytical expression for the nodewise projections. While this is a natural approach in a VAR model, this approach is considerably more difficult to apply in a more general setting, where the structure underlying the covariates is typically unknown. Moreover, they still require conditions on sparsity, which are similar to those found for the VAR model of \cref{ex:sparseVAR1}, i.e.~row- and column-sparsity of the VAR coefficient matrices in addition to sparsity of the inverse error covariance matrix.

\cite{deshpande2020online} use an online debiasing strategy for inference in VAR models with IID Gaussian errors, among other settings. Rather than using a single estimate of $\bTheta$, they use a sequence of precision matrix estimates based on an episodic structure, which can be seen as a generalization of sample-splitting. In addition, they use the precision matrix estimator as in \cite{JavanmardMontanari14}, which does not require sparsity of $\bTheta$. It is an interesting topic for future research to investigate whether these techniques can be leveraged in our setting allowing for misspecification and with potentially serially correlated/heteroskedastic errors.
\end{remark}

\cref{ass:statandvmoments,ass:nodewise} allow us to apply \cref{cor:separateResults} to the nodewise regressions. Specifically, if the conditions on $\lambda$ formulated in \eqref{eq:condition_lambda} hold for both $\underset{\bar{}}{\lambda}:=\min\limits_{j\in H}\lambda_j$ and $\bar{\lambda}:=\max\limits_{j\in H}\lambda_j$, the error bounds -- with $\bar{s}_r:=\max\limits_{j\in H} s^{(j)}_{r}$ substituted for $s_r$ -- apply to the nodewise regressions as well. As we generally need the error bounds to hold uniformly over all relevant nodewise regressions as well as the initial regression, we combine these bounds and state our results on the quantities
\begin{equation} \label{eq:s_l_max}
\lambda_{\min}=  \min\{\lambda, \underset{\bar{}}{\lambda}\}, \qquad \lambda_{\max}=  \max\{\lambda, \bar{\lambda}\}, \qquad  s_{r,\max} =  \max\{s_r, \bar{s}_r\},
\end{equation}
which simplifies many of the final expressions. While some conditions could be weakened if we keep them in terms of $\bar{\lambda}$ or $\bar{s}_r$ explicitly, this would be at the expense of more conditions and readability, and therefore we opt against it.

\subsection{Inference on low-dimensional parameters
} \label{Inference_DL}
In this section we establish the uniform asymptotic normality of the desparsified lasso focusing on low-dimensional parameters of interest. 
We consider testing $P$ joint hypotheses of the form $\boldsymbol{R}_{N}\boldsymbol{\beta}^0=\boldsymbol{q}$ via a Wald statistic, where $\boldsymbol{R}_{N}$ is an appropriate $P\times N$ matrix whose non-zero columns are indexed by the set $H:=\left\lbrace j:\sum_{p=1}^P\vert r_{N,p,j}\vert>0\right\rbrace$ of cardinality $h:=\vert H\vert$. As can be seen from the lemmas in Appendix \ref{app:A2}, all our results up to application of the central limit theorem allow for $h$ to increase in $N$ (and therefore $T$). 
In \Cref{thm:CLT} we first focus on inference on a finite set of parameters, such that we can apply a standard central limit theorem under the assumptions listed above. An alternative, high-dimensional approach under more stringent conditions is considered in \Cref{Inference_HD}.

Given our time series setting, the long-run covariance matrix
\begin{equation*}
    {\boldsymbol{\Omega}}_{N,T} = \E\left[\frac{1}{T}\left(\sum\limits_{t=1}^T \boldsymbol{w}_t\right)\left(\sum\limits_{t=1}^T \boldsymbol{w}'_t\right)\right],
\end{equation*} where $\boldsymbol{w}_t=(v_{1,t}u_t,\dots,v_{N,t}u_t)'$, enters the asymptotic distribution in \Cref{thm:CLT}.  $\boldsymbol{\Omega}_{N,T}$ can equivalently be written as $\boldsymbol{\Omega}_{N,T}=\boldsymbol{\Xi}(0)+\sum\limits_{l=1}^{T-1}(\boldsymbol{\Xi}(l)+\boldsymbol{\Xi}^\prime(l))$, where $\boldsymbol{\Xi}(l) = \frac{1}{T}\sum\limits_{t=l+1}^T\E\left[ \boldsymbol{w}_t \boldsymbol{w}_{t-l}^\prime\right]$.

\begin{theorem}\label{thm:CLT}
Let \cref{ass:dgp,ass:sparsity,ass:compatibility,ass:statandvmoments,ass:nodewise} hold, and assume that the smallest eigenvalue of $\boldsymbol{\Omega}_{N,T}$ is bounded away from 0. 
Furthermore, 
assume that $\lambda_{\max}^2\leq(\ln\ln T) \lambda_{\min}^{r}\left[\sqrt{T}s_{r,\max}\right]^{-1}$, 
and
\begin{equation*}
\begin{split}
0<r<1:&\quad\lambda_{\min}\geq  {(\ln \ln T)}\left[s_{r,\max}\left(\frac{N^{\left(\frac{2}{d}+\frac{2}{m-1}\right)}}{\sqrt{T}}\right)^{\frac{1}{\left(\frac{1}{d}+\frac{m}{m-1}\right)}}\right]^{\frac{1}{r}}\\
r=0:&\quad s_{0,\max}\leq {(\ln \ln T)^{-1}}\left[\frac{\sqrt{T}}{N^{\left(\frac{2}{d}+\frac{2}{m-1}\right)}}\right]^{\frac{1}{\left(\frac{1}{d}+\frac{m}{m-1}\right)}},\quad \lambda_{\min}\geq {(\ln \ln T)}\frac{N^{1/m}}{\sqrt{T}}.
\end{split}
\end{equation*}

Let $\boldsymbol{R}_N\in\mathds{R}^{P\times N}$ satisfy $\max\limits_{1\leq p\leq  P} \norm{\boldsymbol{r}_{N,p}}_1\leq C$,
where $\boldsymbol{r}_{N,p}$ denotes the $p$-th row of $\boldsymbol{R}_N$, and $P, h\leq C$. Then we have that 
\begin{equation*}
\sqrt{T}\boldsymbol{R}_{N}(\hat{{\boldsymbol{b}}}-{\boldsymbol{\beta}}^0)\overset{d}{\to}N\left(\boldsymbol{0},\boldsymbol{\Psi}\right),
\end{equation*}
uniformly in ${\boldsymbol{\beta}}^0\in\boldsymbol{B}_N(r, s_r)$, where
\begin{equation*}
\boldsymbol{\Psi}:=\lim\limits_{N,T\to\infty}\boldsymbol{R}_{N}{\boldsymbol{\Upsilon}}^{-2}{\boldsymbol{\Omega}_{N,T}}{\boldsymbol{\Upsilon}}^{-2}{\boldsymbol{R}'_{N}} \text{ and } \boldsymbol{\Upsilon}^{-2}:=\text{diag}(1/\tau^2_1,\dots,1/\tau^2_N).
\end{equation*}  
\end{theorem}

\begin{remark}
Unlike  \cite{vandeGeer14}, we do not require the regularization parameters $\lambda_j$ to have a uniform growth rate. We only control the slowest and fastest converging $\lambda_j$ (covered by $\lambda_{\max}$ and $\lambda_{\min}$ respectively) through convergence rates that also involve $N,T$, and the sparsity $s_{r,\max}$. We provide a specific example of a joint asymptotic setup for these quantities in \cref{cor:usefulResult}. 
\end{remark}

\begin{remark}\label{rem:SampleSplitting}
\cite{BCCH12} and \cite{chernozhukov2018double}, among others,  show that sample splitting can improve the convergence rates for the desparsified lasso in IID settings. 
The idea is to estimate the initial and nodewise regressions with two independent parts of the sample, and exploit this independence to efficiently bound certain 
terms in the proofs. Efficiency loss is then avoided by so-called cross-fitting and  combining two estimators in which the roles of the two sub-samples are swapped. 
However, with time series data naive sample splitting will not yield (asymptotically) independent subsamples. Instead, subsamples must carefully be chosen to leave sufficiently large `gaps' in-between to ensure (at least asymptotic) independence. These ideas are explored
in \cite{lunde2019sample} and \cite{BHS21}, though for different purposes and dependence concepts. They could however provide a useful starting point for future research on investigating the potential of sample-splitting in the NED framework. 
\end{remark}

In order to estimate the asymptotic variance $\boldsymbol{\Psi}$, we suggest to estimate $\boldsymbol{\Omega}_{N,T}$ with the long-run variance kernel estimator
\begin{equation}\label{eq:LRC}
\hat{\boldsymbol{\Omega}}=\hat{\boldsymbol{\Xi}}(0)+\sum\limits_{l=1}^{Q_T-1} K\left(\frac{l}{Q_T}\right) \left(\hat{\boldsymbol{\Xi}}(l)+\hat{\boldsymbol{\Xi}}^\prime(l)\right),
\end{equation}
where $\hat{\boldsymbol{\Xi}}(l)=\frac{1}{T-l}\sum\limits_{t=l+1}^{T}\hat{\boldsymbol{w}}_t \hat{\boldsymbol{w}}_{t-l}^\prime$ with 
$\hat{w}_{j,t}=\hat{v}_{j,t}\hat{u}_t$, the kernel $K(\cdot)$ can be taken as the Bartlett kernel $K(l/Q_T) =  \left(1-\frac{l}{Q_T}\right)$ \citep{NeweyWest87} and the bandwidth $Q_T$ should increase with the sample size at an appropriate rate. A similar heteroskedasticity and autocorrelation consistent (HAC) estimator was considered by \cite{ghysels20}, though under a different framework of dependence. In \cref{thm:LRVconsistency}, we show that $\hat{\boldsymbol{\Psi}} = \boldsymbol{R}_{N}({\hat{\boldsymbol{\Upsilon}}}^{-2}\hat{\boldsymbol{\Omega}}{\hat{\boldsymbol{\Upsilon}}}^{-2}){\boldsymbol{R}'_{N}}$ is a consistent estimator of $\boldsymbol{\Psi}$ in our NED framework.
\begin{theorem}\label{thm:LRVconsistency} 
Take $\hat{\boldsymbol{\Omega}}$ with  $Q_T\to\infty$ as $T\to\infty$, such that
$Q_Th^2(\sqrt{T}h^2)^{-\frac{1}{1/d+m/(m-2)}}\to0$. Assume that 
\begin{equation*}\begin{split}
&\lambda_{\max}^{2-r}\leq (\ln \ln T)^{-1} \min\left\lbrace\left[\sqrt{Q_T}\sqrt{T}s_{r,\max}\right]^{-1}\right.,\left[Q_Th^{1/m}T^{1/m}s_{r,\max}\right]^{-1},\\
&\qquad\qquad\qquad\qquad\quad\quad\left.\left[Q_T^2h^{3/m}T^{(3-m)/m}s_{r,\max}\right]^{-1},\left[Q_T^{2/3}h^{1/(3m)}T^{(m+1)/3m}s_{r,\max}\right]^{-1}\right\rbrace,\\
&\lambda_{\max}^2\leq (\ln \ln T)^{-1} \lambda_{\min}^{r}\left[\sqrt{T}h^{2/m}s_{r,\max}\right]^{-1},\text{ and }\\
\end{split}\end{equation*}
\begin{equation*}
\begin{split}
0<r<1:&\quad\lambda_{\min}\geq  (\ln \ln T)\left[s_{r,\max}\left(\frac{(hN)^{\left(\frac{2}{d}+\frac{2}{m-1}\right)}}{\sqrt{T}}\right)^{\frac{1}{\left(\frac{1}{d}+\frac{m}{m-1}\right)}}\right]^{\frac{1}{r}},\\
r=0:&\quad s_{0,\max}\leq (\ln \ln T)^{-1}\left[\frac{\sqrt{T}}{(hN)^{\left(\frac{2}{d}+\frac{2}{m-1}\right)}}\right]^{\frac{1}{\left(\frac{1}{d}+\frac{m}{m-1}\right)}},\quad \lambda_{\min}\geq (\ln \ln T)\frac{(hN)^{1/m}}{\sqrt{T}}.
\end{split}
\end{equation*}

Furthermore, let $\boldsymbol{R}_N\in\mathds{R}^{P\times N}$ satisfy $\max\limits_{1\leq p\leq P}\norm{\boldsymbol{r}_{N,p}}_1\leq C$ and $P\leq Ch$. Then under \Cref{ass:dgp,ass:sparsity,ass:compatibility,ass:statandvmoments,ass:nodewise}, uniformly in ${\boldsymbol{\beta}}^0\in\boldsymbol{B}_N(r,s_r)$,
\begin{equation*}
\norm{\boldsymbol{R}_{N}({\hat{\boldsymbol{\Upsilon}}}^{-2}\hat{\boldsymbol{\Omega}}{\hat{\boldsymbol{\Upsilon}}}^{-2}-\bUpsilon^{-2}\bOmega_{N,T}\bUpsilon^{-2}){\boldsymbol{R}'_{N}}}_{\max}\overset{p}{\to}0.
\end{equation*}
\end{theorem}

Note that here we restrict $\boldsymbol{R}_N$ such that the number of hypotheses $P$ may not grow faster than the number of parameters of interest $h$, but $h$ may grow with $T$ at a controlled rate. \cref{thm:LRVconsistency} therefore allows for variance estimation of an increasing number of estimators. We believe the restrictions on $P$ are reasonable, as they apply to the most commonly performed hypothesis tests in practice, such as joint significance tests (where $\boldsymbol{R}_N$ is the identity matrix), or tests for the equality of parameter pairs.

As a natural implication of \cref{thm:CLT,thm:LRVconsistency}, \cref{cor:usefulResult} gives an asymptotic distribution result for a quantity composed exclusively of estimated components.
\begin{corollary}\label{cor:usefulResult}
Let \Cref{ass:dgp,ass:sparsity,ass:compatibility,ass:statandvmoments,ass:nodewise} hold, and assume that the smallest eigenvalue of $\boldsymbol{\Omega}_{N,T}$ is bounded away from 0, and $Q_TT^{-\frac{1}{2/d+2m/(m-2)}}\to0$ for some $Q_T\to\infty$.
Further, assume that $\lambda\sim\lambda_{\max}\sim\lambda_{\min}$, and
\begin{equation*}
\begin{split}
0<r<1:&\quad (\ln \ln T)^{-1} s_{r,\max}^{1/r}\left[\frac{N^{\left(\frac{2}{d}+\frac{2}{m-1}\right)}}{\sqrt{T}}\right]^{\frac{1}{r\left(\frac{1}{d}+\frac{m}{m-1}\right)}}\leq\lambda\leq\ \ln \ln T \left[Q_T^2\sqrt{T}s_{r,\max}\right]^{-1/(2-r)},\\
r=0:&\quad\ (\ln \ln T)^{-1} \frac{N^{1/m}}{\sqrt{T}}\leq\lambda\leq \ln \ln T \left[Q_T^2\sqrt{T}s_{0,\max}\right]^{-1/2}.
\end{split}
\end{equation*}
These bounds are feasible when $Q_T^rs_{r,\max}N^{\left(2-r\right)\left(\frac{d+m-1}{dm+m-1}\right)}T^{\frac{1}{4}\left(r-\frac{d(m-1)(2-r)}{dm+m-1}\right)}\to0$, and additionally when $Q_T^2s_{0,\max}\frac{N^{2/m}}{\sqrt{T}}\to0$  if $r=0$.
Under these conditions, for ${\boldsymbol{R}_{N}}\in\mathds{R}^{P\times N}$ with $\max\limits_{1\leq p\leq P}\norm{\boldsymbol{r}_{N,p}}_1\leq C$ and $P, h\leq C$,
we have that
\begin{align}
&\sup_{\underset{1\leq p\leq P, z\in\mathds{R}} {\boldsymbol{\beta}^0 \in \boldsymbol{B}_N(r,s_r)}} \left\vert\P\left(\sqrt{T}\frac{\boldsymbol{r}_{N,p}(\hat{{\boldsymbol{b}}}-{\boldsymbol{\beta}}^{0})}{\sqrt{\boldsymbol{r}_{N,p}(\hat{{\boldsymbol{\Upsilon}}}^{-2}\hat{\boldsymbol{\Omega}}\hat{{\boldsymbol{\Upsilon}}}^{-2}){\boldsymbol{r}'_{N,p}}}}\leq z\right)-\boldsymbol{\Phi}(z)\right\vert=o_p(1)
\\
&\sup\limits_{\underset{z\in\mathds{R}}{{\boldsymbol{\beta}}^0\in\boldsymbol{B}_N(r,s_r)}}\abs{\P\left( \left[\boldsymbol{R}_{N}\hat{\boldsymbol{b}} - \boldsymbol{q} \right]' \left[\frac{\bR_N\hat{\boldsymbol{\Upsilon}}^{-2}\hat{\boldsymbol{\Omega}}{\hat{\boldsymbol{\Upsilon}}}^{-2}\bR_N^\prime}{T}\right]^{-1}\left[\boldsymbol{R}_{N}\hat{\boldsymbol{b}} - \boldsymbol{q} \right]\leq z \right)-F_P(z)}=o_p(1)
\label{eq:Wald}
\end{align}
where $\boldsymbol{\Phi}(\cdot)$ is the CDF of $N(0,1)$, $F_P(z)$ is the $CDF$ of $\chi^2_P$, and $\bq\in \mathds{R}^P$ is chosen to test a null hypothesis of the form $\bR_N\bbeta^0=\bq$.
\end{corollary}

\cref{cor:usefulResult} allows one to perform a variety of hypothesis tests.
For a significance test on a single variable $j$, for instance, take 
$\boldsymbol{R}_{N}$ as the $j$th basis vector. Then, inference on $\beta^0_j$ of the form
$\P\left(\frac{\sqrt{T}(\hat b_j-{{\beta}}_j^0)}{\sqrt{\hat\omega_{j,j}/\hat\tau^4_j}}\leq z\right)-\boldsymbol{\Phi}(z)=o_p(1),\quad \forall z\in\mathds{R}$,
can be obtained where $\boldsymbol{\Phi}(\cdot)$ is the standard normal CDF. One can then obtain standard confidence intervals
$CI(\alpha):=\left[\hat{b}_j-z_{\alpha/2}\sqrt{\frac{\hat{\omega}_{j,j}/\hat{\tau}_j^4}{T}}, \ \hat{b}_j+z_{\alpha/2}\sqrt{\frac{\hat{\omega}_{j,j}/\hat{\tau}_j^4}{T}}\right]$,
where $z_{\alpha/2}:=\boldsymbol{\Phi}^{-1}(1-\alpha/2)$, with the property that
$\sup\limits_{{\boldsymbol{\beta}}^0\in\boldsymbol{B}(s_r)}\left\vert\P\left(\beta_j^0\in CI(\alpha)\right)-(1-\alpha)\right\vert=o_p(1)$.
For a joint test with $P$ restrictions on $h$ variables of interest of the form $\boldsymbol{R}_{N} \boldsymbol{\beta}^0 = \boldsymbol{q}$, one can construct a Wald type test statistic based on \cref{eq:Wald}, and compare it to the critical value $F_P^{-1}(1-\alpha)$. Note that these results can also be used to test for nonlinear restrictions of parameters via the Delta method \citep[e.g.,][Theorems 5.5.23,28]{CB2002}.

As the bounds and convergence rates as displayed in full generality in \cref{cor:usefulResult} may be hard to interpret, we investigate in \cref{cor:simplifiedRates} how the conditions of \cref{cor:usefulResult} can be satisfied in a simplified asymptotic setup, thereby illustrating how the different growth rates interact.
As for \cref{cor:separateResults}, the conditions on $\lambda$ effectively require that $Q_T$, $N$, and $s_{r,\max}$ grow at a polynomial rate of $T$, which we exploit in \cref{cor:simplifiedRates} to simplify the conditions.

\begin{example}\label{cor:simplifiedRates}
    The requirements of \cref{cor:usefulResult} are satisfied when 
    $N\sim T^{a}$ for $a>0$, $s_{r,\max}\sim T^{b}$ for $b>0$, $Q_T\sim T^{\mathcal{Q}}$ for an arbitrarily small $\mathcal{Q}>0$, and $\lambda\sim T^{-\ell}$ for \begin{equation*}
    \begin{split}
    0<r<1:&\quad \frac{b+1/2}{2-r}<\ell<\frac{1}{r(\frac{1}{d}+\frac{m}{m-1})}\left[\frac{1}{2}-b\left(\frac{1}{d}+\frac{m}{m-1}\right)-2a\left(\frac{1}{d}+\frac{1}{m-1}\right)\right],\\
    r=0:&\quad\frac{b+1/2}{2}<\ell<\frac{1}{2}-\frac{a}{m}
    .
    \end{split}
\end{equation*}
    This choice of $\ell$ is feasible if 
    \begin{equation}\label{eq:exampleRequirements}
        \left(\frac{4b+r}{2-r}\right)\left(\frac{1}{d}+\frac{m}{m-1}\right)+4a\left(\frac{1}{d}+\frac{1}{m-1}\right)<1.
    \end{equation}
There is thus a limit on how fast $s_{r,\max}$ and $N$ can grow relative to $T$, and 
there exists a trade-off between both:
$s_{r,\max}$ can grow faster if we limit the growth rate of $N$, and vice versa. 
Besides,
for larger $r$, the conditions on the growth rate of $s_{r,\max}$ are more strict. 
The strictness of these bounds is 
additionally influenced by the number of moments $m$ and the size of the NED $-d$:
the bounds
become easier to satisfy when $m$ and $d$ are large.

Depending on the growth rates of $s_{r,\max}$ and $N$, inequality \eqref{eq:exampleRequirements} may put stricter requirements on $m$ and $d$ than those in \cref{ass:dgp}. For example, if we assume that $s_{r,\max}$ is asymptotically bounded $(b=0)$, and $N$ grows proportionally to $T$ ($a=1$), then $m$ and $d$ should satisfy $\frac{1}{d}+\frac{1}{m-1}<\frac{1}{4}$.  
If, on the other hand, 
$m$ and $d$ are allowed to be arbitrarily large, such as when the data are mixing and sub-exponential, then we only need $b<\frac{1-r}{2}$,
and we do not have an effective upper bound on $a$, implying that $N$ can grow at any polynomial rate of $T$. For a more general understanding of the restrictions imposed by \cref{eq:exampleRequirements}, Figure \ref{fig:momentsRequired} shows feasible regions for different combinations of $a$, $b$, $d$, and $r$, as well as how many moments $m$ are needed in those cases.

\begin{figure}
\caption{Required moments $m$ implied by \cref{eq:exampleRequirements}. Contours mark intervals of 10 moments, and values above $m=100$ are truncated to 100. Non-shaded areas indicate infeasible regions.}\label{fig:momentsRequired}
\centering
\includegraphics[width=\textwidth]{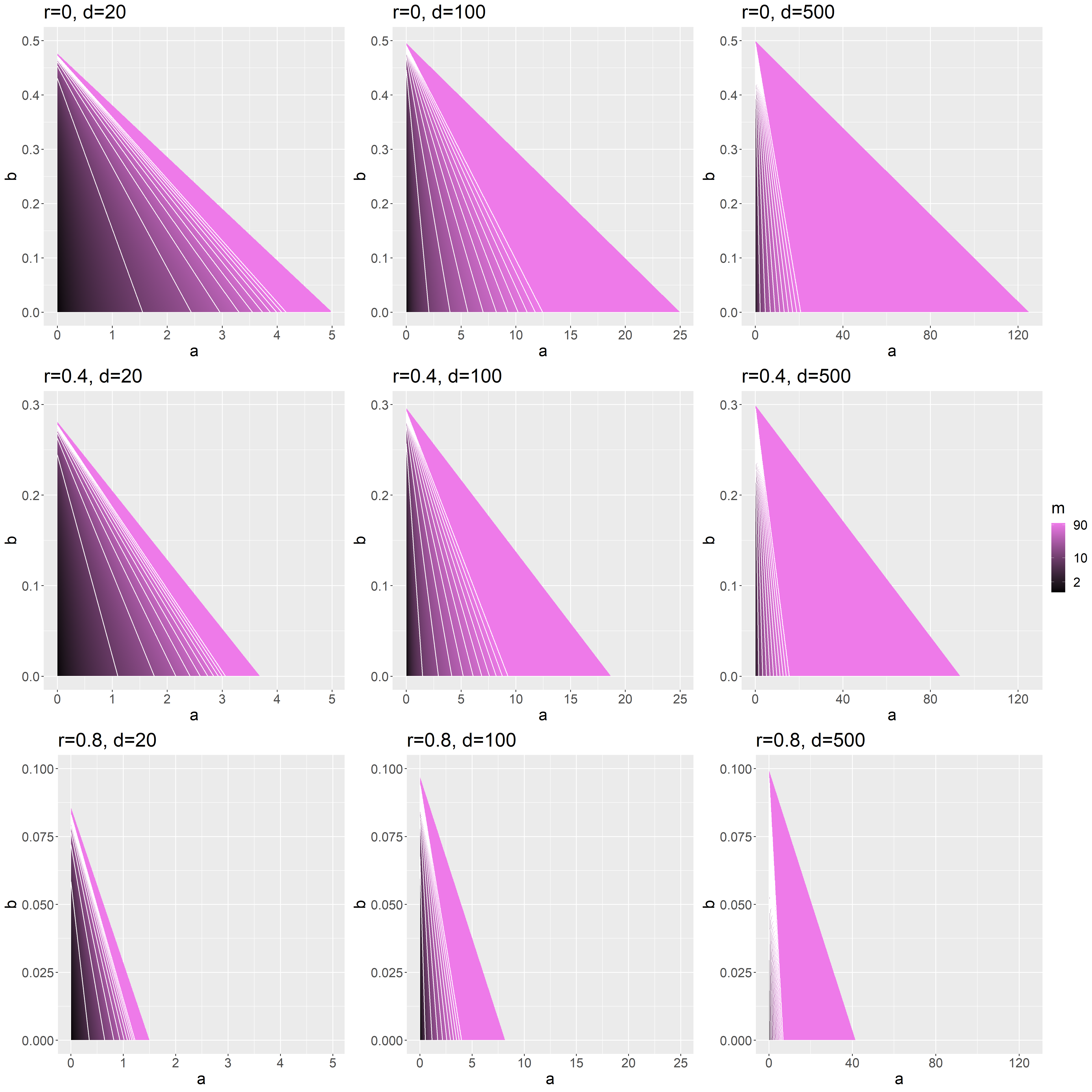}
\end{figure}
\end{example}

\subsection{Inference on high-dimensional parameters} \label{Inference_HD} 
The reason for considering $h\leq C$ in \cref{thm:CLT} lies entirely in the application of the central limit theorem. However, while inference on a finite set of parameters covers many cases of interest in practice, it does not allow for simultaneous inference on all parameters. We therefore next consider inference on a growing number of parameters (or hypotheses). We follow the approach pioneered by \cite{chernozhukov2013gaussian} to consider tests which can be formulated as a maximum over individual tests, and apply a high-dimensional CLT for the maximum of a random vector of increasing length. \cite{zhang2017gaussian} and \cite{zhang2018gaussian} provide such a CLT for high-dimensional time series, with serial dependence characterized through the functional dependence framework of \cite{Wu05}, while \cite{chernozhukov2019inference} derive a similar result under general $\beta$-mixing conditions.  In more recent work, \cite{chang2021central} derive a high-dimensional CLT for $\alpha$-mixing processes, that we base our result on. 
Recalling that a process which is NED on an $\alpha$-mixing process can be well-approximated by a mixing process, this mixing condition remains conceptually close to, if more stringent than, our NED framework.\footnote{Ideally one would directly have a high-dimensional CLT available for NED processes, such that it would directly fit to our assumptions. However, such a result is, to our knowledge, currently not available in the literature. While such a result would clearly be very interesting to obtain, this is left for future research given the intricacies needed to derive it.} We therefore build on their results to provide
distributional results for high-dimensional  inference
in \cref{thm:HDCLT}. 
While the core of the proof directly follows by applying the CLT of \cite{chang2021central}, one still needs to integrate this with the results from \cref{thm:LRVconsistency} on the consistency of the covariance matrix, as well as adapting the CLT to our estimators. We therefore believe it is worthwhile to state this as a formal result in \cref{thm:HDCLT}. Correspondingly, we now strengthen our assumptions as follows.
\begin{assumption}\label{ass:HDCLT}\hfill
\begin{enumerate}[label=(\roman*)]
    \item\label{ass:HDLCT:DGP} Let $\bz_t$ be uniformly $\alpha$-mixing with mixing coefficients satisfying $\alpha_T(q)\leq C_1\exp\left(-C_2q^K\right)$ for some $K>0$ and all $q\geq 1$.
    \item\label{ass:HDLCT:moments} Let there exist sequences $d_{u,T}$, $d_{v,T}$, $D_T=d_{u,T}d_{v,T}\geq 1$ such that $\norm{u_t}_{\psi_{2}}\leq d_{u,T},\ \norm{\bm^\prime \bv_t}_{\psi_{2}}\leq d_{v,T},\ \forall \bm\in\mathds{R}^N:\norm{\bm}_{1}\leq C,$ where  $\norm{x}_{\psi_{2}}:=\inf\left[c>0:\E\left\lbrace\exp\left[\left(x/c\right)^{2}\right]-1\right\rbrace\leq 1\right]$.
\end{enumerate}
\end{assumption}
\cref{ass:HDCLT}\ref{ass:HDLCT:DGP} 
implies \cref{ass:dgp}\ref{ass:dgpNED}. \cref{ass:dgp}\ref{ass:dgpNED} states that the NED process $\bz_t$ can be well-approximated by an $\alpha$-mixing process; clearly this holds when it is itself $\alpha$-mixing. More specifically, the sequence is NED on itself, such that \cref{ass:dgp}\ref{ass:dgpNED} is satisfied for any positive $d$.
Furthermore, the exponential decay of the $\alpha$-mixing coefficients is 
stricter than our restrictions on $\bs_{T,t}$. Similarly, the sub-gaussian moments in
\cref{ass:HDCLT}\ref{ass:HDLCT:moments} imply that all finite moments 
in \cref{ass:dgp}\ref{ass:dgpStationary} and \cref{ass:statandvmoments} exist, so $m$ may be arbitrarily large.

\begin{corollary}\label{thm:HDCLT}
Let \Cref{ass:dgp,ass:sparsity,ass:compatibility,ass:statandvmoments,ass:nodewise,ass:HDCLT} hold, and let $h\sim T^{\mcH}$ for $\mcH>0$,  $N\sim T^a$ for $a>0$, $s_{r,\max}\sim T^b$ for $0<b<\frac{1-r}{2}$, $Q_T\sim T^{\mathcal{Q}}$ for $0<\mathcal{Q}<2/3$ and  $\lambda_{\min}\sim\lambda_{\max}\sim\lambda\sim T^{-\ell}$ where
\begin{equation*}
\begin{split}
0<r<1:&\quad \frac{b+1/2}{2-r}<\ell<\frac{1/2-b}{r},\\
r=0:&\quad \frac{b+1/2}{2}<\ell<1/2.
\end{split}
\end{equation*}

Additionally, let the smallest eigenvalue of $\bOmega_{N,T}$ be bounded away from 0, and  $\frac{D_T^{2/3}(\ln T)^{(1+2K)/(3K)}}{T^{1/9}}+\frac{D_T(\ln T)^{7/6}}{T^{1/9}}\to 0$.
Then, for $1/C\leq\max\limits_{1\leq p\leq P}\norm{\br_{N,p}}_{1}\leq C$, $P\leq Ch$,
\begin{equation*}
\sup\limits_{z\in\mathds{R}, \bbeta^0 \in \bB_N(r,s_r)}\abs{\P\left(\max\limits_{1\leq p\leq P}\sqrt{T}\boldsymbol{r}_{N,p}\left(\hat{\bb}-\bbeta^0\right)\leq z\right)-\P^*\left(\max\limits_{1\leq p\leq P}\hat g_p\leq z\right)}=o_p(1),
\end{equation*}
where ${\hat\bg}$ is a $P$-dimensional vector which is distributed as $N(\bzero, \bR_N\hat\bUpsilon^{-2}\hat\bOmega\hat\bUpsilon^{-2}\bR_N^\prime)$ conditionally on the data, and $\P^*$ is the corresponding conditional probability.
\end{corollary}

Unlike \cref{cor:usefulResult}, \cref{thm:HDCLT} allows one to simultaneously test a growing number of hypotheses, while controlling for family-wise error rate, for example by the stepdown method described in Section 5 of \cite{chernozhukov2013gaussian}. One such test is an overall test of significance, with the null hypothesis $\bbeta^0=\bzero$; in this case $P=h=N$ and $\bR_N=\bI$. Note that although $\P\left(\max\limits_{1\leq p\leq P}\hat g_p\leq z\right)$ cannot be calculated analytically, it can easily be approximated with arbitrary accuracy by simulation.

Due to the stronger assumptions in \cref{thm:HDCLT}, we can relax the conditions on the growth rates of $N$ and $s_{r,\max}$ compared to \cref{cor:usefulResult} and \cref{cor:simplifiedRates}. In particular, the size of $a$ and $\mathcal{H}$ are not restricted, meaning that $N$ and $h$ can grow at an arbitrarily large polynomial rate of $T$. The conditions on $s_{r,\max}$ can also be relaxed so it can grow up to a rate of $\sqrt{T}$, depending on $r$. This corresponds to our analysis in \cref{cor:simplifiedRates} when we let $m$ and $d$ tend to infinity.

\section{Analysis of Finite-Sample Performance}\label{sec:simulations}
We analyze the finite sample performance of the desparsified lasso by means of simulations. 
We start by discussing tuning parameter selection in \Cref{subsec:lambdaselection}. We then discuss
three simulation settings:  a high-dimensional autoregressive model with exogenous variables (in \Cref{subsec:blockdiag}),  a factor model (in \Cref{subsec:factor}), and a weakly sparse VAR model (in  \Cref{subsec:var}). In \Cref{subsec:blockdiag} and \Cref{subsec:factor}, we compute coverage rates of confidence intervals for single hypothesis tests. In \Cref{subsec:var}, we perform a multiple hypothesis test for Granger causality. 

\subsection{Tuning parameter selection \label{subsec:lambdaselection}}
While the previous sections give some theoretical restrictions on the tuning parameter choice, these results cannot be used in practice since its value depends on properties of the underlying model that are unobservable. In this section, we provide a feasible recommendation to select the tuning parameters (in both the original regression and nodewise regressions) in a data-driven way.

In particular, we adapt the iterative plug-in procedure (PI) used in, for instance, \cite{BCCH12,BCH14,belloni2017program}  
to a time series setting.
We build on the theoretical relation between the tuning parameter and the empirical process in \cref{thm:ourContribution},
namely the restriction that 
$\frac{1}{T}\norm{\bX^\prime \bu}_{\infty}\leq C\lambda$ 
needs to hold with high probability, to guide the choice of $\lambda$. For large $N$ and $T$, 
$\frac{1}{T}\norm{\bX^\prime \bu}_{\infty}$ can be approximated by the maximum over an $N$-dimensional multivariate Gaussian distribution with covariance matrix $\Omega_{N,T}^{(\mathcal{E})} = \E \left[ \frac{1}{T} \bX^\prime \bu \bu^\prime \bX \right]$.\footnote{Under minimal extra assumptions (sub-Gaussian moments for $\bx_t$, and minimum eigenvalue of the long-run covariance matrix bounded away from 0),
\cref{thm:HDCLT} substantiates the validity of this approximation.} One may therefore approximate its quantiles by simulating from a multivariate Gaussian with covariance matrix a consistent estimate $\hat{\Omega}^{(\mathcal{E})}$ of $\Omega_{N,T}^{(\mathcal{E})}$.

Our time series setting requires the usage of a consistent long-run variance estimator, which is provided by \cref{thm:LRVconsistency}. We therefore take $\hat{\Omega}^{(\mathcal{E})}$ as in \cref{eq:LRC} with $\hat\bXi^{(\mathcal{E})}(l)=\frac{1}{T-l}\sum\limits_{t=l+1}^T \bx_t\hat u_t\hat u_{t-l}\bx_{t-l}^\prime$. We set the number of lags in the long-run covariance estimator as the automatic bandwidth estimator in \cite{andrews1991heteroskedasticity}, specifically $Q_T=\left\lceil1.1447(\hat\alpha(1)T)^{1/3}\right\rceil$, with $\hat\alpha(1)$ computed based on an AR(1) model, as detailed in eq. (6.4) therein. As the estimates $\hat{u}_t$ require a choice of $\lambda$, we iterate the algorithm until the chosen $\lambda$ converges. Full details are provided in Supplementary Appendix C.5. Throughout all simulations, the lasso estimates are obtained through the coordinate descent algorithm  \citep{GLMnet} applied to standardized data.

\begin{remark}
We opt to only base our empirical choice for $\lambda$ on its relation to the empirical process and hence the set $\mathcal{E}_T(\cdot)$ in \cref{thm:ourContribution}, not on its relation to the set $\mathcal{CC}_{S_\lambda}$ which also implies a lower hound $\lambda$. The latter bound, however, requires one to approximate $\Vert\hat\bSigma-\bSigma\Vert_{\max}$ which is considerably more difficult as it cannot be approximated by plugging in estimated quantities directly. With eigenvalue assumptions typically stated in terms of the sample rather than the population, this kind of additional restriction may be avoided, but such assumptions often still need to be justified by showing that the sample covariance matrix is close to the population matrix. As the additional bound only appears under weak sparsity ($r>0$), it can also be avoided by assuming exact sparsity. However, given that weak sparsity may often be the more relevant concept in practice, it may well be that the extra restriction on $\lambda$ from  bounding $\Vert\hat\bSigma-\bSigma\Vert_{\max}$ is relevant beyond our paper. Investigating ways to incorporate this in the tuning parameter selection therefore seems an interesting avenue for future research.
\end{remark}

\subsection{Autoregressive model with exogenous variables} \label{subsec:blockdiag}
Inspired by the simulation studies in \cite{KockCallot2015} (Experiment B) and \cite{MedeirosMendes16}, we take the following DGP
\begin{equation*}\begin{split}
y_t&=\rho y_{t-1}+\boldsymbol{\beta}'\boldsymbol{x}_{t-1}+u_t,\qquad
\boldsymbol{x}_t=\boldsymbol{A}_1\boldsymbol{x}_{t-1}+\boldsymbol{A}_4\boldsymbol{x}_{t-4}+\boldsymbol{\nu}_t,
\end{split}\end{equation*}
where $\boldsymbol{x}_t$ is a $(N-1)\times 1$ vector of exogenous variables.
In this simulation design (and the following ones), we consider different values of the time series length $T=\left\lbrace100,200,500,1000\right\rbrace$ and number of regressors $N=\left\lbrace101,201,501,1001\right\rbrace$. 
For this data generating process,
we take $\rho=0.6$, $\beta_j=\frac{1}{\sqrt{s}}(-1)^j$ for $j=1,\dots,s$, and zero otherwise.  For $N=101,201$ we set $s=5$ and $s=10$ for $N=501,1001$.  The autoregressive parameter matrices $\boldsymbol{A}_1$ and $\boldsymbol{A}_4$ are block-diagonal with each block of dimension $5\times5$. Within each matrix, all blocks are identical with typical elements of 0.15 and -0.1 for $\boldsymbol{A}_1$ and $\boldsymbol{A}_4$ respectively.  
Due to the misspecification of nodewise regressions, there is induced autocorrelation in the nodewise errors $v_{j,t}$.
However, the block diagonal structure of $\boldsymbol{A}_1$ and $\boldsymbol{A}_4$ keeps the sparsity of nodewise regressions constant asymptotically. 

We consider different processes for the error terms $u_t$ and $\boldsymbol{\nu}_t$:
\begin{enumerate}[label=(\roman*)]
    \item[(A)] IID errors: 
    $
    (u_t,\boldsymbol{\nu}_t^\prime)^\prime\sim\ IID\ N(\boldsymbol{0},I)$. Since all moments of the Normal distribution are finite, all moment conditions are satisfied. 
    \item[(B)] GARCH(1,1) errors: $u_t=\sqrt{h_t}\varepsilon_t,\ h_t=5\times10^{-4}+0.9h_{t-1}+0.05u_{t-1}^2,\ \varepsilon_t\sim IID\ N(0,1)$, 
    $\nu_{j,t}\sim u_t$ for $j=1,\dots,N-1$. Under this choice of GARCH parameters, not all moments of $u_t$ are guaranteed to exist, but $\E\left[ u_t^{24}\right]<\infty$.
    \item[(C)] Correlated errors: $\boldsymbol{\nu}_t\sim IID\  N(\bzero,\boldsymbol{S})$, where $\boldsymbol{S}$ has a Toeplitz structure $S_{j,k}=(-1)^{\abs{j-k}}\rho^{\abs{j-k}+1}$, with $\rho=0.4.$
\end{enumerate}
For all designs,
we evaluate whether the 95\% confidence intervals
corresponding to $\rho$ and $\beta_1$ cover their true values at the correct rates.  
The intervals are constructed as $\left[\hat{\rho}\pm z_{0.025}\sqrt{\frac{\hat{\omega}_{1,1}/\hat{\tau}_1^4}{T}}\right]$ and $\left[\hat{\beta}_1\pm z_{0.025}\sqrt{\frac{\hat{\omega}_{2,2}/\hat{\tau}_2^4}{T}}\right]$.
These results are obtained based on 2,000 replications. The rates at which the intervals contain the true values are reported in \Cref{tab:HDR}.

\begin{table}[t]
\caption{Autoregressive model with exogenous variables: 95\% confidence interval coverage. The mean interval widths are reported in parentheses.
}
\label{tab:HDR}
\centering
\begin{tabular}{@{\extracolsep{3pt}}cccccccccc@{}}
\hline
\hline
 & & \multicolumn{4}{c}{$\rho$} & \multicolumn{4}{c}{$\beta_1$}\\
 \cline{3-6}\cline{7-10}
\small{Model} & $N\backslash T$ & 100 & 200 & 500 & \multicolumn{1}{c}{1000} & 100 & 200 & 500 & 1000\\ 
\cline{1-2}\cline{3-6}\cline{7-10} 
\multirow{6}{*}{A} & 101 & $\underset{(0.366)}{0.958}$ & $\underset{(0.220)}{0.953}$ & $\underset{(0.113)}{0.951}$ & $\underset{(0.070)}{0.948}$ & $\underset{(0.383)}{0.809}$ & $\underset{(0.257)}{0.731}$ & $\underset{(0.152)}{0.751}$ & $\underset{(0.102)}{0.843}$ \\ 
 & 201 & $\underset{(0.387)}{0.965}$ & $\underset{(0.224)}{0.955}$ & $\underset{(0.116)}{0.959}$ & $\underset{(0.071)}{0.955}$ & $\underset{(0.388)}{0.790}$ & $\underset{(0.258)}{0.720}$ & $\underset{(0.154)}{0.721}$ & $\underset{(0.103)}{0.802}$ \\  
 &  501 & $\underset{(0.418)}{0.937}$ & $\underset{(0.238)}{0.950}$ & $\underset{(0.129)}{0.955}$ & $\underset{(0.081)}{0.952}$ & $\underset{(0.399)}{0.850}$ & $\underset{(0.260)}{0.786}$ & $\underset{(0.165)}{0.773}$ & $\underset{(0.113)}{0.770}$ \\ 
 & 1001 & $\underset{(0.429)}{0.936}$ & $\underset{(0.244)}{0.950}$ & $\underset{(0.130)}{0.944}$ & $\underset{(0.083)}{0.946}$ & $\underset{(0.388)}{0.819}$ & $\underset{(0.260)}{0.777}$ & $\underset{(0.164)}{0.780}$ & $\underset{(0.114)}{0.821}$ \\
\cline{1-2}\cline{3-6}\cline{7-10} 
\multirow{6}{*}{B} & 101 & $\underset{(0.374)}{0.961}$ & $\underset{(0.219)}{0.957}$ & $\underset{(0.115)}{0.953}$ & $\underset{(0.071)}{0.941}$ & $\underset{(0.390)}{0.797}$ & $\underset{(0.261)}{0.735}$ & $\underset{(0.153)}{0.760}$ & $\underset{(0.102)}{0.839}$ \\
 & 201 & $\underset{(0.387)}{0.949}$ & $\underset{(0.227)}{0.959}$ & $\underset{(0.117)}{0.954}$ & $\underset{(0.073)}{0.959}$ & $\underset{(0.398)}{0.810}$ & $\underset{(0.260)}{0.726}$ & $\underset{(0.156)}{0.721}$ & $\underset{(0.103)}{0.817}$ \\  
 & 501 & $\underset{(0.425)}{0.951}$ & $\underset{(0.241)}{0.960}$ & $\underset{(0.130)}{0.953}$ & $\underset{(0.082)}{0.954}$ & $\underset{(0.400)}{0.838}$ & $\underset{(0.263)}{0.796}$ & $\underset{(0.165)}{0.759}$ & $\underset{(0.114)}{0.775}$ \\  
 & 1001 & $\underset{(0.434)}{0.937}$ & $\underset{(0.246)}{0.960}$ & $\underset{(0.131)}{0.947}$ & $\underset{(0.084)}{0.942}$ & $\underset{(0.394)}{0.820}$ & $\underset{(0.261)}{0.787}$ & $\underset{(0.165)}{0.769}$ & $\underset{(0.115)}{0.806}$ \\   
 \cline{1-2}\cline{3-6}\cline{7-10} 
\multirow{6}{*}{C} & 101 & $\underset{(0.410)}{0.964}$ & $\underset{(0.231)}{0.960}$ & $\underset{(0.121)}{0.956}$ & $\underset{(0.080)}{0.943}$ & $\underset{(0.628)}{0.936}$ & $\underset{(0.394)}{0.887}$ & $\underset{(0.232)}{0.902}$ & $\underset{(0.166)}{0.911}$ \\
 & 201 & $\underset{(0.421)}{0.975}$ & $\underset{(0.239)}{0.965}$ & $\underset{(0.123)}{0.968}$ & $\underset{(0.081)}{0.964}$ & $\underset{(0.646)}{0.917}$ & $\underset{(0.398)}{0.899}$ & $\underset{(0.233)}{0.901}$ & $\underset{(0.166)}{0.900}$ \\ 
 & 501 & $\underset{(0.457)}{0.969}$ & $\underset{(0.260)}{0.965}$ & $\underset{(0.129)}{0.951}$ & $\underset{(0.081)}{0.948}$ & $\underset{(0.665)}{0.950}$ & $\underset{(0.420)}{0.935}$ & $\underset{(0.243)}{0.892}$ & $\underset{(0.168)}{0.903}$ \\ 
 &1001 & $\underset{(0.475)}{0.974}$ & $\underset{(0.265)}{0.960}$ & $\underset{(0.132)}{0.957}$ & $\underset{(0.082)}{0.960}$ & $\underset{(0.669)}{0.947}$ & $\underset{(0.421)}{0.938}$ & $\underset{(0.244)}{0.895}$ & $\underset{(0.168)}{0.894}$ \\ 
\hline 
\end{tabular}
\end{table}

We start by discussing the results for the model with Gaussian errors (Model A). 
Coverage for $\rho$ is close to the nominal level of 95\%  for all combinations of $N$ and $T$, with some combinations producing slightly  conservative results.
The coverage rates for $\beta_1$ are worse than for $\rho$.
This is likely due to the fact that the exogenous variables $\bx_t$ within the same block are strongly correlated to each other which negatively impacts the performance of the lasso.

Turning to the results for the model with GARCH errors (Model B), 
similar finite sample coverage rates are obtained. 
We do see a small increase in the mean interval width, which is to be expected given the heteroskedastic error structure. 
With correlated errors (Model C), we  again observe  consistent coverage rates near the nominal level for $\rho$. Interestingly, the coverage rates for $\beta_1$ appear 
considerably better than in Models A and B, though in most cases still remaining below the nominal rate at around 90\%. 
We 
also observe higher mean interval widths than Model A, which is due to larger variance of $\bx_t$ induced by the cross-sectional covariance of the errors.

In Supplementary Appendix C.6 we provide details on an examination of various selection methods for tuning parameters through heat maps for the coverage levels, which also shed some further light on the relatively poor performance for $\beta_1$ compared to $\rho$ visible for models A and B. In addition to selection by our PI method, we indicate selection by the BIC, the AIC, and the EBIC as in \cite{chen2012extended}, with $\gamma=1$.\footnote{For additional stability in the high-dimensional settings, we restrict the BIC, AIC, and EBIC to only select models with at most $T/2$ nonzero parameters, though this restriction appears to  be binding for the AIC only.} We summarize the main findings below. First, notice that there are regions with coverage close to the nominal level in nearly all scenarios and combinations of $N$ and $T$, suggesting that good coverage could be achieved by selecting the tuning parameters well. Second, across all scenarios, PI generally tends to result in coverage rates closest to the nominal coverage of 95\%. As expected, the AIC produces, overall, the least sparse solutions, the EBIC the sparsest and BIC lies in between. PI lies mostly between the BIC and EBIC. Third, there is a region of relatively low coverage for large values of the tuning parameter in the initial and nodewise regressions (see the top right corner of the heat maps). This occurs more pronouncedly for $\beta_1$ than for $\rho$ and especially for $T=1000$. Since PI tends to select near this region, it partly explains why its coverage is worse for $\beta_1$. The relatively better coverage of $\beta_1$ in Model C is matched by this region being much less prominent. Given that the regions of good coverage are in different places for $\rho$ and $\beta_1$, using the BIC or EBIC for generally smaller or larger $\lambda$ would not lead to consistently better coverage across scenarios.\footnote{To confirm this analysis, we also performed the simulations results for all three setups using selection of $\lambda$ by BIC (the best performing information criterion); in line with the heat maps, the coverage rates for BIC are generally somewhat worse than for PI. Results are available upon request.}

\subsection{Factor model} \label{subsec:factor}
We take the following factor model
\begin{equation*}\begin{split}
y_t&=\boldsymbol{\beta}'\boldsymbol{x}_{t}+u_t,\ u_t\sim IID\ N(0,1)\\
\boldsymbol{x}_t&=\bLambda f_t+\boldsymbol{\nu}_t,\ \boldsymbol{\nu}_t\sim IID\ N(\boldsymbol{0},\boldsymbol{I}),\qquad
f_t =0.5f_{t-1}+\varepsilon_t,\ \varepsilon_t\sim IID\ N(0,1),
\end{split}\end{equation*}
where $\boldsymbol{x}_t$ is a $N\times 1$ vector generated by the AR(1) factor $f_t$.
We take  $\boldsymbol{\beta}$ as in \Cref{subsec:blockdiag} with $s$ increased by one to match the number of non-zero parameters. The $N \times 1$ vector of factor loadings $\boldsymbol{\Lambda}$ is chosen with the first $s$ entries (corresponding to the variables with non-zero entries in $\bbeta$) set to 0.5, and the remaining entries $\Lambda_i=(i-s+1)^{-1}$. This choice of weakly sparse factor loadings ensures that the nodewise regressions are weakly sparse too, as shown in \Cref{ex:SparseFactor}. By letting the large loadings coincide with the non-zero entries in $\bbeta$, we ensure that there is a large potential for incurring (omitted variable) bias in the estimates, and thus that this DGP provides a serious test for the desparsified lasso.  

We investigate whether the confidence interval for $\beta_1$, $\left[\hat{\beta}_1\pm z_{0.025}\sqrt{\frac{\hat{\omega}_{1,1}/\hat{\tau}_2^4}{T}}\right]$, covers the true value at the correct rate. Results are reported in \Cref{tab:factor}.
Coverage rates 
improve with growing values of $N$ and $T$,  
with empirical coverages of approximately 85\% for small $N$ and $T$, and increasing towards the nominal level when either $N$ or $T$ increases. 
This result is therefore in line with our theoretical framework, and provides a relevant practical setting in which the desparsified lasso is appropriate to use even if exact sparsity is not present.

\begin{table}
\centering
\caption{Factor model: 95\% confidence interval coverage for $\beta_1$. The mean interval widths are reported in parentheses. 
}
\label{tab:factor}
\begin{tabular}{@{\extracolsep{3pt}}ccccc@{}}
\hline
\hline
$N\backslash T$ & 100 & 200 & 500 & 1000\\ 
\cline{1-1}\cline{2-5} 
 101 & $\underset{(0.480)}{0.890}$ & $\underset{(0.299)}{0.851}$ & $\underset{(0.163)}{0.889}$ & $\underset{(0.112)}{0.907}$ \\ 
  201 & $\underset{(0.490)}{0.873}$ & $\underset{(0.307)}{0.849}$ & $\underset{(0.165)}{0.879}$ & $\underset{(0.112)}{0.897}$ \\ 
  501 & $\underset{(0.489)}{0.956}$ & $\underset{(0.327)}{0.940}$ & $\underset{(0.180)}{0.890}$ & $\underset{(0.117)}{0.910}$ \\ 
  1001 & $\underset{(0.498)}{0.951}$ & $\underset{(0.331)}{0.943}$ & $\underset{(0.184)}{0.881}$ & $\underset{(0.117)}{0.896}$ \\ 
\cline{1-5} 
\end{tabular}
\end{table}

\subsection{Weakly sparse VAR(1)}\label{subsec:var}
Inspired by  \cite{KockCallot2015} (Experiment D), we consider the VAR(1) model
\begin{equation*}\begin{split}
\boldsymbol{z}_t= (y_t, x_{t}, \boldsymbol{w}_t)^\prime =\boldsymbol{A}_1 \boldsymbol{z}_{t-1}+\boldsymbol{u}_t,\qquad \boldsymbol{u}_t\sim IID\ N(0,1),
\end{split}\end{equation*}
with $\boldsymbol{z}_{t}$  a $(N/2)\times 1$ vector. We focus on testing whether $x_t$ Granger causes $y_t$ by fitting a a VAR(2) model, such that we have a total of $N$ explanatory variables per equation.
The $(j,k)$-th element of the autoregressive matrix $A_1^{(j,k)}=(-1)^{\vert j-k\vert}\rho^{\vert j-k\vert+1}$, with $\rho=0.4$. 
To measure the size of the test, we set $A_1^{(1,2)}=0$; to measure the power of the test, we keep its regular value of $-\rho^2$. Weak sparsity holds\footnote{The weak sparsity measure is $\sum\limits_{j=1}^{N}\vert\rho^{j}\vert^{r}$ with asymptotic limit $\frac{\rho^{r}}{1-\rho^{r}}<\infty$, trivially satisfying $B=0$.} under our choice of the autoregressive parameters, but exact sparsity is violated by having half of the parameters non-zero. 
Note that the desparsified lasso is convenient for estimating the full VAR equation-by-equation, since all equations share the same regressors, and $\hat{\boldsymbol\Theta}$ needs to be computed only once. For our Granger causality test, however, only a single equation needs to be estimated.

We test whether $x_t$ Granger causes $y_t$ by regressing $y_t$ on the first and second lag of $\boldsymbol{z}_t$. To this end, we test the null hypothesis $A^{(1,2)}_{1}=A^{(1,2)}_2=0$ by using the Wald test statistic in \cref{eq:Wald}, with $\hat{\boldsymbol{b}}_H=\left(0,\hat{A}^{(1,2)}_{1},0\dots0,\hat{A}^{(1,2)}_{2},0\dots0\right)'$, $H=\left\lbrace2,N/2+1\right\rbrace$, and $\hat{A}^{(1,2)}_{1}$, $\hat{A}^{(1,2)}_{2}$ obtained by regressing $y_t$ on $\left(\boldsymbol{z}_{t-1}',\boldsymbol{z}_{t-2}'\right)'$. We reject the null hypothesis when the statistic exceeds $\chi^2_{2,0.05}\approx5.99$.

\begin{table}[t]
\caption{Weakly sparse VAR: Joint test rejection rates for a nominal size of $\alpha=5\%$.}
\label{tab:VAR}
\centering
\begin{tabular}{@{\extracolsep{3pt}}ccccccccc@{}}
\hline\hline
& \multicolumn{4}{c}{Size} & \multicolumn{4}{c}{Power}\\
\cline{2-5}\cline{6-9}
$N\backslash T$  & 100 & 200 & 500 & 1000  & 100 & 200 & 500 & 1000 \\
\cline{1-1}\cline{2-5}\cline{6-9}
102 & 0.050 & 0.070 & 0.070 & 0.073 & 0.415 & 0.751 & 0.982 & 1.000 \\ 
  202 & 0.062 & 0.075 & 0.081 & 0.078 & 0.411 & 0.775 & 0.987 & 1.000 \\ 
  502 & 0.051 & 0.067 & 0.106 & 0.076 & 0.401 & 0.776 & 0.990 & 1.000 \\ 
  1002 & 0.059 & 0.083 & 0.101 & 0.091 & 0.407 & 0.769 & 0.995 & 1.000 \\ 
\hline
\end{tabular}
\end{table}

We start by discussing the size of the test in \Cref{tab:VAR}.
Overall, the empirical sizes exceed the nominal size of 5\%, with performance generally not improving for larger sample sizes.
In particular,
rejection rates 
slightly deteriorate
for larger $N$.
However, the observed changes in performance across $N$ and $T$ are rather small
and may be due to simulation randomness. 
The power of the test 
increases
with both $N$ and $T$, reaching 
1 at $T=1000$ 
regardless of the value for $N$.

To improve the finite-sample performance of the method, a natural extension would be to consider the bootstrap for constructing confidence
intervals as opposed to 
asymptotic theory. Bootstrap-based inference 
for desparsified lasso methods in high dimensions has already been explored by several authors, for example \cite{dezeure2017high} in the IID setting, and 
in time series by \cite{Krampe18}, \cite{chernozhukov2019inference} and \cite{chernozhukov2021timeandspace}. 
In particular, block or block multiplier bootstrap methods, which would allow one to capture serial dependence nonparametrically, would fit our setup well. The block bootstrap has the additional advantage of correcting the finite-sample performance of statistics based on long-run variance estimators, which might be a factor for our tests as well \citep{GoncalvesVogelsang11}. However, due to the lack of theory about such bootstrap methods, and the associated selection of tuning parameters like the block length, for high-dimensional NED processes, we do not consider such methods here. The development of such theory would be a highly relevant and
interesting topic for future research.

\section{Conclusion}\label{sec:conclusion} 
We provide a complete set of tools for uniformly valid inference in high-dimensional stationary time series settings, where the number of regressors $N$ can possibly grow at a faster rate than the time dimension $T$.
Our main results include
(i) an error bound for the lasso under a weak sparsity assumption on the parameter vector, thereby establishing parameter and prediction consistency;
(ii) the asymptotic normality of the desparsified lasso under a general set of conditions,  
leading to uniformly valid inference for finite subsets of parameters; 
(iii) asymptotic normality of a maximum-type statistic of a growing, high-dimensional, number of tests, valid under more stringent conditions, thereby also permitting simultaneous inference over a potentially large number of parameters, and  
(iv) a consistent Bartlett kernel Newey-West long-run covariance estimator to conduct inference in practice.

These results are established under very general conditions, thereby allowing for typical settings encountered in many econometric applications where the errors may be non-Gaussian, autocorrelated, heteroskedastic and weakly dependent.  
Crucially, this allows for certain types of misspecified time series models, such as omitted lags in an AR model. 

Through a small simulation study, we examine the finite sample performance of the desparsified lasso in popular types of time series models. 
We perform both single and joint hypothesis tests and examine the desparsified lasso's robustness to, amongst others, regressors and error terms exhibiting serial dependence and conditional heteroskedasticity, and a violation of the sparsity assumption in the nodewise regressions. 
Overall our results show that good coverage rates are obtained even when $N$ and $T$ increase jointly. 
The factor model design shows that the desparsified lasso remains applicable when the exact sparsity assumption of the nodewise regressions is violated.
Finally, Granger causality tests in the VAR are slightly oversized, but empirical sizes generally remain close to the nominal sizes, and the test's power increases with both $N$ and $T$.

There are several extensions to our approach that are interesting 
to consider. The development of a high-dimensional central limit theorem for NED processes would allow to weaken the dependence conditions needed for establishing simultaneous, high-dimensional inference. Similarly, using sample splitting would likely allow for weakening sparsity assumptions. Finally, improvements in finite sample performance may be achieved by bootstrap procedures. All of these extensions would require the development of novel theory, and thus provide challenging but worthwhile avenues for future research.

\section*{Acknowledgements}
We thank the editor, associate editor and three referees for their thorough review and highly appreciate their constructive comments which substantially improved the quality of the manuscript.

The first and second author were financially supported by the Netherlands Organization for Scientific Research (NWO) under grant number 452-17-010. The third author was supported by the European Union's Horizon 2020 research and innovation programme under the Marie Sk\l{}odowska-Curie grant agreement  No 832671. Previous versions of this paper were presented at CFE-CM Statistics 2019, NESG 2020,  Bernoulli-IMS One World Symposium 2020, (EC)$^2$ 2020, and the 2021 Maastricht Workshop on Dimensionality Reduction and Inference in High-Dimensional Time Series. We gratefully acknowledge the comments by participants at these conferences. In addition, we thank Etienne Wijler for helpful discussions. All remaining errors are our own.

\bibliographystyle{chicago}
\bibliography{bibliography}

\numberwithin{lemma}{section}
\numberwithin{equation}{section}
\numberwithin{example}{section}
\numberwithin{definition}{section}

\begin{appendices}

\section{Proofs for Section \ref{sec:Lasso}}\label{app:A1}
This section provides the theory for the lasso consistency established in \Cref{sec:Lasso}. 
We first provide some 
definitions in Appendix \ref{sec:definitions} and
preliminary lemmas in Appendix \ref{sec:prelSec3} which are proved in Supplementary Appendix C.1. The proofs of the main results are then provided in Appendix \ref{sec:mainSec3}.

\subsection{Definitions}\label{sec:definitions}
\begin{definition}[Near-Epoch Dependence, \cite{Davidson02}, ch. 17]\label{def:NED}
Let there exist non-negative NED constants $\{c_t\}_{t=-\infty}^{\infty}$, an NED sequence $\{\psi_q\}_{q=0}^{\infty}$ such that $\psi_q\to0$ as $q\to\infty$, and a (possibly vector-valued) stochastic sequence $\{\bs_{t}\}_{t=-\infty}^{\infty}$ with $\mathcal{F}_{t-l-q}^{t-l+q}=\sigma\{\bs_{t-q},\dots,\bs_{t+q}\}$, such that $\{\mathcal{F}_{t-l-q}^{t-l+q}\}_{q=0}^{\infty}$ is an increasing sequence of $\sigma$-fields. For $p>0$, the random variable $\{X_t\}_{t=-\infty}^{\infty}$ is $L_p$-NED on $\bs_t$ if
\begin{equation*}
    \left(\E\left[\abs{X_t-\E\left(X_t\vert \mathcal{F}_{t-l-q}^{t-l+q}\right)}^p\right]\right)^{1/p}\leq c_t\psi_q.
\end{equation*}
for all $t$ and $q\geq0$. Furthermore, we say $\{X_t\}$ is $L_p$-NED of size $-d$ on $\bs_t$ if $\psi_q=O(q^{-d-\varepsilon})$ for some $\varepsilon>0$.
\end{definition}
\begin{definition}[Mixingale, \cite{Davidson02}, ch. 16]\label{def:mixingale}
Let there exist non-negative mixingale constants $\{c_t\}_{t=-\infty}^{\infty}$ and mixingale sequence $\{\psi_q\}_{q=0}^{\infty}$ such that $\psi_q\to0$ as $q\to\infty$. For $p\geq1$, the random variable $\{X_t\}_{t=-\infty}^{\infty}$ is an $L_p$-mixingale with respect to the $\sigma$-algebra $\{\mathcal{F}_{t}\}_{t=-\infty}^{\infty}$ if
\begin{equation*}
    \left(\E\left[\abs{\E\left(X_t\vert \mathcal{F}_{t-q}\right)}^p\right]\right)^{1/p}\leq c_t\psi_q,
\end{equation*}
\begin{equation*}
    \left(\E\left[\abs{X_t-\E\left(X_t\vert \mathcal{F}_{t+q}\right)}^p\right]\right)^{1/p}\leq c_t\psi_q,
\end{equation*}
for all $t$ and $q\geq0$. Furthermore, we say $\{X_t\}$ is an $L_p$-mixingale of size $-d$ with respect to $\{\mathcal{F}_{t}\}$ if $\psi_q=O(q^{-d-\varepsilon})$ for some $\varepsilon>0$. Note that the latter condition holds automatically when $X_t$ is $\mathcal{F}_t$-measurable, as is the case in this paper. We 
use the same notation for the constants $c_t$ and sequence $\psi_q$ as with near-epoch dependence, since they play the same role in both types of dependence.
\end{definition}
\subsection{Preliminary results} \label{sec:prelSec3}

\begin{lemma}\label{lma:NEDimpliesdgp} Under \cref{ass:dgp}, for every $j=1,\ldots,N$,
$ \left\lbrace u_{t}x_{j,t}\right\rbrace$ is an $L_m$-Mixingale 
with respect to 
$\mathcal{F}_{t}=\sigma\left\lbrace \boldsymbol{z}_{t},\boldsymbol{z}_{t-1},\dots\right\rbrace$, with non-negative mixingale constants $c_t \leq C$ and sequence $\psi_{q}$ satisfying $\sum\limits_{q=1}^\infty \psi_q <\infty$. 
\end{lemma}

\begin{lemma}\label{lma:xxmixingale}
Under \cref{ass:dgp}, $\{x_{i,t}x_{j,t}-\E x_{i,t}x_{j,t}\}$ is $L_{\bar{m}}$-bounded and an $L_m$-mixingale with respect to $\mathcal{F}_{t}=\sigma\left\lbrace \boldsymbol{z}_{t},\boldsymbol{z}_{t-1},\dots\right\rbrace$, with non-negative mixingale constants $c_t \leq C$, and mixingale sequences of size $-d$. 
\end{lemma}

\begin{lemma}\label{lma:CovarianceCloseness}
Recall the set $\setCC(S):=\left\lbrace\norm{\hat\bSigma-\bSigma}_{\max}\leq C/\abs{S}\right\rbrace$ and $S_\lambda=\{j:\abs{\beta_j^0}>\lambda\}$.
Under \cref{ass:dgp,ass:sparsity,ass:compatibility}, for a sequence $\eta_T\to0$ such that $\eta_T\leq\frac{N^2}{e}$, if the following is satisfied
\begin{equation*}
    \lambda^{-r}s_r\leq C\eta_T^{\frac{d+m-1}{dm+m-1}}\left[\frac{\sqrt{T}}{N^{\left(\frac{2}{d}+\frac{2}{m-1}\right)}}\right]^{\frac{1}{\frac{1}{d}+\frac{m}{m-1}}}.
\end{equation*}
then $\P\left(\setCC(S_\lambda)\right)\geq 1-3\eta_T\to1$ as $N,T\to\infty$.
\end{lemma}

\begin{lemma}\label{lma:empiricalProcess} Let  $\setEP{T}(z):=\left\lbrace\max\limits_{j\leq N,l\leq T}\left[\left\vert\sum\limits_{t=1}^{l}u_t x_{j,t}\right\vert\right]\leq z\right\rbrace$.
Under \cref{ass:dgp}, we have for $z>0 $ that
\begin{equation*}
\P\left(\setEP{T}(z)\right)\geq1-CN\left(\frac{\sqrt{T}}{z}\right)^{m}.
\end{equation*}
\end{lemma}

\begin{lemma}\label{lma:compatibilityApplication} Take an index set $S$ with cardinality $\vert S\vert$.
Assuming that 
$\Vert \boldsymbol{\beta}_S\Vert_1^2\leq C\vert S\vert \boldsymbol{\beta}'{\boldsymbol{\Sigma}} \boldsymbol{\beta}$ holds for $\left\lbrace \boldsymbol{\beta}\in\mathds{R}^{N}:\Vert \boldsymbol{\beta}_{S^c}\Vert_1\leq3\Vert \boldsymbol{\beta}_{S}\Vert_1\right\rbrace$, then on the set $\setCC(S)
=\left\lbrace  \Vert\hat{\boldsymbol{\Sigma}}-{\boldsymbol{\Sigma}}\Vert_{\max}\leq C/\vert S\vert
	\right\rbrace$
\begin{equation*}
\Vert \boldsymbol{\beta}_S\Vert_1\leq C\sqrt{\vert S\vert \boldsymbol{\beta}'\hat{{\boldsymbol{\Sigma}}}\boldsymbol{\beta}},
\end{equation*}
for $\left\lbrace \boldsymbol{\beta}\in\mathds{R}^{N}:\Vert \boldsymbol{\beta}_{S^c}\Vert_1\leq3\Vert \boldsymbol{\beta}_{S}\Vert_1\right\rbrace$.
\end{lemma}

\begin{lemma}\label{lma:generalOracle} 
Let \cref{ass:compatibility} hold for an index set S, i.e. $\phi^2_{\bSigma}(S)\geq 1/C \implies \norm{\bz_S}_{1}^2\leq C\abs{S}\bz^\prime\bSigma\bz$.
On the set $\setEP{T}(T\lambda/4)\cap\setCC(S)$:
\begin{equation*}\begin{split}
\frac{\Vert \boldsymbol{X}(\hat{\boldsymbol{\beta}}-{\boldsymbol{\beta}}^0)\Vert_2^2}{T}+\frac{\lambda}{4}\Vert\hat{\boldsymbol{\beta}}-{\boldsymbol{\beta}}^0\Vert_1\leq&C\lambda^2\vert S\vert+\frac{8}{3}\lambda\Vert{\boldsymbol{\beta}}^0_{S^c}\Vert_1.\\
\end{split}\end{equation*}
\end{lemma}

\begin{lemma}\label{lma:errorBoundonSets}
Under \cref{ass:sparsity,ass:compatibility}, on the set $\setCC(S_\lambda)\cap\setEP{T}(T\lambda/4)$,
\begin{align*}
\frac{\Vert \boldsymbol{X}(\hat{\boldsymbol{\beta}}-{\boldsymbol{\beta}}^0)\Vert_2^2}{T}+\frac{\lambda}{4}\Vert\hat{\boldsymbol{\beta}}-{\boldsymbol{\beta}}^0\Vert_1 &\leq C\lambda^{2-r}s_{r}.
\end{align*}
\end{lemma}

\subsection{Proofs of the main results} \label{sec:mainSec3}

\begin{proof}[\bf Proof of \cref{thm:ourContribution}] 
In this proof we combine the results of \cref{lma:CovarianceCloseness,lma:empiricalProcess}. By applying \cref{lma:empiricalProcess} to the set $\setEP{T}(T\lambda/4)$, we have that $\P\left(\setEP{T}(T\lambda/4)\right)\geq 1-CN(\lambda \sqrt{T})^{-m}$. Choose $\eta_T$  such that $N(\lambda \sqrt{T})^{-m}\leq \eta_T$, meaning that
\begin{equation*}
    \P\left(\setEP{T}(T\lambda/4)\right)\geq 1-\eta_T\quad\text{when}\quad \lambda\geq C\eta_T^{-1/m}\frac{N^{1/m}}{\sqrt{T}}.
\end{equation*} 
For \cref{lma:CovarianceCloseness}, we need that $\eta_T\leq\frac{N^2}{e}$, which is true for sufficiently large $N,T$, since $N$ diverges, and $\eta_T$ converges with $T\to\infty$. Then
\begin{equation*}
    \P\left(\setCC(S_\lambda)\right)\geq 1-\eta_T\quad\text{when}\quad\lambda^{-r}s_r\leq C\eta_T^{\frac{d+m-1}{dm+m-1}}\left[\frac{\sqrt{T}}{N^{\left(\frac{2}{d}+\frac{2}{m-1}\right)}}\right]^{\frac{1}{\frac{1}{d}+\frac{m}{m-1}}}.
\end{equation*}
When $0<r<1$ , the required bound for the set $\setEP{T}(T\lambda/4)$ is dominated by the bound for $\setCC(S_\lambda)$ when $s_r$ does not converge to 0, i.e. $s_r\geq1/C$ (when $s_r\to0$ these results are trivial). To show this, note that for $m>2$, $d\geq1$,
$\left(\frac{N^{\left(\frac{2}{d}+\frac{2}{m-1}\right)}}{\sqrt{T}}\right)^{\frac{1}{\left(\frac{1}{d}+\frac{m}{m-1}\right)}}\geq\frac{N^{1/m}}{\sqrt{T}}$, $\eta_T^{-\frac{d+m-1}{r(dm+m-1)}}\geq\eta_T^{-1/-m}$, and $1/r>1$.
The result then follows by the union bound, $\P\left(\setCC(S_\lambda)\bigcap\setEP{T}(T\lambda/4)\right)\geq 1-(1-\P(\setCC(S_\lambda)))-(1-\P(\setEP{T}(T\lambda/4)))\geq 1-C\eta_T\to1$ as $N,T\to\infty$. The result of the theorem follows from choosing $\eta_T=C (\ln \ln T)^{-1}$.
\end{proof}

\begin{proof}[\bf Proof of \cref{cor:separateResults}] 
By \cref{thm:ourContribution}, the set  $\setCC(S_\lambda)\cap\setEP{T}(T\lambda/4)$ holds with probability at least $1-C\eta_T$, and so the error bound of \cref{lma:errorBoundonSets} holds with the same probability. With the error bound, items (i) and (ii) follow straightforwardly.
\end{proof}

\section{Proofs for Section \ref{sec:desparsifiedLasso}}\label{app:A2}
This section provides the theory for the desparsified lasso established in \Cref{sec:desparsifiedLasso}. We first provide some preliminary lemmas in Appendix \ref{sec:prelSec4} which are proved in Supplementary Appendix C.2. The proofs of the main results are then provided in Appendix \ref{sec:mainSec4}.

\subsection{Preliminary results} \label{sec:prelSec4}

\begin{lemma}\label{ass:nodewisedgp} Under \cref{ass:dgp,ass:statandvmoments}, the following holds:
\begin{enumerate}[label=(\roman*)]
\item $\E\left[v_{j,t}\right]=\boldsymbol{0},\ \forall j$, $\E\left[v_{j,t}x_{k,t}\right]=0,\ \forall k\neq j,t$.
\item\label{ass:CLTmoments} $\max\limits_{1\leq j\leq N,\ 1\leq t\leq T}\E\left[\left\vert v_{j,t}x_{j,t}\right\vert^{m}\right] \leq C.$ 
\item\label{ass:NDGPmixingale} $\{v_{j,t}  x_{k,t}\}$ is an $L_{m}$-Mixingale 
with respect to $\mathcal{F}_{t}^{(j)}=\sigma\left\lbrace v_{j,t},\boldsymbol{x}_{-j,t},v_{j,t-1},\boldsymbol{x}_{-j,t-1},\dots\right\rbrace$,
$\forall k\neq j$, with non-negative mixingale constants $c_t\leq C$ and sequences $\psi_q$ satisfying $\sum\limits_{q=1}^\infty\psi_q\leq C$.
\end{enumerate}
\end{lemma}

\begin{lemma}\label{ass:CLT} 
Let $\boldsymbol{w}_t=(w_{1,t}, \ldots, w_{N,t})'$ with $w_{j,t}=v_{j,t}u_t$. Under \cref{ass:dgp,ass:statandvmoments} the following holds: 
    \begin{enumerate}[label=(\roman*)]
		\item\label{ass:CLTmixingales1}
		$\{ w_{j,t}\}$ is $L_{\bar m}$-bounded 
		and an $L_{m}$-Mixingale of size $-d$ uniformly over $j\in\{1,\ldots,N\}$ with respect to $\mathcal{F}_{t}=\sigma\left\lbrace u_{t},\boldsymbol{v}_{t},u_{t-1},\boldsymbol{v}_{t-1},\dots\right\rbrace$, with non-negative mixingale constants $C_1\leq c_t\leq C_2$.
		\item\label{ass:CLTsummability} 
	    $\max\limits_{q\leq j,k\leq N,\ 1\leq t\leq T} \abs{\E\left[w_{j,t}w_{k,t-l}\right]}\leq C \phi_{l}$, where $\phi_{l}$ is a sequence of size $-d$, and the covariances are therefore absolutely summable.
		\item\label{ass:CLTmixingales2}
 		For all $l$, $\{w_{j,t}w_{k,t-l}-\E\left[w_{j,t}w_{k,t-l}\right]\}$ is $L_{m/2}$-bounded 
 		and an $L_1$-Mixingale of size $-d$ uniformly over $j,k\in\{1,\ldots,N\}$ with respect to 
        $\mathcal{F}_t$, with non-negative mixingale constants $c_t\leq C$.
	\end{enumerate}
\end{lemma}

\begin{lemma}
\label{lma:nodewiseCovarianceCloseness}
Recall the sets $\setCC(S):=\left\lbrace\norm{\hat\bSigma-\bSigma}_{\max}\leq C/\abs{S}\right\rbrace$, $S_\lambda=\{j:\abs{\beta^0_j}>\lambda\}$, and $S_{\lambda,j}:=\{k:\abs{\gamma^0_{j,k}}>\lambda_j\}$.  Under \cref{ass:dgp,ass:sparsity,ass:compatibility}, for a sequence $\eta_T\to0$ such that $\eta_T\leq\frac{N^2}{e}$, if the following is satisfied 
    \begin{equation*}
    \lambda_{\min}^{-r}s_{r,\max}\leq C\eta_T^{\frac{d+m-1}{dm+m-1}}\left[\frac{\sqrt{T}}{N^{\left(\frac{2}{d}+\frac{2}{m-1}\right)}}\right]^{\frac{1}{\frac{1}{d}+\frac{m}{m-1}}},
\end{equation*}
    $\P\left(\setCC(S_\lambda)\bigcap\limits_{j\in H}\setCC(S_{\lambda,j})\right)\geq 1-3(1+h)\eta_T$. 
\end{lemma}

\begin{lemma}\label{lma:nodewiseEmpiricalProcess}
Under \cref{ass:dgp,ass:statandvmoments}, for $x_j>0$ the following holds
\begin{equation*}\begin{split}
\P\left(\bigcap\limits_{j\in H}\setEP{T}^{(j)}(x_j) \right)\geq 1- C \frac{h N T^{m/2}}{\min\limits_{j\in H} x_j^m}.
\end{split}\end{equation*}
\end{lemma}

\begin{lemma}\label{ass:tailBound}
Define the set $\setLL:= \left\lbrace\max\limits_{j\in H} \abs{\frac{1}{T} \sum\limits_{t=1}^T v_{j,t}^2-\tau_j^2}\leq \frac{h} {\delta_T} \right\rbrace$, and let \cref{ass:statandvmoments} hold.
When
\begin{equation*}
    \delta_T\leq  C\eta_T(\sqrt{T}h)^{\frac{1}{1/d+m/(m-1)}},
\end{equation*}
$\P\left(\setLL\right)\geq 1-3\eta_T^{\frac{dm+m-1}{d+m-1}}\to 1$ as $N,T\to\infty$.
\end{lemma}

\begin{lemma}\label{lma:hatTauConsistency}
Under \cref{ass:nodewise}\ref{ass:nodewiseCompatibility}
\begin{equation}\label{eq:trueTauBounded}
\frac{1}{C} \leq \tau^2_j \leq C,\text{ uniformly over } j= 1,\dots,N.
\end{equation} 
Furthermore, define the set $\setNWcons:=\bigcap\limits_{j\in H} \setEP{T}^{(j)}(T\frac{\lambda_j}{4})  \bigcap\limits_{j\in H} \setCC{(S_{\lambda,j})}$ and let
\cref{ass:nodewise}\ref{ass:nodewiseSparsity} hold. On the set $\setNWcons \cap\setLL$, we have
\begin{equation*}
\max_{j\in H}\abs{\hat\tau_j^2-\tau_j^2}\leq \boundCT +C_1\bar{\lambda}^{2-r}\smaxN+C_2\sqrt{\bar{\lambda}^2\underset{\bar{}}{\lambda}^{-r}\smaxN},
\end{equation*}
and
\begin{equation*}
\max_{j\in H} \abs{\frac{1}{\hat\tau_j^2}-\frac{1}{\tau_j^2}}\leq \frac{ \boundCT +C_1\bar{\lambda}^{2-r}\smaxN+C_2\sqrt{\bar{\lambda}^2\underset{\bar{}}{\lambda}^{-r}\smaxN}}{C_3-C_4\left( \boundCT +C_1\bar{\lambda}^{2-r}\smaxN+C_2\sqrt{\bar{\lambda}^2\underset{\bar{}}{\lambda}^{-r}\smaxN}\right)}.
\end{equation*}
\end{lemma}

\begin{lemma}\label{lma:inverseQuality}
Under \cref{ass:nodewise}\ref{ass:nodewiseSparsity}--\ref{ass:nodewiseCompatibility},
it holds for a sufficiently large $T$ that on the set $\bigcap\limits_{j\in H} \setEP{T}^{(j)}(T\frac{\lambda_j}{4})
\cap\setLL$,
\begin{equation*}
\max\limits_{j\in H}\left\lbrace\Vert \boldsymbol{e}'_j-\hat{\boldsymbol{\Theta}}_j\hat{\boldsymbol{\Sigma}}\Vert_{\infty}\right\rbrace\leq \frac{\bar{\lambda}}{C_1- \boundCT -C_2\bar{\lambda}^{2-r}\smaxN},
\end{equation*}
where $\hat{\boldsymbol{\Theta}}_j$ is the $j$th row of $\hat{\boldsymbol{\Theta}}$.
\end{lemma}

\begin{lemma}\label{lma:DeltaNegligible} 
Define $\Delta:=\sqrt{T}\left(\hat{\boldsymbol{\Theta}}\hat{\boldsymbol{\Sigma}}-I\right)\left(\hat{\boldsymbol{\beta}}-{\boldsymbol{\beta}}^0\right)$, and $\setILcons:=\setEP{T} (T\frac{\lambda}{4}) \cap \setCC{(S_\lambda)}$ Under
\cref{ass:dgp,ass:sparsity,ass:nodewise}\ref{ass:nodewiseSparsity}--\ref{ass:nodewiseCompatibility}, 
on the set $\setILcons \cap \setNWcons \cap \setLL$ we have that
\begin{equation*}
\max\limits_{j\in H} \vert\Delta_j\vert \leq\sqrt{T}\lambda^{1-r}{s}_{r} \frac {\bar{\lambda}}{C_1- \boundCT -C_2\bar{\lambda}^{2-r}\smaxN}.
\end{equation*}
\end{lemma}

\begin{lemma}\label{lma:vuConsistency}
Under \cref{ass:nodewise}\ref{ass:nodewiseSparsity}--\ref{ass:nodewiseCompatibility}, on the set $\setEP{T}(T\lambda)\cap \setNWcons$,
\begin{equation*}
\max_{j\in H}\frac{1}{\sqrt{T}}\left\vert\hat{\boldsymbol{v}}_j'\boldsymbol{u}-\boldsymbol{v}_j'\boldsymbol{u}\right\vert\leq C\sqrt{T}\lambda_{\max}^{2-r}\smaxN.
\end{equation*}
\end{lemma}

\begin{lemma}\label{lma:vuempirical}
Define the set
$\setEPvuj{T}( x):=\left\lbrace\max\limits_{s\leq T}\abs{
\sum\limits_{t=1}^{s}v_{j,t}u_t}\leq x \right\rbrace$. Under \cref{ass:dgp,ass:statandvmoments}, for $x>0$ it follows that 
$\P\left(\bigcap\limits_{j\in H}\setEPvuj{T}( x)\right)\geq 1-\frac{C hT^{m/2}}{x^{m}}.$
\end{lemma}

\begin{lemma}\label{lma:nodewiseConsistency} 
Under \cref{ass:dgp,ass:compatibility,ass:statandvmoments,ass:nodewise}\ref{ass:nodewiseSparsity}--\ref{ass:nodewiseCompatibility}, on the set\\
$\setEP{T}(T\lambda) \cap \setNWcons \cap \setLL 
\bigcap\limits_{j\in H}\setEPvuj{T}(h^{1/m}T^{1/2}
\eta_T^{-1})$ with $\eta_T^{-1}\leq C \sqrt{T}$, we have 
\begin{equation*}\begin{split}
\max\limits_{j\in H}\left\vert\frac{1}{\sqrt{T}}\frac{\hat{\boldsymbol{v}}_{j}'\boldsymbol{u}}{\hat\tau_j^2}-\frac{1}{\sqrt{T}}\frac{\boldsymbol{v}_{j}'\boldsymbol{u}}{\tau_j^2}\right\vert
\leq\frac{h^{1/m}\eta_T^{-1}\boundCT+C_1h^{1/m}\eta_T^{-1}\sqrt{T}\lambda_{\max}^{2-r}\smaxN+C_2h^{1/m}\eta_T^{-1}\sqrt{\bar{\lambda}^2\underset{\bar{}}{\lambda}^{-r}\smaxN}}{C_3-C_4\left(\boundCT +C_1\bar{\lambda}^{2-r}\smaxN+C_2\sqrt{\bar{\lambda}^2\underset{\bar{}}{\lambda}^{-r}\smaxN}\right)}.
\end{split}\end{equation*}
\end{lemma}

\begin{lemma}\label{lma:tailbounds}
For any process $\{d_t\}_{t=1}^T$ and constant $x>0$, define the set
$\setTail{d}{x} := \left\{\norm{\boldsymbol{d}}_{\infty} \leq x\right\}.$ Let $\max_t \E \abs{d_t}^p \leq C < \infty$. Then for $x>0$, $\P \left(\left\lbrace\setTail{d}{x}\right\rbrace^{c} \right) \leq C x^{-p} T$.
\end{lemma}

\begin{lemma} \label{lma:crossterm} 
Under \cref{ass:dgp,ass:sparsity,ass:statandvmoments,ass:nodewise}\ref{ass:nodewiseSparsity}--\ref{ass:nodewiseCompatibility}, on the set
\begin{equation*}
\setUVcons := \setILcons \cap \setNWcons   \cap \mathcal{E}_{T,uvw},
\end{equation*}
\begin{equation*}\begin{split}
\max\limits_{(j,k)\in H^2}&\abs{\frac{1}{T}\sum\limits_{t=l+1}^{T} \left(\hat{w}_{j,t}\hat{w}_{k,t-l}- w_{j,t} w_{k,t-l} \right)} \leq C_1 \left[T^{1/2} \lambda_{\max}^{2-r} s_{r,\max}\right]^2\\
&\quad+ C_2h^{\frac{1}{m}} T^{\frac{1}{m}} \lambda_{\max}^{2-r} s_{r,\max} + C_3 \sqrt{h^{\frac{3}{m}}T^{\frac{3-m}{m}} \lambda_{\max}^{2-r} s_{r,\max}}+C_4\left[h^{\frac{1}{3m}}T^{\frac{m+1}{3m}}\lambda_{\max}^{2-r} s_{r,\max}\right]^{\frac{3}{2}}.
\end{split}\end{equation*}
\end{lemma}

\begin{lemma}\label{lma:togetherAgain}
Define 
\begin{equation*}
\setEP{T,ww}(x) := \left\{\max\limits_{(j,k)\in H^2}\abs{\frac{1}{T}\sum\limits_{t=l+1}^{T}\left({w}_{j,t}{w}_{k,t-l} - \E {w}_{j,t}{w}_{k,t-l}\right)} \leq x \right\}.
\end{equation*}
Under \cref{ass:dgp,ass:statandvmoments}, it holds that
\begin{equation*}
\begin{split}
\P \left[ \setEP{T,ww}\left(\eta_T^{-1} h^2\left(\sqrt{T}h^2\right)^{-\frac{1}{1/d+m/(m-2)}}\right) \right] \geq 1 - 3\eta_T^{\frac{dm+m-2}{2d+m-2}}.
\end{split}
\end{equation*}
\end{lemma}

\begin{lemma}\label{lma:CLTpreliminary}
Assume that $\lambda_{\max}^2\lambda_{\min}^{-r}\leq\eta_T\left[h^{2/m}\sqrt{T}s_{r,\max}\right]^{-1}$, $\frac{h^{\frac{m+1}{dm}+\frac{2}{m-1}}}{\sqrt{T}}\to0$, 
\begin{equation*}
    \lambda_{\min}^{-r}s_{r,\max}\leq C\eta_T^{\frac{d+m-1}{dm+m-1}}\left[\frac{\sqrt{T}}{\left(hN\right)^{\left(\frac{2}{d}+\frac{2}{m-1}\right)}}\right]^{\frac{1}{\frac{1}{d}+\frac{m}{m-1}}},
\end{equation*}
and if $r=0$, $\lambda_{\min}\geq \eta_T^{-1}\frac{(hN)^{1/m}}{\sqrt{T}}$.
Furthermore, assume that $\boldsymbol{R}_N$ satisfies $\max\limits_{1\leq p\leq P}\norm{\boldsymbol{r}_{N,p}}_1\leq C$, and $P\leq Ch$. Then, as $N,T\to \infty$,
\begin{equation*}\begin{split}
\max\limits_{1\leq p\leq P}\abs{\boldsymbol{r}_{N,p}\left(\frac{\hat{\boldsymbol{\Theta}}\boldsymbol{X}'\boldsymbol{u}}{\sqrt{T}}+\Delta-\frac{\boldsymbol{\Upsilon}^{-2} \boldsymbol{V}'\boldsymbol u}{\sqrt{T}}\right)}\overset{p}{\to}0.
\end{split}\end{equation*}
\end{lemma}

\begin{lemma}\label{lma:RemaindersSpeed}

Let \Cref{ass:dgp,ass:sparsity,ass:compatibility,ass:statandvmoments,ass:nodewise,ass:HDCLT} hold, and let $h\sim T^{\mcH}$ for $\mcH>0$,  $N\sim T^a$ for $a>0$, $s_{r,\max}\sim T^b$ for $0<b<\frac{1-r}{2}$, $\lambda_{\min}\sim\lambda_{\max}\sim\lambda\sim T^{-\ell}$ and
\begin{equation*}
\begin{split}
0<r<1:&\ \frac{1/2+b}{2-r}<\ell<\frac{1/2-b}{r},\\
r=0:&\ \frac{1/2+b}{2-r}<\ell<1/2,
\end{split}
\end{equation*}
and $Q_T\sim T^{\mathcal{Q}}$ for $0<\mathcal{Q}<2/3$. Under these conditions,
\begin{align}
R^{\Omega}_{N,T} &:= \norm{\bR_N \left(\bUpsilon^{-2} \bOmega_{N,T} \bUpsilon^{-2}-\hat{\bUpsilon}^{-2} \hat{\bOmega}_{N,T} \hat{\bUpsilon}^{-2} \right) \bR_N^\prime}_{\max}
=O_p\left(T^{\frac{1}{2}(b-\ell(2-r))}\right), \label{eq:def_deltaT}\\
R^{\beta}_{N,T} &:= \max\limits_{1\leq p\leq P}\abs{\boldsymbol{r}_{N,p}\left(\frac{\hat{\boldsymbol{\Theta}}\boldsymbol{X}'\boldsymbol{u}}{\sqrt{T}}+\Delta-\frac{\boldsymbol{\Upsilon}^{-2} \boldsymbol{V}'\boldsymbol u}{\sqrt{T}}\right)} =O_p\left(T^{\epsilon-1/2}+T^{1/2+b-\ell(2-r)}\right), \label{eq:def_rhoNT}
\end{align}
for an arbitrarily small $\epsilon>0$, with $\frac{1}{2}(b-\ell(2-r))<-1/4$, and $1/2+b-\ell(2-r)<0$.
\end{lemma}

\subsection{Proofs of main results} \label{sec:mainSec4}

\begin{proof}[\bf Proof of \cref{thm:CLT}]
Using \cref{eq:bhatDef}, we can write 
\begin{equation*}\begin{split}
\sqrt{T}\boldsymbol{R}_{N}\left(\hat{\boldsymbol{b}}-{\boldsymbol{\beta}}^0\right)=&\sqrt{T}\boldsymbol{R}_{N}\left(\hat{\boldsymbol{\beta}} -{\boldsymbol{\beta}}^0+\frac{\hat{\boldsymbol{\Theta}}\boldsymbol{X}'({\boldsymbol{y}}-\boldsymbol{X}\hat{\boldsymbol{\beta}})}{T}\right)
=\boldsymbol{R}_{N}\left(\frac{\hat{\boldsymbol{\Theta}}\boldsymbol{X}'\boldsymbol{u}}{\sqrt{T}}+\Delta\right),
\end{split}\end{equation*}
and by \cref{lma:CLTpreliminary}, 
\begin{equation*}\begin{split}
\max\limits_{1\leq p\leq P}\abs{\boldsymbol{r}_{N,p}\left(\frac{\hat{\boldsymbol{\Theta}}\boldsymbol{X}'\boldsymbol{u}}{\sqrt{T}}+\Delta-\frac{\boldsymbol{\Upsilon}^{-2} \boldsymbol{V}'\boldsymbol u}{\sqrt{T}}\right)}\overset{p}{\to}0.
\end{split}\end{equation*}
Note that under the assumption that $h\leq C$, the requirements for \cref{lma:CLTpreliminary} reduce to the requirements for \cref{thm:CLT} (note that one of the bounds becomes redundant for $0<r<1$, see the proof of \cref{thm:ourContribution} for details). 
The proof will therefore continue by deriving the asymptotic distribution of
\begin{equation*}
    \boldsymbol{R}_N\frac{\boldsymbol{\Upsilon}^{-2} \boldsymbol{V}'\boldsymbol u}{\sqrt{T}}=\frac{1}{\sqrt{T}}\boldsymbol{R}_N\boldsymbol{\Upsilon}^{-2}\sum\limits_{t=1}^T\boldsymbol{w}_t,
\end{equation*}
and applying Slutsky's theorem.
Regarding $\boldsymbol{R}_{N}$, under the assumption that $h<\infty$, we may without loss of generality consider the case with $P=1$. In the multivariate setting, let $\boldsymbol{R}^*_N$ be a $P\times N$ matrix with $1<P<\infty$, and non-zero columns indexed by the set $H$ of cardinality $h=\vert H\vert<\infty$. By the Cram\'er-Wold theorem, $\sqrt{T}\boldsymbol{R}^*_N(\hat{\boldsymbol{b}}-\boldsymbol{\beta}^0)\overset{d}{\to}N(\boldsymbol{0},\boldsymbol{\Psi}^*)$ if and only if $\sqrt{T}\boldsymbol{\alpha}'\boldsymbol{R}^*_N(\hat{\boldsymbol{b}}-\boldsymbol{\beta}^0)\overset{d}{\to}N(\boldsymbol{0},\boldsymbol{\alpha}'\boldsymbol{\Psi}^*\boldsymbol{\alpha})$ for all $\boldsymbol{\alpha}\neq\boldsymbol{0}$. We show this directly by letting the $1\times N$ vector $\boldsymbol{R}_N=\boldsymbol{\alpha}'\boldsymbol{R}^*_N$ and the scalar $\psi=\lim\limits_{N,T\to\infty}\boldsymbol{\alpha}'\boldsymbol{R}^*_{N}({\boldsymbol{\Upsilon}}^{-2}{\boldsymbol{\Omega}_{N,T}}{\boldsymbol{\Upsilon}}^{-2}){{\boldsymbol{R}^{*\prime}_{N}}}\boldsymbol{\alpha}$.
The final part of the proof is then devoted to establishing the central limit theorem.
This result can be shown by applying Theorem 24.6 and Corollary 24.7 of \cite{Davidson02}. Following the notation therein, let $X_{T,t}=\frac{1}{\sqrt{P_{N,T}\psi T}}\boldsymbol{R}_{N}\boldsymbol{\Upsilon}^{-2}\boldsymbol{w}_t,$ where $P_{N,T}=\frac{\boldsymbol{R}_N\boldsymbol{\Upsilon}^{-2}\boldsymbol{\Omega}_{N,T}\boldsymbol{\Upsilon}^{-2}\boldsymbol{R}'_N}{\psi}$; note that by definition of $\psi$, $P_{N,T}\to1$ as $N,T\to\infty$. Further, let $\mathcal{F}^t_{T,-\infty}=\sigma\left\lbrace\boldsymbol{s}_{T,t},\boldsymbol{s}_{T,t-1},\dots\right\rbrace$, the positive constant array $\left\lbrace c_{T,t}\right\rbrace=\frac{1}{\sqrt{P_{N,T}\psi T}}$, and $r=\bar m$. We show that the requirements of this Theorem are satisfied.

Part (a), $\mathcal{F}^t_{T,-\infty}$-measurability of $X_{T,t}$, follows from the measurability of $\boldsymbol{z}_t$ in \cref{ass:dgp}\ref{ass:dgpNED}, $\E\left[X_{T,t}\right]=\frac{1}{\sqrt{P_{N,T}\psi T}}\boldsymbol{R}_{N}\boldsymbol{\Upsilon}^{-2}\E\left[\boldsymbol{w}_t\right]=0$ follows from the rewriting $w_{j,t}=\left(x_{j,t}-\boldsymbol{x}'_{-j,t}{\boldsymbol{\gamma}}^0_j\right)u_t$ and noting that $\E\left[x_{j,t}u_t\right]=0,~\forall j$ by \cref{ass:dgp}\ref{ass:dgpStationary}, and 
\begin{equation*}\begin{split}
\E\left[\left(\sum\limits_{t=1}^TX_{T,t}\right)^2\right]&=\frac{1}{P_{N,T}\psi }\boldsymbol{R}_{N}\boldsymbol{\Upsilon}^{-2}\E\left[\frac{1}{T}\left(\sum\limits_{t=1}^{T}\boldsymbol{w}_{t}\right)\left(\sum\limits_{t=1}^{T}\boldsymbol{w}_{t}'\right)\right]\boldsymbol{\Upsilon}^{-2}\boldsymbol{R}'_{N}\\
&=\frac{1}{P_{N,T}\psi }\boldsymbol{R}_{N}\boldsymbol{\Upsilon}^{-2}\boldsymbol{\Omega}_{N,T}\boldsymbol{\Upsilon}^{-2}\boldsymbol{R}'_{N}=1.
\end{split}\end{equation*}

For part (b) we get that
\begin{equation*}\begin{split}
&\sup_{T,t}\left\lbrace\left(\E\vert \boldsymbol{R}_{N}\boldsymbol{\Upsilon}^{-2}\boldsymbol{w}_t\vert^{\bar m}\right)^{1/\bar m}\right\rbrace=\sup_{T,t}\left\lbrace\left(\E\left\vert \sum\limits_{j\in H}\frac{r_{N,j}}{\tau^2_j}w_{j,t}\right\vert^{\bar m}\right)^{1/\bar m}\right\rbrace\\
&\underset{(1)}{\leq} \sum\limits_{j\in H}\frac{\vert r_{N,j}\vert}{\tau^2_j}\sup_{T,t}\left\lbrace\left( \E\vert w_{j,t}\vert^{\bar m}\right)^{1/\bar m}\right\rbrace \underset{(2)}{\leq} C,
\end{split}\end{equation*}
where (1) is due to Minkowski's inequality, and (2) follows from $h<0$, $\tau^2_j\leq C$ by \cref{eq:trueTauBounded}, and $w_{j,t}$ is $L_{\bar m}$-bounded by \cref{ass:CLT}\ref{ass:CLTmixingales1}.

For part (c'), 
by the arguments in the proof of \cref{ass:CLT}, $w_{j,t}$ is $L_m$-NED of size $-d$, and therefore also size $-1$  on $\boldsymbol{s}_{T,t}$, which is $\alpha$-mixing of size $-\frac{d}{1/m-1/\bar m}<-\bar m/(\bar m-2)$ under \cref{ass:dgp}.

For (d'), we let $M_T=\max\limits_{t}\left\lbrace c_{T,t}\right\rbrace=\frac{1}{\sqrt{P_{N,T}\psi T}}$, such that $\sup\limits_{T}TM_T^2=\sup\limits_{T}\frac{1}{\boldsymbol{R}_{N}\boldsymbol{\Upsilon}^{-2}\boldsymbol{\Omega}_{N,T}\boldsymbol{\Upsilon}^{-2}\boldsymbol{R}'_{N}} \leq C$,
where the inequality follows from $\frac{1}{\tau_j^2}\geq \frac{1}{C}$ by \cref{eq:trueTauBounded}, and $\boldsymbol{R}_{N}\boldsymbol{\Upsilon}^{-2}\boldsymbol{\Omega}_{N,T}\boldsymbol{\Upsilon}^{-2}\boldsymbol{R}'_{N}$ is bounded from below by the minimum eigenvalue of  $\boldsymbol{\Omega}_{N,T}$ (assumed to be bounded away from 0), via the Min-max theorem.

Finally, \Cref{thm:CLT} states that this convergence is uniform in ${\boldsymbol{\beta}}^0\in\boldsymbol{B}(s_r)$. This follows by noting that eq.~(C.3) holds uniformly in ${\boldsymbol{\beta}}^0\in\boldsymbol{B}(s_r)$.
\end{proof}

\begin{proof}[\bf Proof of \cref{thm:LRVconsistency}]
The following derivations collectively require that the set
\begin{equation*}
\setILcons \cap \setNWcons \cap \setLL \cap \mathcal{E}_{T,uvw} \cap \setEP{T,ww}\left(\eta_T^{-1}h^2\left(\sqrt{T}h^2\right)^{-\frac{1}{1/d+m/(m-2)}}\right)
\end{equation*}
holds with probability converging to 1. For $\setILcons \cap \setNWcons \cap \setLL$, this can be shown by the arguments in the proof of
\cref{lma:CLTpreliminary} when the following convergence rates hold: 
 $\lambda_{\max}^2\lambda_{\min}^{-r}\leq\eta_T\left[h^{2/m}\sqrt{T}s_{r,\max}\right]^{-1}$, $\frac{h^{\frac{m+1}{dm}+\frac{2}{m-1}}}{\sqrt{T}}\to0$,
\begin{equation*}
    \lambda_{\min}^{-r}s_{r,\max}\leq C\eta_T^{\frac{d+m-1}{dm+m-1}}\left[\frac{\sqrt{T}}{\left(hN\right)^{\left(\frac{2}{d}+\frac{2}{m-1}\right)}}\right]^{\frac{1}{\frac{1}{d}+\frac{m}{m-1}}},
\end{equation*}
and if $r=0$, $\lambda_{\min}\geq \eta_T^{-1}\frac{(hN)^{1/m}}{\sqrt{T}}$.
$\mathcal{E}_{T,uvw}$ follows from \cref{lma:crossterm}, and $\setEP{T,ww}\left(\eta_T^{-1}h^2\left(\sqrt{T}h^2\right)^{-\frac{1}{1/d+m/(m-2)}}\right)$ holds with probability converging to 1 by \cref{lma:togetherAgain}.
We can write
\begin{equation*}
\begin{split}
\abs{\boldsymbol{R}_{N} \left[\hat{\boldsymbol{\Upsilon}}^{-2} \hat{\boldsymbol{\Omega}} \hat{\boldsymbol{\Upsilon}}^{-2}-{\bUpsilon}^{-2} {\bOmega_{N,T}} {\bUpsilon}^{-2}\right] \boldsymbol{R}'_{N}} 
&\leq \abs{\boldsymbol{R}_{N} \left[\hat{\boldsymbol{\Upsilon}}^{-2} \hat{\boldsymbol{\Omega}} \hat{\boldsymbol{\Upsilon}}^{-2} - \boldsymbol{\Upsilon}^{-2}  \hat{\boldsymbol{\Omega}} \boldsymbol{\Upsilon}^{-2} \right] \boldsymbol{R}'_{N}} \\
&\quad + \abs{\boldsymbol{R}_{N} \left[\boldsymbol{\Upsilon}^{-2} \hat{\boldsymbol{\Omega}} \boldsymbol{\Upsilon}^{-2}-\bUpsilon^{-2}\bOmega_{N,T}\bUpsilon^{-2}\right] \boldsymbol{R}'_{N}} =: R_{(\text{a})} + R_{(\text{b})}.
\end{split}
\end{equation*}
For $R_{(a)}$ we get that
\begin{equation*}
\begin{split}
R_{(\text{a})} & \leq \abs{\boldsymbol{R}_{N} \left[ \hat{\boldsymbol{\Upsilon}}^{-2} - \boldsymbol{\Upsilon}^{-2} \right] \hat{\boldsymbol{\Omega}} \left[ \hat{\boldsymbol{\Upsilon}}^{-2} - \boldsymbol{\Upsilon}^{-2} \right] \boldsymbol{R}'_{N}} + 2\abs{\boldsymbol{R}_{N} \left[ \hat{\boldsymbol{\Upsilon}}^{-2} - \boldsymbol{\Upsilon}^{-2} \right] \hat{\boldsymbol{\Omega}} \boldsymbol{\Upsilon}^{-2} \boldsymbol{R}'_{N}} \\
&\leq \abs{\boldsymbol{R}_{N} \left[ \hat{\boldsymbol{\Upsilon}}^{-2} - \boldsymbol{\Upsilon}^{-2} \right] \left[ \hat{\boldsymbol{\Omega}} - \boldsymbol{\Omega}_{N,Q_T} \right] \left[ \hat{\boldsymbol{\Upsilon}}^{-2} - \boldsymbol{\Upsilon}^{-2} \right] \boldsymbol{R}'_{N}} \\
&\quad+ \abs{\boldsymbol{R}_{N} \left[ \hat{\boldsymbol{\Upsilon}}^{-2} - \boldsymbol{\Upsilon}^{-2} \right] \boldsymbol{\Omega}_{N,Q_T} \left[ \hat{\boldsymbol{\Upsilon}}^{-2} - \boldsymbol{\Upsilon}^{-2} \right] \boldsymbol{R}'_{N}} \\
&\quad + 2\abs{\boldsymbol{R}_{N} \left[ \hat{\boldsymbol{\Upsilon}}^{-2} - \boldsymbol{\Upsilon}^{-2} \right] \left[ \hat{\boldsymbol{\Omega}} - \boldsymbol{\Omega}_{N,Q_T} \right] \boldsymbol{\Upsilon}^{-2} \boldsymbol{R}'_{N}} + 2\abs{\boldsymbol{R}_{N} \left[ \hat{\boldsymbol{\Upsilon}}^{-2} - \boldsymbol{\Upsilon}^{-2} \right] \boldsymbol{\Omega}_{N,Q_T} \boldsymbol{\Upsilon}^{-2} \boldsymbol{R}'_{N}},
\end{split}
\end{equation*}
where
\begin{equation*}
\boldsymbol{\Omega}_{N,Q_T}:=\E\left[\frac{1}{Q_T}\left(\sum\limits_{t=1}^{Q_T}\boldsymbol{w}_t\right)\left(\sum\limits_{t=1}^{Q_T}\boldsymbol{w}'_t\right)\right] =  \boldsymbol{\Xi}(0)+\sum_{l=1}^{Q_T-1}  \boldsymbol{\Xi}(l)+\boldsymbol{\Xi}^\prime(l),
\end{equation*}
where the $(j,k)$th element of $\bXi(l)$ is $\xi_{j,k}=\frac{1}{T}\sum\limits_{t=l+1}^T\E w_{j,t}w_{k,t-l}$.

Starting with the third term of $R_{(\text{a})}$, applying the triangle inequality
\begin{equation*}\begin{split}
    &\norm{\boldsymbol{R}_{N} \left[ \hat{\boldsymbol{\Upsilon}}^{-2} - \boldsymbol{\Upsilon}^{-2} \right] \left[ \hat{\boldsymbol{\Omega}} - \boldsymbol{\Omega}_{N,Q_T} \right] \boldsymbol{\Upsilon}^{-2} \boldsymbol{R}'_{N}}_{\max}\\
    &\quad\leq\max\limits_{1\leq p,q\leq P}\left\lbrace\sum\limits_{j\in H}\sum\limits_{k\in H}\abs{r_{N,p,j}\left(\frac{1}{\hat{\tau}_j^2}-\frac{1}{\tau_j^2}\right)\left(\hat\omega_{j,k}-\omega_{j,k}^{N,Q_T}\right)\frac{1}{\tau_k^2}r_{N,q,k}}\right\rbrace\\
    &\quad\leq\max\limits_{j\in H}\abs{\frac{1}{\hat{\tau}_j^2}-\frac{1}{\tau_j^2}}\max\limits_{j\in H}\frac{1}{\tau_{j}^2}\max\limits_{(j,k)\in H^2}\abs{\hat\omega_{j,k}-\omega_{j,k}^{N,Q_T}}\max\limits_{1\leq p,q\leq P}\left\lbrace\norm{\boldsymbol{r}_{N,p}}_1\norm{\boldsymbol{r}_{N,q}}_1\right\rbrace,
\end{split}\end{equation*}
$\max\limits_{1\leq p\leq P}\norm{\boldsymbol{r}_{N,p}}_1\leq C$ by assumption,  $\max\limits_{j\in H}\frac{1}{\tau_{j}^2}\leq C$ by \cref{eq:trueTauBounded}, and 
\begin{equation*}
    \max\limits_{j\in H}\abs{\frac{1}{\hat{\tau}_j^2}-\frac{1}{{\tau}_j^2}}\leq \frac{ \boundCT +C_1\bar{\lambda}^{2-r}\smaxN+C_2\sqrt{\bar{\lambda}^2\underset{\bar{}}{\lambda}^{-r}\smaxN}}{C_3-C_4\left( \boundCT +C_1\bar{\lambda}^{2-r}\smaxN+C_2\sqrt{\bar{\lambda}^2\underset{\bar{}}{\lambda}^{-r}\smaxN}\right)}\to 0,
\end{equation*}
on the set $\setNWcons \cap\setLL$ by \cref{lma:hatTauConsistency}. Finally, we show that $\max\limits_{(j,k)\in H^2}\abs{\hat{\omega}_{j,k}-\omega_{j,k}^{N,Q_T}}\to 0$.
\begin{equation*}
\begin{split}
\max\limits_{(j,k)\in H^2}\abs{\hat\omega_{j,k} - \omega_{j,k}^{N,Q_T}} &\leq 2\sum_{l=0}^{Q_T-1} \max\limits_{(j,k)\in H^2}\abs{(1-l/Q_T)\hat\xi_{j,k}(l) - \xi_{j,k}(l)} \\
&2\sum_{l=0}^{Q_T-1} \max\limits_{(j,k)\in H^2}\abs{(1-l/Q_T)\frac{1}{T-l}\sum_{t=l+1}^T\hat w_{j,t}\hat w_{k,t-l} - \frac{1}{T}\sum_{t=l+1}^T \E w_{j,t}w_{k,t-l}}. \\
\end{split}
\end{equation*}
Using a telescopic sum argument, 
\begin{equation*}
\begin{split}
\max\limits_{(j,k)\in H^2}\abs{\hat\omega_{j,k} - \omega_{j,k}^{N,Q_T}} &\leq 
2\sum_{l=0}^{Q_T-1}\max\limits_{(j,k)\in H^2}\abs{\frac{1}{T}\sum_{t=l+1}^T(\hat w_{j,t}\hat w_{k,t-l}-\E w_{j,t} w_{k,t-l})}\\
&+\frac{l}{Q_T}\max\limits_{(j,k)\in H^2}\abs{\frac{1}{T}\sum\limits_{t=l+1}^T\E w_{j,t} w_{k,t-l}}. \\
\end{split}
\end{equation*}

For the second term, it follows by \cref{ass:CLT}\ref{ass:CLTsummability} that 
\begin{equation*}\begin{split}
2\sum_{l=0}^{Q_T-1} \frac{l}{Q_T} \max\limits_{j,k\in H^2}\abs{\E w_{j,t} w_{k,t-l}} \leq \frac{C}{Q_T}\sum_{l=1}^{Q_T-1} l^{1-d-\epsilon}
&\leq C Q_T^{-d-\epsilon} \sum_{l=1}^{Q_T-1}\frac{l^{1-d-\epsilon}}{Q_T^{1-d-\epsilon}}\leq C Q_T^{1-d-\epsilon},
\end{split}\end{equation*}
since $l/Q_T<1$, and
$Q_T^{1-d-\epsilon}\to0$ for $d\geq1$, and $ \sum_{l=1}^{Q_T-1} l^{-1-\delta}\leq C$ by properties of $p$-series.
It follows from \cref{lma:crossterm,lma:togetherAgain} that 
\begin{equation*}
\begin{split}
\max\limits_{(j,k)\in H^2}\abs{\frac{1}{T}\sum_{t=l+1}^T\left(\hat w_{j,t}\hat w_{k,t-l}- w_{j,t} w_{k,t-l}\right)} &\leq
C_1 \left[T^{1/2} \lambda_{\max}^{2-r} s_{r,\max}\right]^2
+ C_2h^{\frac{1}{m}} T^{\frac{1}{m}} \lambda_{\max}^{2-r} s_{r,\max} \\
&\quad+ C_3 \sqrt{h^{\frac{3}{m}}T^{\frac{3-m}{m}} \lambda_{\max}^{2-r} s_{r,\max}}+C_4\left[h^{\frac{1}{3m}}T^{\frac{m+1}{3m}}\lambda_{\max}^{2-r} s_{r,\max}\right]^{\frac{3}{2}}\\
\max\limits_{(j,k)\in H^2}\abs{\frac{1}{T}\sum_{t=l+1}^T \left(w_{j,t} w_{k,t-l}-\E w_{j,t} w_{k,t-l}\right)} &\leq C_5 \eta_T^{-1}h^2\left(\sqrt{T}h^2\right)^{-\frac{1}{1/d+m/(m-2)}}.
\end{split}
\end{equation*}
on the set $\setUVcons \cap \setEP{T,ww}\left(\eta_T^{-1}h^2\left(\sqrt{T}h^2\right)^{-\frac{1}{1/d+m/(m-2)}}\right)$. Plugging the upper bounds in, we find that
\begin{equation*}
\begin{split}
\max\limits_{(j,k)\in H^2}\abs{\hat\omega_{j,k} - \omega_{j,k}^{N,Q_T}} 
&\leq 2 Q_T  \left[ C_1 \left[T^{1/2} \lambda_{\max}^{2-r} s_{r,\max}\right]^2
+ C_2h^{\frac{1}{m}} T^{\frac{1}{m}} \lambda_{\max}^{2-r} s_{r,\max}\right.\\
&\quad + C_3 \sqrt{h^{\frac{3}{m}}T^{\frac{3-m}{m}} \lambda_{\max}^{2-r} s_{r,\max}}+C_4\left[h^{\frac{1}{3m}}T^{\frac{m+1}{3m}}\lambda_{\max}^{2-r} s_{r,\max}\right]^{\frac{3}{2}}\\
&\quad \left. +  C_5 \eta_T^{-1}h^2\left(\sqrt{T}h^2\right)^{-\frac{1}{1/d+m/(m-2)}} \right] +C_6 Q_T^{1-d-\epsilon}.
\end{split}
\end{equation*}
Hence, $\max\limits_{(j,k)\in H^2}\abs{\hat\omega_{j,k} - \omega_{j,k}^{N,Q_T}} \xrightarrow{p} 0$ if we take 
\begin{equation*}\begin{split}
    \lambda_{\max}^{2-r}\leq \eta_T\min&\left\lbrace\left[\sqrt{Q_T}\sqrt{T}s_{r,\max}\right]^{-1}\right.,\left[Q_Th^{1/m}T^{1/m}s_{r,\max}\right]^{-1},\\
    &\left.\left[Q_T^2h^{3/m}T^{(3-m)/m}s_{r,\max}\right]^{-1},\left[Q_T^{2/3}h^{1/(3m)}T^{(m+1)/3m}s_{r,\max}\right]^{-1}\right\rbrace,
\end{split}\end{equation*}
and $Q_T\eta_T^{-1}h^2\left(\sqrt{T}h^2\right)^{-\frac{1}{1/d+m/(m-2)}}\to0$. For the latter term, since we can choose $\eta_T^{-1}$ to grow arbitrarily slowly, it is sufficient to assume $Q_Th^2\left(\sqrt{T}h^2\right)^{-\frac{1}{1/d+m/(m-2)}}\to0$. Furthermore, this convergence rate is stricter than the previous rate $\frac{h^{\frac{m+1}{dm}+\frac{2}{m-1}}}{\sqrt{T}}\to0$, and therefore makes it redundant.

For the fourth term of $R_{(\text{a})}$, we may bound as follows
\begin{equation*}\begin{split}
    &\norm{\boldsymbol{R}_{N} \left[ \hat{\boldsymbol{\Upsilon}}^{-2} - \boldsymbol{\Upsilon}^{-2} \right] \boldsymbol{\Omega}_{N,Q_T}  \boldsymbol{\Upsilon}^{-2} \boldsymbol{R}'_{N}}_{\max}\\
    &\quad\leq\max\limits_{1\leq p,q\leq P}\left\lbrace\sum\limits_{j\in H}\sum\limits_{k\in H}\abs{r_{N,p,j}\left(\frac{1}{\hat{\tau}_j^2}-\frac{1}{\tau_j^2}\right)\omega_{j,k}^{N,Q_T}\frac{1}{\tau_k^2}r_{N,q,k}}\right\rbrace\\
    &\quad\leq\max\limits_{j\in H}\abs{\frac{1}{\hat{\tau}_j^2}-\frac{1}{\tau_j^2}}\max\limits_{j\in H}\frac{1}{\tau_{j}^2}\max\limits_{(j,k)\in H^2}\abs{\omega_{j,k}^{N,Q_T}}\max\limits_{1\leq p,q\leq P}\left\lbrace\norm{\boldsymbol{r}_{N,p}}_1\norm{\boldsymbol{r}_{N,q}}_1\right\rbrace,
\end{split}\end{equation*}
The only new term here is $\max\limits_{(j,k)\in H^2}\abs{\omega_{j,k}^{N,Q_T}}$, which can by bounded by
\begin{equation*}
    \max\limits_{(j,k)\in H^2}\abs{\omega_{j,k}^{N,Q_T}}\leq \norm{\boldsymbol{\Omega}_{N,Q_T}}_{\max}\leq 2\sum_{l=0}^{Q_T-1} \norm{\boldsymbol{\Xi}(l)}_{\max}\leq C,
\end{equation*}
where the last inequality follows from \cref{ass:CLT}\ref{ass:CLTsummability}.

Note that when the third and fourth terms of $R_{(\text{a})}$ converge to 0, this holds for the first and second terms as well; one may simply replace $\max\limits_{j\in H}\frac{1}{\tau_{j}^2}$ by a second $\max\limits_{j\in H}\abs{\frac{1}{\hat\tau_{j}^2}-\frac{1}{\tau_{j}^2}}\to0$ in the upper bound.

This concludes the part of $R_{(\text{a})}$. With the results above, it remains to be shown for $R_{(\text{b})}$ that $\norm{\boldsymbol{R}_{N}\boldsymbol{\Upsilon}^{-2} \left({\hat{\boldsymbol{\Omega}}} - \boldsymbol{\Omega}_{N,T}\right) \boldsymbol{\Upsilon}^{-2}{\boldsymbol{R}'_{N}}}_{\max}\to0$. Using similar arguments as for the terms of $R_{(\text{a})}$, 
it suffices to show that $\max\limits_{(j,k)\in H^2}\abs{\omega_{j,k}^{N,Q_T} - \omega_{j,k}} \to 0$. Note that by \cref{ass:CLT}\ref{ass:CLTsummability}
\begin{equation*}
\begin{split}
\abs{\omega_{j,k}^{N,Q_T} - \omega_{j,k}}& \leq \abs{\sum_{l=1}^{Q_T-1}\xi_{j,k}-\sum_{l=1}^{T-1}\xi_{j,k}}
\leq\sum_{l=Q_T}^{T} \abs{\xi_{j,k}(l)} \leq \sum_{l=Q_T}^{T} C\phi_l\leq C\sum_{l=Q_T}^{T} l^{-d-\epsilon},
\end{split}
\end{equation*}
which converges to 0 by letting $\delta=\epsilon/2$, and writing $ \sum_{l=Q_T}^{T} l^{-d-\epsilon}\leq Q_T^{1-d-\delta}\sum_{l=Q_T}^{T} l^{-1-\delta}$, where $Q_T^{1-d-\delta}\to0$ for $d\geq1$, and $\sum_{l=Q_T}^{T} l^{-1-\epsilon}\to0$ by properties of $p$-series and $Q_T\to\infty$.
This shows that $\norm{R_{(\text{b})}}_{\max} \xrightarrow{p} 0$.

Summarizing the above, we argue that for some $\delta>0$,
\begin{equation*}
\begin{split}
\norm{\bR_N \left(\bUpsilon^{-2} \bOmega_{N,T} \bUpsilon^{-2}-\hat{\bUpsilon}^{-2} \hat{\bOmega}_{N,T} \hat{\bUpsilon}^{-2} \right) \bR_N^\prime}_{\max} \leq C_1\Delta\tau \left[1 + \Delta \tau + \Delta \tau \Delta \omega \right] + C_2 Q_T^{1-d-\delta}
\end{split}
\end{equation*}
where $\Delta\tau:=\max\limits_{j\in H}\abs{\frac{1}{\hat{\tau}_j^2}-\frac{1}{{\tau}_j^2}}$ and $\Delta\omega:=\max\limits_{(j,k)\in H^2}\abs{\hat\omega_{j,k} - \omega_{j,k}^{N,Q_T}}$.

Finally, this result holding uniformly in ${\boldsymbol{\beta}}^0\in\boldsymbol{B}(s_r)$ follows the same logic as the proof of \Cref{thm:CLT}, namely that eq.~(C.3) holds uniformly in ${\boldsymbol{\beta}}^0\in\boldsymbol{B}(s_r)$.
\end{proof}

\begin{proof}[\bf Proof of \cref{cor:usefulResult}]
The result follows by applying \cref{thm:CLT,thm:LRVconsistency}, so the assumed conditions from both must be satisfied. Since we assume that $h\leq C$ and $\lambda\sim\lambda_{\max}\sim\lambda_{\min}$, the conditions will simplify considerably. To summarize, we require the following six conditions: For \cref{thm:CLT} we require that

\begin{align*}
&\text{(1)} & & \lambda_{\max}^2\lambda_{\min}^{-r}\leq\eta_T\left[\sqrt{T}s_{r,\max}\right]^{-1}, \\
&\text{(2)} & & \lambda_{\min}^{-r}s_{r,\max}\leq  \eta_{T}\left[\frac{\sqrt{T}}{N^{\left(\frac{2}{d}+\frac{2}{m-1}\right)}}\right]^{\frac{1}{\left(\frac{1}{d}+\frac{m}{m-1}\right)}}, \\
&\text{($3^*$)} & & \lambda_{\min}\geq\eta_{T}^{-1}\frac{N^{1/m}}{\sqrt{T}} \text{ when $r=0$},\\
&\text{(4)} & & \lambda^{2-r}_{\max}\leq \eta_T\min \left\lbrace\left[\sqrt{Q_T}\sqrt{T}s_{r,\max}\right]^{-1}\right.,\left[Q_Th^{1/m}T^{1/m}s_{r,\max}\right]^{-1},\\
& & & \qquad\left.\left[Q_T^2h^{3/m}T^{(3-m)/m}s_{r,\max}\right]^{-1},\left[Q_T^{2/3}h^{1/(3m)}T^{(m+1)/3m}s_{r,\max}\right]^{-1}\right\rbrace, \\
&\text{(5)} & & \lambda_{\max}^2\lambda_{\min}^{-r} \leq\eta_T\left[\sqrt{T}h^{2/m}s_{r,\max}\right]^{-1} \\
&\text{(6)} & & \lambda_{\min}^{-r}s_{r,\max}\leq  \eta_{T}\left[\frac{\sqrt{T}}{(hN)^{\left(\frac{2}{d}+\frac{2}{m-1}\right)}}\right]^{\frac{1}{\left(\frac{1}{d}+\frac{m}{m-1}\right)}},\\
&\text{($7^*$)} & & \lambda_{\min}\geq\eta_{T}^{-1}\frac{(hN)^{1/m}}{\sqrt{T}} \text{ when $r=0$},\\
&\text{(8)} & & Q_Th^2(\sqrt{T}h^2)^{-\frac{1}{1/d+m/(m-2)}}\to 0,
\end{align*}
where (1)-($3^*$) follow from \Cref{thm:CLT} and (4)-(8) from \Cref{thm:LRVconsistency}.
Note that (1), (2), and ($3^*$) are same as the terms (4), (5), and ($6^*$) and without the $h$ terms. For (4), this can be simplified into a single (slightly more strict) upper bound $\lambda_{\max}\leq C\eta_T\left[Q_T^2\sqrt{T}h^{3/m}s_{r,\max}\right]^{\frac{-1}{2-r}}$. We may then combine this with (5),
and both are satisfied when $\lambda_{\max}^2\lambda_{\min}^{-r}\leq\eta_T\left[Q_T^2\sqrt{T}h^{3/m}s_{r,\max}\right]^{-1}$.
Using $h\leq C$ and $\lambda\sim\lambda_{\max}\sim\lambda_{\min}$, these simplify to 
\begin{align*}
\text{(1)}&\qquad \lambda^{-r}s_{r,\max}\leq  \eta_{T}\left[\frac{\sqrt{T}}{N^{\left(\frac{2}{d}+\frac{2}{m-1}\right)}}\right]^{\frac{1}{\left(\frac{1}{d}+\frac{m}{m-1}\right)}}, \\
\text{($2^*$)}& \qquad \lambda\geq\eta_{T}^{-1}\frac{N^{1/m}}{\sqrt{T}} \text{ when $r=0$}, \\
\text{(3)}&\qquad  \lambda^{2-r}\leq\eta_T\left[Q_T^2\sqrt{T}s_{r,\max}\right]^{-1},\\
\text{(4)}& \qquad Q_TT^{-\frac{1}{2/d+2m/(m-2)}}\to0.\\
\end{align*}

When $0<r<1$, from (1) and (3) we get 
\begin{equation*}
    \eta_T^{-1}s_{r,\max}^{1/r}\left[\frac{N^{\left(\frac{2}{d}+\frac{2}{m-1}\right)}}{\sqrt{T}}\right]^{\frac{1}{r\left(\frac{1}{d}+\frac{m}{m-1}\right)}}\leq\lambda\leq\eta_T\left[Q_T^2\sqrt{T}s_{r,\max}\right]^{-1/(2-r)},
\end{equation*}
and by combining the upper and lower bounds, we obtain the condition
\begin{equation*}
    Q_T^rs_{r,\max}N^{\left(2-r\right)\left(\frac{d+m-1}{dm+m-1}\right)}T^{\frac{1}{4}\left(r-\frac{d(m-1)(2-r)}{dm+m-1}\right)}\to0.
\end{equation*}
When $r=0$, the bounds on $\lambda$ come from ($2^*$) and (3)
\begin{equation*}
    \eta_{T}^{-1}\frac{N^{1/m}}{\sqrt{T}}\leq\lambda\leq \eta_{T}\left[Q_T^2\sqrt{T}s_{0,\max}\right]^{-1/2}.
\end{equation*}
Combining the upper and lower bounds, we obtain the condition 
\begin{equation*}
    Q_T^2s_{0,\max}\frac{N^{2/m}}{\sqrt{T}}\to0.
\end{equation*}
From (1), we then obtain the condition
\begin{equation*}
     s_{0,\max}N^{2\left(\frac{d+m-1}{dm+m-1}\right)}T^{-\frac{1}{2}\left(\frac{d(m-1)}{dm+m-1}\right)}\to0,
\end{equation*}
which is the same condition which came from (1) and (3) in the $0<r<1$ case. Collectively, we then need to satisfy the following
\begin{equation*}
\begin{split}
&\left\lbrace\begin{array}{ll}
\eta_T^{-1}s_{r,\max}^{1/r}\left[\frac{N^{\left(\frac{2}{d}+\frac{2}{m-1}\right)}}{\sqrt{T}}\right]^{\frac{1}{r\left(\frac{1}{d}+\frac{m}{m-1}\right)}}\leq\lambda\leq\eta_T\left[Q_T^2\sqrt{T}s_{r,\max}\right]^{-1/(2-r)}\text{ when $0<r<1$},\\
\eta_{T}^{-1}\frac{N^{1/m}}{\sqrt{T}}\leq\lambda\leq \eta_{T}\left[Q_T^2\sqrt{T}s_{0,\max}\right]^{-1/2}\text{ when $r=0$},\\
Q_T^rs_{r,\max}N^{\left(2-r\right)\left(\frac{d+m-1}{dm+m-1}\right)}T^{\frac{1}{4}\left(r-\frac{d(m-1)(2-r)}{dm+m-1}\right)}\to0,\\
    Q_T^2s_{0,\max}\frac{N^{2/m}}{\sqrt{T}}\to0\text{ when $r=0$},\\
    Q_TT^{-\frac{1}{2/d+2m/(m-2)}}\to0.
\end{array}\right.		
\end{split}
\end{equation*}

By implication of \Cref{thm:CLT}
\begin{equation*}
\sqrt{T}\boldsymbol{r}_{N,p}(\hat{{\boldsymbol{b}}}-{\boldsymbol{\beta}}^{0})\overset{d}{\to}N(0,\psi),
\end{equation*}
uniformly in ${\boldsymbol{\beta}}^0\in\boldsymbol{B}(s_r)$. Then, by \Cref{thm:LRVconsistency}
\begin{equation*}
\boldsymbol{r}_{N,p}({\hat{\boldsymbol{\Upsilon}}}^{-2}\hat{\boldsymbol{\Omega}}{\hat{\boldsymbol{\Upsilon}}}^{-2}){\boldsymbol{r}'_{N,p}}\overset{p}{\to}\psi,
\end{equation*}
also uniformly in ${\boldsymbol{\beta}}^0\in\boldsymbol{B}(s_r)$. By Slutsky's Theorem, it is then the case that
\begin{equation*}
\sqrt{T}\boldsymbol{r}_{N,p}(\hat{{\boldsymbol{b}}}-{\boldsymbol{\beta}}^{0})\overset{d}{\to}N(0,\psi),
\end{equation*}
uniformly in ${\boldsymbol{\beta}}^0\in\boldsymbol{B}(s_r)$, for every $1\leq p\leq P$. As $P<\infty$ by assumption, it follows that
\begin{equation*}
\sup\limits_{\underset{1\leq p\leq P, z\in\mathds{R}}{{\boldsymbol{\beta}}^0\in\boldsymbol{B}(s_r)}}\left\vert\P\left(\sqrt{T}\frac{\boldsymbol{r}_{N,p}(\hat{{\boldsymbol{b}}}-{\boldsymbol{\beta}}^{0})}{\sqrt{\boldsymbol{r}_{N,p}(\hat{{\boldsymbol{\Upsilon}}}^{-2}\hat{\boldsymbol{\Omega}}\hat{{\boldsymbol{\Upsilon}}}^{-2}){\boldsymbol{r}'_{N,p}}}}\leq z\right)-\boldsymbol{\Phi}(z)\right\vert=o_p(1).
\end{equation*}
Note that uniform convergence over $z\in \mathds{R}$ follows automatically by Lemma 2.11 in \cite{vandervaart1998asymptoticstatistics}, since the distribution is continuous. The second result then follows from the fact that a sum of $P$ squared standard Normal variables have a $\chi_P^2$ distribution. 
\end{proof}

\begin{proof}[\bf Proof of \cref{thm:HDCLT}]
Define $\bg \sim N(\bzero, \bR_N \bUpsilon^{-2} \bOmega_{N,T} \bUpsilon^{-2}\bR_N^\prime)$ as the `population counterpart' of $\hat{\bg}$ and define the following distribution functions:
\begin{align*}
F_{1,T} (z) &:= \P \left(\max_{1\leq p \leq P} \sqrt{T}\boldsymbol{r}_{N,p}\left(\hat{\bb}-\bbeta^0\right) \leq z \right)
& F_{2,T} (z) &:= \P\left(\max_{1\leq p \leq P} \frac{1}{\sqrt{T}}\boldsymbol{r}_{N,p}\bUpsilon^{-2} \sum_{t=1}^T\bw_t \leq z \right),\\
G_T (z) &:= \P\left(\max_{1\leq p \leq P} g_p \leq z \right)
& G_T^* (z) &:= \P^* \left(\max_{1\leq p \leq P} \hat{g}_p \leq z \right).
\end{align*}
Now note that
\begin{equation*}
\abs{F_{1,T}(z) - G_T^* (z)} \leq \underbrace{\abs{F_{1,T} (z) - G_T(z)}}_{R^{FG}_T(z)} + \underbrace{\abs{G_{T} (z) - G_T^*(z)}}_{R^{GG}_T(z)}.
\end{equation*}

For $R^{FG}_T(z)$, write $\hat{\bx}_T =  \sqrt{T}\boldsymbol{r}_{N,p}\left(\hat{\bb}-\bbeta^0\right)$ and $\bx_T =  \frac{1}{\sqrt{T}}\boldsymbol{r}_{N,p}\bUpsilon^{-2} \sum_{t=1}^T\bw_t$, such that $F_{1,T} (z) = \P (\max_{p} \hat{x}_{T,p} \leq z)$ and $F_{2,T} (z) = \P (\max_{p} x_{T,p} \leq z)$, and let $r_T := \max_{1\leq p \leq P} \hat{x}_{T,p} - \max_{1\leq p \leq P} x_{T,p}$. Then
\begin{equation*}
\abs{r_T} = \abs{\max_{1\leq p \leq P} \hat{x}_{T,p} - \max_{1\leq p \leq P} x_{T,p}} \leq \max_{1\leq p \leq P} \abs{\hat{x}_{T,p} - x_{T,p}} = R_{N,T}^\beta,
\end{equation*}
where $R_{N,T}^\beta$ is defined in \eqref{eq:def_rhoNT}. Given our assumptions, we therefore know that there exist sequences $\eta_{T,1}$ and $\eta_{T,2}$ such that $\P\left(\abs{r_T}>\eta_{T,1}\right)\leq \eta_{T,2}$, such that
\begin{align*}
&\abs{F_{1,T}(z) - G_T (z)} 
\leq \abs{\P\left(\left.\max_{p} x_{T,p} + r_T \leq z \right| \abs{r_T} \leq \eta_{T,1} \right) \P\left(\abs{r_T} \leq \eta_{T,1}\right) - \P(\max_{p} g_p \leq z)}\\
&\quad\qquad\qquad\qquad\quad + \P\left(\left.\max_{p} \hat{x}_{T,p} \leq z \right| \abs{r_T} > \eta_{T,1}\right) \P\left(\abs{r_T} > \eta_{T,1}\right)\\
&\quad \leq \abs{\P\left(\max_{p} x_{T,p} \leq z + \eta_{T,1} \right) - \P(\max_{p} g_p \leq z)} + 2 \eta_{T,2}\\
&\quad \leq \underbrace{\abs{\P\left(\max_{p} x_{T,p} \leq z + \eta_{T,1} \right) - \P(\max_{p} g_p \leq z + \eta_{T,1})}}_{R^{FG}_{T,1}(z+\eta_{T,1})}\\
&\quad \quad + \underbrace{\abs{\P\left(\max_{p} g_p \leq z + \eta_{T,1} \right) - \P(\max_{p} g_p \leq z)}}_{R^{FG}_{T,2} (z)}  + 2 \eta_{T,2}.
\end{align*}
For the term $R^{FG}_{T,1}(z+\eta_{T,1})$ we apply the high-dimensional CLT in Theorem 1 of \cite{chang2021central}, noting that our assumptions imply the conditions required for this theorem. In particular, for the sub-exponential moment assumption, we need that $\norm{\br_{N,p}\bUpsilon^{-2}\bw_t}_{\psi_{\gamma_1}}\leq D_T$ for all $t$ and $p$, for some $\gamma_1\geq 1$. We choose $\gamma_1=1$, and use Lemma 2.7.7 of \cite{vershynin2019high} to bound $\norm{\br_{N,p}\bUpsilon^{-2}\bw_t}_{\psi_1}\leq \norm{\br_{N,p}\bUpsilon^{-2}\bv_t}_{\psi_2}\norm{u_t}_{\psi_2}\leq d_{v,T}d_{u,T}=D_T$. We assume that $L_1$-bounded linear combinations of $\bv_t$ are sub-Gaussian, which covers this case, since the $\norm{r_{N,p}}_{1}\leq C$ by assumption, and $\norm{\bUpsilon^{-2}}_{\max}\leq C$ by \cref{eq:trueTauBounded}. The non-degeneracy condition then follows from choosing $1/C\leq \norm{\bR_N}_1$, and assuming the minimum eigenvalue (and therefore the smallest diagonal element) of $\bOmega_{N,T}$ is bounded away from 0. Defining $\underline{\omega}_T := \min_{1\leq p \leq P} \E g_i^2$, this implies that  $\underline{\omega}_T \geq C>0$. Applying the CLT, we bound as follows
\begin{equation*}\begin{split}
R^{FG}_{T,1}(z+\eta_{T,1}) &\leq \sup \limits_{z\in\mathds{R}} R^{FG}_{T,1}(z) \leq \sup\limits_{\bz\in\mathds{R}^{P} }\abs{\P\left(\frac{\bR_N\bUpsilon^{-2}}{\sqrt{T}}\sum\limits_{t=1}^T\bw_t\leq \bz\right)-\P\left(\bg\leq \bz\right)}\\
&\leq C_1\frac{B_T^{2/3} (\ln P)^{(1+2K)/(3K)}} {T^{1/9}}+C_2\frac{B_T(\ln P)^{7/6}}{T^{1/9}} \to 0.
\end{split}\end{equation*}
The final result holds as $\ln P\leq \ln Ch=O\left(\ln T^\mcH\right)=O(\ln T)$, since $\mcH$ is a constant.

For the term $R^{FG}_{T,2}(z)$, apply the anti-concentration bound in Lemma 2.1 of \cite{chernozhukov2013gaussian} to show that
\begin{equation*}
\begin{split}
R^{FG}_{T,2} (z) &\leq \sup_{z\in\mathbb{R}} \P(z \leq \max_{p} g_p \leq z + \eta_{T,1}) \leq \sup_{z\in\mathbb{R}} \P \left(\abs{\max_{p} g_p - z} \leq \eta_{T,1} \right)\\
&\leq C \eta_{T,1} \left[\sqrt{2 \ln P} + \sqrt{1 \vee \ln (\underline{\omega}_T / \eta_{T,1})} \right] 
\leq C_1 \eta_{T,1} \sqrt{2 \ln P}.
\end{split}
\end{equation*}
By \cref{lma:RemaindersSpeed} we find that $R_{N,T}^{\beta} = O_p\left(\left[T^{\epsilon-1/2}+T^{1/2+b-\ell(2-r)}\right]\sqrt{\ln T}\right) = O_p(T^{-\delta})$ for some $\delta>0$, since $\epsilon>0$ can be chosen arbitrarily small, and $1/2+b-\ell(2-r)<0$. We may therefore take $\eta_{T,1}$ at a polynomial rate as well, such that $\eta_{T,1} \sqrt{2\ln(P)}\to0$.

For $R_{T}^{GG} (z)$, it follows by Theorem 2 in \cite{chernozhukov2015comparison} that
\begin{equation*}
\sup_{z\in\mathbb{R}} \abs{R_{T}^{GG} (z)} \leq C (R^{\Omega}_{N,T})^{1/3}\left(\max\{1,\ln(P/R^{\Omega}_{N,T})\}\right)^{2/3},
\end{equation*}
with $R^{\Omega}_{N,T}$ as defined in \cref{eq:def_deltaT}.
By \cref{lma:RemaindersSpeed} we have $R^{\Omega}_{N,T}=O_p(T^{-1/4})$, such that
\begin{equation*}
    (R^{\Omega}_{N,T})^{1/3}\left(\max\{1,\ln(P/R^{\Omega}_{N,T})\}\right)^{2/3}= O_p\left(T^{-1/12}\left(\max\left\lbrace 1,(\mcH+1/4)\ln T\right\rbrace\right)^{2/3}\right)=o_p(1). \qedhere
\end{equation*}
\end{proof}

\section{Supplementary Results} \label{app:suppl}
Appendices \ref{sec:supplS3} and \ref{sec:supplS4} present the proofs of the preliminary results from Sections 3 and 4, respectively. Appendix \ref{sec:simplifiedRatesLasso} contains \cref{cor:simplifiedRatesLasso}. Appendix \ref{sec:supplS4examples} provides the details on Examples 5 and 6. Appendix \ref{sec:supplS4.3} contains the algorithm for choosing the tuning parameter.

\subsection{Proofs of preliminary results Section 3} \label{sec:supplS3}

\begin{proof}[\bf Proof of Lemma A.1]
$L_{\bar m}$-boundedness of $\{x_{j,t} u_t\}$ follows directly from the $L_{2\bar m}$-boundedness of $\{\boldsymbol{z}_t\}$ and the Cauchy--Schwarz inequality. By Theorem 17.9 in \cite{Davidson02} it follows that $\{x_{j,t} u_t\}$ is $L_m$-NED on $\{\boldsymbol{s}_{T,t}\}$ of size $-1$. We then apply Theorem 17.5 in \cite{Davidson02} to conclude that $\{x_{j,t} u_t\}$ is an $L_m$-mixingale of size $-\min\{1, \frac{d}{(1/m-1/\bar{m})}(1/m - 1 / \bar{m})\} = -1$, with respect to $\mathcal{F}^{\boldsymbol{s}}_t=\sigma\{\boldsymbol{s}_{T,t},\boldsymbol{s}_{T,t-1},\dots\}$; the $\mathcal{F}_t^{\boldsymbol{s}}$-measurability of $\boldsymbol{z}_t$ implies $\sigma\{\boldsymbol{z}_{t},\boldsymbol{z}_{t-1},\dots\}\subset\mathcal{F}^{\boldsymbol{s}}_t$, which in turn implies that $\{x_{j,t} u_t\}$ it is also an $L_m$-mixingale with respect to $\mathcal{F}_t=\sigma\{\boldsymbol{z}_{t},\boldsymbol{z}_{t-1},\dots\}$. The summability condition $\sum\limits_{q=1}^{\infty}\psi_q<\infty$ is satisfied by the convergence property of $p$-series: $\sum\limits_{q=1}^{\infty}q^{-p}<\infty$ for any $p>1$.
\end{proof}
\begin{proof}[\bf Proof of Lemma A.2]
$L_{\bar {m}}$-boundedness of $\{x_{i,t}x_{j,t}-\E x_{i,t}x_{j,t}\}$ follows directly from the $L_{2\bar{m}}$-boundedness of $\{\bz_t\}$ and the Cauchy-Schwarz inequality. By Theorem 17.9 of \cite{Davidson02}  
the product of two NED processes is also NED, with the order halved. It follows that $\{x_{i,t} x_{j,t}\}$ is $L_m$-NED on $\{\boldsymbol{s}_{T,t}\}$ of size $-d$. Therefore, $\E x_{i,t}x_{j,t}$ is  trivially NED. Theorem 17.8 in \cite{Davidson02} implies that also $\{x_{i,t}x_{j,t}-\E x_{i,t}x_{j,t}\}$ is $L_m$-NED. We then apply Theorem 17.5 in \cite{Davidson02} to conclude that $\{x_{i,t}x_{j,t}-\E x_{i,t}x_{j,t}\}$ is an $L_m$-mixingale of size
$-\min\{d, \frac{d}{(1/m-1/\bar{m})}(1/m - 1 / \bar m)\} = -d$, with respect to $\mathcal{F}^{\boldsymbol{s}}_t=\sigma\{\boldsymbol{s}_{T,t},\boldsymbol{s}_{T,t-1},\dots\}$; the $\mathcal{F}_t^{\boldsymbol{s}}$-measurability of $\boldsymbol{z}_t$ implies $\sigma\{\boldsymbol{z}_{t},\boldsymbol{z}_{t-1},\dots\}\subset\mathcal{F}^{\boldsymbol{s}}_t$, which in turn implies that $\{x_{i,t}x_{j,t}-\E x_{i,t}x_{j,t}\}$ is also an $L_m$-mixingale with respect to $\mathcal{F}_t=\sigma\{\boldsymbol{z}_{t},\boldsymbol{z}_{t-1},\dots\}$.
The boundedness of mixingale constants comes from Theorem 17.5, noting that the NED constants of $\{\bz_{j,t}\}$ are bounded by Assumption 1(ii),
%\cref{ass:dgp}\ref{ass:dgpNED}, 
and $\{x_{i,t}x_{j,t}-\E x_{i,t}x_{j,t}\}$ is appropriately $L_{\bar m}$-bounded. 
\end{proof}

\begin{proof}[\bf Proof of Lemma A.3]
By the union bound
\begin{equation*}\begin{split}
\P\left(\norm{\hat\bSigma-\bSigma}_{\max}>C/\abs{S}\right)
\leq&\sum\limits_{i=1}^{N}\sum\limits_{j=1}^{N}\P\left(\abs{\sum\limits_{t=1}^T\left(x_{i,t}x_{j,t}-\E\left[x_{i,t}x_{j,t}\right]\right)}>CT/\abs{S}\right).
\end{split}\end{equation*}
Now apply the Triplex inequality (\citealp{Jiang09Triplex}) 
\begin{equation*}\begin{split}
&\P\left(\abs{\sum\limits_{t=1}^T\left(x_{i,t}x_{j,t}-\E\left[x_{i,t}x_{j,t}\right]\right)}>CT/\abs{S}\right)\leq 2q\exp\left(\frac{C^2}{288}\frac{-T}{\abs{S}^2 q^2\kappa_T^2}\right)\\
&\quad+\frac{6}{C}\frac{\abs{S}}{T}\sum\limits_{t=1}^T\E\left[\abs{ \E\left(x_{i,t}x_{j,t}\vert\mathcal{F}_{t-q}\right)-\E\left(x_{i,t}x_{j,t}\right)}\right]
+\frac{15}{C}\frac{\abs{S}}{T}\sum\limits_{t=1}^T\E\left[\abs{x_{i,t}x_{j,t}}\mathds{1}_{\left\lbrace\abs{x_{i,t}x_{j,t}}>\kappa_T\right\rbrace}\right]\\
&:=R_{(\text{i})}+R_{(\text{ii})}+R_{(\text{iii})}.
\end{split}\end{equation*}
For the first term, we have
\begin{equation*}
\sum\limits_{i=1}^N\sum\limits_{j=1}^{N}R_{(i)} = 2N^2q \exp\left(\frac{C^2}{288} \frac{-T}{\abs{S}^2 q^2\kappa_T^2}\right)
\end{equation*}
so we need $N^2q\exp\left(\frac{-T}{ \abs{S}^2q^2\kappa_T^2}\right)\to0$.
By Lemma A.2
%\cref{lma:xxmixingale}
and Jensen's inequality, 
we have that $\E\left[\abs{\E\left[x_{i,t}x_{j,t}\vert \mathcal{F}_{t-q}\right]-\E\left[x_{i,t}x_{j,t}\right]}\right]\leq c_t\psi_q$, and thus for the second term that
\begin{equation*}
R_{(\text{ii})}\leq \frac{6}{C} \frac{\abs{S}} {T} \sum\limits_{t=1}^T c_t \psi_q \leq C\abs{S} \psi_q 
, \qquad \sum\limits_{i=1}^{N}\sum\limits_{j=1}^{N}R_{(\text{ii})}\leq C N^2\abs{S} q^{-d},
\end{equation*}
so we need $N^2\abs{S} q^{-d}\to0$. For the third term, we have by H\"older's and Markov's inequalities
\begin{equation*}\begin{split}
&\E\left[\abs{x_{i,t}x_{j,t}}\boldsymbol{1}_{\left\lbrace \abs{x_{i,t}x_{j,t}}>\kappa_T\right\rbrace }\right]
\leq \left(\E\abs{x_{i,t}x_{j,t}}^{m}\right)^{1/m} \left(\frac{\E\abs{x_{i,t}x_{j,t}}^{m}}{\kappa_T^{m}}\right)^{1-1/m} \leq \kappa_T^{1-m} \E\left[ \abs{x_{i,t}x_{j,t}}^{m}\right],
\end{split}\end{equation*}
\begin{equation*}
    \sum\limits_{i=1}^{N}\sum\limits_{j=1}^{N}R_{(\text{iii})}\leq CN^2\abs{S} \kappa_T^{1-m}
\end{equation*}
so we need $N^2\abs{S}\kappa_T^{1-m}\to0$. We then jointly bound all three terms
\begin{equation*}
\begin{split}
&\text{(1)} \quad CN^2q\exp\left(\frac{-T}{\abs{S}^2 q^2\kappa_T^2}\right)\leq \eta_T,\\
&\text{(2)} \quad C N^2\abs{S} q^{-d}\leq \eta_T, \qquad \text{(3)} \quad CN^2\abs{S} \kappa_T^{1-m} \leq \eta_T.
\end{split}
\end{equation*}
by a sequence $\eta_T\to 0$. Note that in the Triplex inequality, $q$ is a positive integer, $\kappa_T>0$, and $\lambda^{-r}s_r>0$ is also satisfied.
We further assume that $\frac{\eta_T}{N^2}\leq\frac{1}{e}\implies \frac{\eta_T}{qN^2}\leq\frac{1}{e}$.  First, isolate $\kappa_T$ in (1), 
\begin{equation*}
CN^2q\exp\left(\frac{-T}{\abs{S}^2 q^2\kappa_T^2}\right)\leq \eta_T \qquad \Longleftrightarrow \quad
\kappa_T\leq C\frac{\sqrt{T}}{\abs{S} q}\frac{1}{\sqrt{\ln\left(qN^2/\eta_T\right)}}.
\end{equation*}
Similarly, isolating $\kappa_T$ from (3), gives
\begin{equation*}
CN^2\abs{S}\kappa_T^{1-m}\leq\eta_T \qquad \Longleftrightarrow \quad \kappa_T\geq C\left(N^2\abs{S}\right)^{\frac{1}{m-1}}\eta_T^{\frac{-1}{m-1}}.
\end{equation*}
Since we have a lower and upper bound on $\kappa_T$, we need to make sure both bounds are satisfied,
\begin{equation*}\begin{split}
&C_1\left(N^2\abs{S}\right)^{\frac{1}{m-1}}\eta_T^{\frac{-1}{m-1}}\leq C_2 \frac{\sqrt{T}}{\abs{S} q}\frac{1}{\sqrt{\ln\left(qN^2/\eta_T\right)}}\\ 
&\quad\Longleftrightarrow\quad q\sqrt{\ln\left(qN^2/\eta_T\right)} \leq C\sqrt{T}\abs{S}^{\frac{-m}{m-1}}N^{\frac{-2}{m-1}}\eta_T^{\frac{1}{m-1}}.
\end{split}\end{equation*}
Isolating $q$ from (2),
\begin{equation*}
C N^2\abs{S} q^{-d} \leq \eta_T \qquad\Longleftrightarrow\quad q\geq CN^{\frac{2}{d}}\abs{S}^{\frac{1}{d}}\eta_T^{\frac{-1}{d}}.
\end{equation*}
Assuming that $\frac{\eta_T}{qN^2}\leq\frac{1}{e}$, we have that $q\leq q\sqrt{\ln\left(qN^2/\eta_T\right)}$ and therefore we need to ensure 
\begin{equation*}
C N^{\frac{2}{d}}\abs{S}^{\frac{1}{d}}\eta_T^{\frac{-1}{d}}\leq C\sqrt{T}\abs{S}^{\frac{-m}{m-1}}N^{\frac{-2}{m-1}}\eta_T^{\frac{1}{m-1}}\qquad\Longleftrightarrow\quad \abs{S}\leq C\eta_T^{\frac{d+m-1}{dm+m-1}}\left[\frac{\sqrt{T}}{N^{\left(\frac{2}{d}+\frac{2}{m-1}\right)}}\right]^{\frac{1}{\frac{1}{d}+\frac{m}{m-1}}}.
\end{equation*}
For the set $S_\lambda$, we have the bound 
\begin{equation*}
    \vert S_\lambda\vert
\leq\sum\limits_{j=1}^N\mathds{1}_{\{\abs{\beta_j^0}>\lambda\}}\left(\frac{\abs{ \beta^0_j}}{\lambda}\right)^r
\leq\lambda^{-r}\sum\limits_{j=1}^N\mathds{1}_{\{\abs{\beta_j^0}>0\}} \abs{\beta^0_j}^r =\lambda^{-r}s_r,
\end{equation*}
and it is sufficient to assume that
\begin{equation*}
    \lambda^{-r}s_r\leq C\eta_T^{\frac{d+m-1}{dm+m-1}}\left[\frac{\sqrt{T}}{N^{\left(\frac{2}{d}+\frac{2}{m-1}\right)}}\right]^{\frac{1}{\frac{1}{d}+\frac{m}{m-1}}}.
\end{equation*}
 When this bound is satisfied, $\sum\limits_{i=1}^{N}\sum\limits_{j=1}^{N}(R_{(\text{i})}+R_{(\text{ii})}+R_{(\text{iii})})\leq 3\eta_T$, and $\P\left(\setCC(S_\lambda)\right)\geq 1-3\eta_T$. 
\end{proof}

\begin{proof}[\bf Proof of Lemma A.4]
By the union bound, Markov's inequality and the mixingale concentration inequality of \cite[Lemma 2]{Hansen91}, it follows that 
\begin{equation*}\begin{split}
&\P\left(\max\limits_{j\leq N,l\leq T} \left[\left\vert \sum\limits_{t=1}^{l}u_t x_{j,t} \right\vert\right] > z\right) \leq\sum\limits_{j=1}^N \P\left(\max\limits_{l\leq T} \left[\left\vert\sum\limits_{t=1}^l u_t x_{j,t} \right\vert\right] > z\right)\\
&\quad\leq z^{-m} \sum\limits_{j=1}^N \E\left[\max\limits_{l\leq T}\left\vert\sum\limits_{t=1}^l u_tx_{j,t}\right\vert^m\right]
\leq z^{-m} \sum\limits_{j=1}^N C_1^m\left(\sum\limits_{t=1}^T c_{t} ^2\right)^{m/2} \leq C N T^{m/2} z^{-m},
\end{split}\end{equation*}
as $\{x_{j,t} u_t\}$ is a mixingale of appropriate size by Lemma A.1. %\cref{lma:NEDimpliesdgp}. 
\end{proof}

\begin{proof}[\bf Proof of Lemma A.5]
This result follows directly by Corollary 6.8 in \cite{SHD11}.
\end{proof}

\begin{proof}[\bf Proof of Lemma A.6]
The proof largely follows Theorem 2.2 of \cite{vandeGeer2016book} applied to $\beta=\boldsymbol{\beta}^0$ with some modifications. For the sake of clarity and readability, we include the full proof here.
Consider two cases. First, consider the case where $\frac{\Vert\boldsymbol{X}(\hat{\boldsymbol{\beta}}-{\boldsymbol{\beta}}^0)\Vert_2^2}{T}<-\frac{\lambda}{4}\Vert\hat{\boldsymbol{\beta}}-{\boldsymbol{\beta}}^0\Vert_1+2\lambda\Vert{\boldsymbol{\beta}}^0_{S^c}\Vert_1$. Then 
\begin{equation*}
\frac{\Vert\boldsymbol{X}(\hat{\boldsymbol{\beta}}-{\boldsymbol{\beta}}^0)\Vert_2^2}{T}+\frac{\lambda}{4}\Vert\hat{\boldsymbol{\beta}}-{\boldsymbol{\beta}}^0\Vert_1 < 2\lambda\Vert{\boldsymbol{\beta}}^0_{S^c}\Vert_1 < 
\frac{8}{3}\lambda\Vert{\boldsymbol{\beta}}^0_{S^c}\Vert_1+C\lambda^2\vert S\vert,
\end{equation*}
which satisfies Lemma A.6. %\cref{lma:generalOracle}.

Next, consider the case where $\frac{\Vert\boldsymbol{X}(\hat{\boldsymbol{\beta}}-{\boldsymbol{\beta}}^0)\Vert_2^2}{T}\geq-\frac{\lambda}{4}\Vert\hat{\boldsymbol{\beta}}-{\boldsymbol{\beta}}^0\Vert_1+2\lambda\Vert{\boldsymbol{\beta}}^0_{S^c}\Vert_1$.
From the Lasso optimization problem in (3),
%\cref{eq:LassoDef},
we have the Karush-Kuhn-Tucker conditions
$\frac{\boldsymbol{X}'({\boldsymbol{y}}-\boldsymbol{X}\hat{\boldsymbol{\beta}})}{T}=\lambda\hat\kappa,$
where $\hat\kappa$ is the subdifferential of $\Vert\hat{\boldsymbol{\beta}}\Vert_1$.
Premultiplying by $({\boldsymbol{\beta}}^0-\hat{\boldsymbol{\beta}})'$, we get 
\begin{equation*}\begin{split}
\frac{({\boldsymbol{\beta}}^0-\hat{\boldsymbol{\beta}})'\boldsymbol{X}'({\boldsymbol{y}}-\boldsymbol{X}\hat{\boldsymbol{\beta}})}{T} = &\lambda({\boldsymbol{\beta}}^0-\hat{\boldsymbol{\beta}})'\hat\kappa
=\lambda{{\boldsymbol{\beta}}^0}'\hat\kappa-\lambda\Vert\hat{\boldsymbol{\beta}}\Vert_1
\leq\lambda\Vert{\boldsymbol{\beta}}^0\Vert_1-\lambda\Vert\hat{\boldsymbol{\beta}}\Vert_1.
\end{split}\end{equation*}
By plugging in  ${\boldsymbol{y}}=\boldsymbol{X}{\boldsymbol{\beta}}^0+\boldsymbol{u}$, the left-hand-side can be re-written as
$\frac{\Vert \boldsymbol{X}(\hat{\boldsymbol{\beta}}-{\boldsymbol{\beta}}^0)\Vert_2^2}{T}+\frac{\boldsymbol{u}'\boldsymbol{X}({\boldsymbol{\beta}}^0-\hat{\boldsymbol{\beta}})}{T},$
and therefore 
\begin{equation*}\begin{split}
&\frac{\Vert \boldsymbol{X}(\hat{\boldsymbol{\beta}}-{\boldsymbol{\beta}}^0)\Vert_2^2}{T}
\leq \frac{\boldsymbol{u}'\boldsymbol{X}(\hat{\boldsymbol{\beta}}-{\boldsymbol{\beta}}^0)}{T}+\lambda\Vert{\boldsymbol{\beta}}^0\Vert_1-\lambda\Vert\hat{\boldsymbol{\beta}}\Vert_1 \\
&\quad \underset{(1)}{\leq}\frac{1}{T} \norm{ \boldsymbol{u}'\boldsymbol{X}}_{\infty}\Vert \hat{{\boldsymbol{\beta}}}-{\boldsymbol{\beta}}^0\Vert_1+\lambda\Vert{\boldsymbol{\beta}}^0\Vert_1-\lambda\Vert\hat{\boldsymbol{\beta}}\Vert_1\\
&\quad\underset{(2)}{\leq} \frac{\lambda}{4}\Vert \hat{{\boldsymbol{\beta}}}-{\boldsymbol{\beta}}^0\Vert_1+\lambda\Vert{\boldsymbol{\beta}}^0\Vert_1-\lambda\Vert\hat{\boldsymbol{\beta}}\Vert_1 
\underset{(3)}{\leq} \frac{5\lambda}{4} \Vert\hat{\boldsymbol{\beta}}_{S}-{\boldsymbol{\beta}}^0_{S}\Vert_1 - \frac{3\lambda}{4}\Vert\hat{\boldsymbol{\beta}}_{S^c}\Vert_1+ \frac{5\lambda}{4}\Vert{\boldsymbol{\beta}}^0_{S^c}\Vert_1 \\
&\quad\underset{(4)}{\leq}
\frac{5\lambda}{4} \Vert\hat{\boldsymbol{\beta}}_{S}-{\boldsymbol{\beta}}^0_{S}\Vert_1 - \frac{3\lambda}{4} \Vert\hat{\boldsymbol{\beta}}_{S^c} - \boldsymbol{\beta}^0_{S^c}\Vert_1 + 2\lambda \Vert{\boldsymbol{\beta}}^0_{S^c} \Vert_1,
\end{split}\end{equation*}
where 
(1) follows from the dual norm inequality,
(2) from the bound on the empirical process given by $\setEP{T}(T\frac{\lambda}{4})$, 
(3) from the property $\Vert \boldsymbol{\beta}\Vert_1=\Vert \boldsymbol{\beta}_S\Vert_1+\Vert \boldsymbol{\beta}_{S^c}\Vert_1$ with $\beta_{j,S}= \beta_j\mathds{1}_{\{j\in S\}}$, as well as several applications of the triangle inequality, and (4) follows from the fact that $\Vert\hat{\boldsymbol{\beta}}_{S^c}\Vert_1\leq\left[\Vert\hat{\boldsymbol{\beta}}_{S^c}-\boldsymbol{\beta}^0_{S^c}\Vert_1-\Vert\boldsymbol{\beta}^0_{S^c}\Vert_1\right]$.
Note that it follows from the condition $\frac{\Vert\boldsymbol{X}(\hat{\boldsymbol{\beta}}-{\boldsymbol{\beta}}^0)\Vert_2^2}{T}\geq-\frac{\lambda}{4}\Vert\hat{\boldsymbol{\beta}}-{\boldsymbol{\beta}}^0\Vert_1+2\lambda\Vert{\boldsymbol{\beta}}^0_{S^c}\Vert_1$ combined with the previous inequality that $\Vert\hat{\boldsymbol{\beta}}_{S^c}-{\boldsymbol{\beta}}^0_{S^c}\Vert_1\leq 3 \Vert\hat{\boldsymbol{\beta}}_{S}-{\boldsymbol{\beta}}^0_{S}\Vert_1$ such that Lemma A.5 %\Cref{lma:compatibilityApplication}
can be applied.
Adding $\frac{3\lambda}{4} \Vert\hat{\boldsymbol{\beta}}_{S}-{\boldsymbol{\beta}}^0_{S}\Vert_1$ to both sides and re-arranging, we get by applying Lemma A.5 %\Cref{lma:compatibilityApplication}
\begin{equation*}\begin{split}
\frac{4}{3}\frac{\Vert \boldsymbol{X}(\hat{\boldsymbol{\beta}} -{\boldsymbol{\beta}}^0)\Vert_2^2}{T}+\frac{\lambda}{4}\Vert\hat{\boldsymbol{\beta}}-{\boldsymbol{\beta}}^0\Vert_1\leq&\frac{8}{3}\lambda\Vert\hat{\boldsymbol{\beta}}_{S}-{\boldsymbol{\beta}}^0_{S}\Vert_1+\frac{8}{3}\lambda\Vert{\boldsymbol{\beta}}^0_{S^c}\Vert_1\\
\leq&\frac{8}{3}\lambda C\sqrt{\vert S\vert(\hat{{\boldsymbol{\beta}}}-{\boldsymbol{\beta}}^0)'\hat{{\boldsymbol{\Sigma}}} (\hat{{\boldsymbol{\beta}}}-{\boldsymbol{\beta}}^0)}+\frac{8}{3}\lambda\Vert{\boldsymbol{\beta}}^0_{S^c}\Vert_1.
\end{split}\end{equation*}
Using that $2uv \leq u^2+v^2$ with $u=\sqrt{\frac{1}{3}(\hat{{\boldsymbol{\beta}}}-{\boldsymbol{\beta}}^0)'\hat{{\boldsymbol{\Sigma}}} (\hat{{\boldsymbol{\beta}}}-{\boldsymbol{\beta}}^0)}$, $v=\frac{4}{\sqrt{3}}C\lambda\sqrt{\vert S\vert}$, we further bound the 
right-hand-side to arrive at
\begin{equation*}\begin{split}
\frac{4}{3}\frac{\Vert \boldsymbol{X}(\hat{\boldsymbol{\beta}}-{\boldsymbol{\beta}}^0)\Vert_2^2}{T}+\frac{\lambda}{4}\Vert\hat{\boldsymbol{\beta}}-{\boldsymbol{\beta}}^0\Vert_1\leq&\frac{1}{3}\frac{\Vert \boldsymbol{X}(\hat{\boldsymbol{\beta}}-{\boldsymbol{\beta}}^0)\Vert_2^2}{T}+C\lambda^2\vert S\vert+\frac{8}{3}\lambda\Vert{\boldsymbol{\beta}}^0_{S^c}\Vert_1,\\
\end{split}\end{equation*}
from which the result follows.
\end{proof}
	
\begin{proof}[\textbf{Proof of Lemma A.7}]
By Assumption 3
%\Cref{ass:compatibility}
and Lemma A.6,
%\cref{lma:generalOracle},
we have on the set	$\setEP{T}(T\frac{\lambda}{4})\cap\setCC(S_\lambda)$
\begin{equation*}\begin{split}
\frac{\Vert \boldsymbol{X}(\hat{\boldsymbol{\beta}}-{\boldsymbol{\beta}}^0)\Vert_2^2}{T}+\frac{\lambda}{4}\Vert\hat{\boldsymbol{\beta}}-{\boldsymbol{\beta}}^0\Vert_1\leq&C\lambda^2\vert S_\lambda\vert+\frac{8}{3}\lambda\Vert{\boldsymbol{\beta}}^0_{S_{\lambda}^c}\Vert_1.\\
\end{split}\end{equation*}
It follows directly from Assumption 2
%\cref{ass:sparsity} 
that 
\begin{equation*}
    \norm{\boldsymbol{\beta}^0_{S_\lambda^c}}_1
=\sum\limits_{j=1}^N \mathds{1}_{\{0<\abs{\beta_j^0}\leq\lambda\}}\abs{\beta_j^0}
\leq\sum\limits_{j=1}^N\mathds{1}_{\{\abs{\beta_j^0}>0\}} \left(\frac{\lambda}{\abs{\beta_j^0}}\right)^{1-r}\abs{\beta_j^0} =\lambda^{1-r}\sum\limits_{j=1}^N\mathds{1}_{\{\abs{\beta_j^0}>0\}} \abs{\beta^0_j}^r
\leq\lambda^{1-r}s_r.
\end{equation*}
and by arguments in the proof of Lemma A.3,
%\Cref{lma:CovarianceCloseness}, $\abs{S_{\lambda}}\leq\lambda^{-r}s_r$
Plugging these in, we obtain
\begin{align*}
\frac{\Vert \boldsymbol{X}(\hat{\boldsymbol{\beta}}-{\boldsymbol{\beta}}^0)\Vert_2^2}{T}+\frac{\lambda}{4}\Vert\hat{\boldsymbol{\beta}}-{\boldsymbol{\beta}}^0\Vert_1 &\leq C\lambda^2\lambda^{-r}s_r+\frac{8}{3}\lambda\lambda^{1-r}s_r =C\lambda^{2-r}s_{r}.\qedhere
\end{align*}
\end{proof}

\subsection{Proofs of preliminary results Section 4} \label{sec:supplS4}

\begin{proof}[\bf Proof of Lemma B.1]
As $v_{j,t}$ are the projection errors from projecting $x_{j,t}$ on all other $x_{k,t}$, it follows directly that $\E\left[v_{j,t}\right]=0$ and  $\E\left[v_{j,t}x_{k,t}\right]=0$. $L_{\bar m}$-boundedness of $\{v_{j,t}  x_{k,t}\},~\forall j,k$ follows from Assumption 1(i), %\cref{ass:dgp}\ref{ass:dgpStationary}, 
Assumption 4,
%\cref{ass:statandvmoments}, 
and the Cauchy--Schwarz inequality.  
By Theorem 17.8 in \cite{Davidson02}, $\{v_{j,t}\}$ is $L_{2m}$-NED on $\{\boldsymbol{s}_{T,t}\}$ of size 
$-d$. The remainder of the proof follows as in the proof of Lemma A.1.
%\cref{lma:NEDimpliesdgp}.
\end{proof}
\begin{proof}[\bf Proof of Lemma B.2]
It follows by the Cauchy--Schwarz inequality that $\left\lbrace w_{j,t}\right\rbrace$ is $L_{\bar m}$-bounded for all $j=1,\ldots, p$, and from the properties of $\{v_{j,t}\}$ by Theorem 17.9 in \cite{Davidson02} that $\{w_{j,t}\}$ is $L_{m}$-NED of size 
$-d$. Part (i) 
%\ref{ass:CLTmixingales1} 
then follows by Theorem 17.5 in \cite{Davidson02}.
For part (ii), %\ref{ass:CLTsummability},
we adapt the proof of Theorem 17.7 in \cite{Davidson02}. Letting $Y_t=w_{j,t}$ and $X_t=w_{k,t}$, $\E w_{j,t}w_{k,t-l}=\E Y_{t}X_{t-l}$. By the triangle inequality, choosing $q=\left[l/2\right]$, and using $\mathcal{F}_{t-l-q}^{t-l+q}$ as in Definition A.1,
%\cref{def:NED}, 
\begin{equation*}
    \abs{\E Y_{t}X_{t-l}}\leq\abs{\E\left[Y_{t}\left(X_{t-l}-\E\left\lbrace X_{t-l}\vert \mathcal{F}_{t-l-q}^{t-l+q}\right\rbrace\right)\right]}+\abs{\E\left[Y_{t}\E\left(X_{t-l}\vert\mathcal{F}_{t-l-q}^{t-l+q}\right) \right]}.
\end{equation*}
By H\"older's inequality, we can bound the first term
\begin{equation*}
    \abs{\E\left[Y_{t}\left(X_{t-l}-\E\left\lbrace X_{t-l}\vert \mathcal{F}_{t-l-q}^{t-l+q}\right\rbrace\right)\right]}\leq \left(\E\left[\abs{Y_{t+q}}^{\frac{m}{m-1}}\right]\right)^{\frac{m-1}{m}}\left(\E\left[\abs{X_{t-l}-\E\left\lbrace X_{t-l}\vert \mathcal{F}_{t-l-q}^{t-l+q}\right\rbrace}^{m}\right]\right)^{\frac{1}{m}}.
\end{equation*}
Since $\frac{m}{m-1}<m<\bar{m}$, $\left(\E\left[\abs{Y_{t+q}}^{\frac{m}{m-1}}\right]\right)^{\frac{m-1}{m}}\leq C$, and since $X_{t-l}$ is NED of size $-d$,\\ $\left(\E\left[\abs{X_{t-l}-\E\left\lbrace X_{t-l}\vert \mathcal{F}_{t-l-q}^{t-l+q}\right\rbrace}^{m}\right]\right)^{\frac{1}{m}}\leq C \psi_q$, where $\psi_q=O(q^{-d-\epsilon})$ for some $\epsilon>0$. For the second term, we use the tower property and H\"older's inequality again
\begin{equation*}\begin{split}
    \abs{\E\left[Y_{t}\E\left(X_{t-l}\vert\mathcal{F}_{t-l-q}^{t-l+q}\right) \right]}=&\abs{\E\left[\E\left(Y_{t}\vert\mathcal{F}_{t-l-q}^{t-l+q}\right)\E\left(X_{t-l}\vert\mathcal{F}_{t-l-q}^{t-l+q}\right) \right]}\\
    \leq& \left(\E\left[\abs{\E\left(Y_{t}\vert\mathcal{F}_{t-l-q}^{t-l+q}\right)}^m\right]\right)^{\frac{1}{m}}\left(\E\left[\abs{\E\left(X_{t-l}\vert\mathcal{F}_{t-l-q}^{t-l+q}\right)}^\frac{m}{m-1}\right]\right)^{\frac{m-1}{m}}.
\end{split}\end{equation*}
Since conditioning is a contractionary projection in $L_p$ spaces,
\begin{equation*}\begin{split}
    &\left(\E\left[\abs{\E\left(Y_{t}\vert\mathcal{F}_{t-l-q}^{t-l+q}\right)}^m\right]\right)^{\frac{1}{m}}\leq \left(\E\left[\abs{\E\left(Y_{t}\vert\mathcal{F}_{-\infty}^{t-l+q}\right)}^m\right]\right)^{\frac{1}{m}}\\
    &\left(\E\left[\abs{\E\left(X_{t-l}\vert\mathcal{F}_{t-l-q}^{t-l+q}\right)}^\frac{m}{m-1}\right]\right)^{\frac{m-1}{m}}\leq \left(\E\left[\abs{X_{t-l}}^\frac{m}{m-1}\right]\right)^{\frac{m-1}{m}}\leq C.
\end{split}\end{equation*}
Since $Y_{t}$ is a Mixingale of size $-d$, the first term can be bounded by $C\psi_{q-l}$, where $\psi_{q-l}= O((q-l)^{-d-\epsilon})$. The sequence $\phi_l$ is then obtained by recalling that we chose $q=[l/2]$, $\phi_l=O((l/2)^{-d-\epsilon})=O(l^{-d-\epsilon})$.  Absolute summability follows by properties of $p$-series, since $d\geq 1$. Note this results also holds for $\max\limits_{q\leq j,k\leq N,\ 1\leq t\leq T}\abs{\E\left[w_{j,t}w_{k,t-l}\right]}$ since $C$ and $\phi_{l}$ are independent of $j$, $k$, and $t$.
(iii)
%\ref{ass:CLTmixingales2} 
follows by repeated application of Corollary 17.11 and Theorem 17.5 in \cite{Davidson02}, noting that $\E(w_{j,t}w_{k,t-l})$ is a non-random and bounded, so trivially NED.
\end{proof}

\begin{proof}[\bf Proof of Lemma B.3] By Lemma A.3,
%\cref{lma:CovarianceCloseness},
$\P\left(\setCC(S_\lambda)\right)\geq 1-3\eta_T$ when 
\begin{equation*}
     \lambda^{-r}s_{r}\leq C\eta_T^{\frac{d+m-1}{dm+m-1}}\left[\frac{\sqrt{T}}{N^{\left(\frac{2}{d}+\frac{2}{m-1}\right)}}\right]^{\frac{1}{\frac{1}{d}+\frac{m}{m-1}}},
\end{equation*} 
 for a sequence $\eta_T\to0$ such that $\eta_T\leq\frac{N^2}{e}$.
We can similarly apply this lemma to the sets $\setCC(S_{\lambda, j})$; when
\begin{equation*}
    \lambda_{j}^{-r}s_{r,j}\leq C\eta_T^{\frac{d+m-1}{dm+m-1}}\left[\frac{\sqrt{T}}{N^{\left(\frac{2}{d}+\frac{2}{m-1}\right)}}\right]^{\frac{1}{\frac{1}{d}+\frac{m}{m-1}}},
\end{equation*}
$\P\left(\setCC(S_{\lambda,j})\right)\geq 1-3\eta_T$.  By the union bound,  $\P\left(\setCC(S_\lambda)\bigcap\limits_{j\in H}\setCC(S_{\lambda,j})\right)\geq 1-\left[1-\P\left(\setCC(S_\lambda)\right)\right]-\sum\limits_{j\in H}\left[1-\P\left(\setCC(S_{\lambda,j})\right)\right]\geq 1-3(1+h)\eta_T$, when the conditions above hold for all $j\in H$. These conditions are then jointly satisfied by the conditions this lemma, which are expressed in terms of $s_{r,\max}$ and $\lambda_{\min}$.
\end{proof}

\begin{proof}[\bf Proof of Lemma B.4]
By Lemmas A.4 and B.1,
we have $\P\left(\setEP{T}^{(j)}(x_j)\right) \leq CN(\sqrt{T}/x_j)^{m}$. Then
\begin{equation*}
\P\left(\bigcap\limits_{j\in H} \setEP{T}^{(j)}(x_j) \right)
\geq 1-\sum\limits_{j\in H}\P\left(\left\lbrace\setEP{T}^{(j)} x_j \right\rbrace^c\right)
\geq1-C\frac{h N T^{m/2}}{\min\limits_{j \in H} x_j^m}. \qedhere
\end{equation*}
\end{proof}

\begin{proof}[\bf Proof of Lemma B.5]
Note that
\begin{equation*}\begin{split}
\P(\setLL)&=\P\left(\bigcap\limits_{j\in H}\left\lbrace\abs{\frac{1}{T} \sum\limits_{t=1}^T v_{j,t}^2-\tau_j^2}\leq \frac{h} {\delta_T}\right\rbrace\right)=1-\P\left(\bigcup\limits_{j\in H}\left\lbrace\abs{\frac{1}{T} \sum\limits_{t=1}^T v_{j,t}^2-\tau_j^2}> \frac{h} {\delta_T}\right\rbrace\right)\\
&\geq 1-\sum\limits_{j\in H}\P\left(\abs{\frac{1}{T} \sum\limits_{t=1}^T v_{j,t}^2-\tau_j^2}> \frac{h} {\delta_T}\right).
\end{split}\end{equation*}
Recalling that $\tau_j^2=\frac{1}{T}\sum\limits_{t=1}^T\E\left[v_{j,t}^2\right]$, write $\P\left(\abs{\frac{1}{T} \sum\limits_{t=1}^T v_{j,t}^2-\tau_j^2}> \frac{h} {\delta_T}\right)=\P\left(\abs{ \sum\limits_{t=1}^T (v_{j,t}^2-\E v_{j,t}^2)}> T\frac{h} {\delta_T}\right)$. As in the proof of Lemma A.3, 
%\cref{lma:CovarianceCloseness},
we use the Triplex inequality to bound this probability. 
\begin{equation*}\begin{split}
&\P\left(\abs{ \sum\limits_{t=1}^T (v_{j,t}^2-\E v_{j,t}^2)}> T\frac{h} {\delta_T}\right)\leq 2q \exp\left(-\frac{Th^2}{288q^2\kappa_T^2\delta_T^2}\right)\\
&\quad + 6\frac{\delta_T}{Th}\sum\limits_{t=1}^T\E\left[\abs{\E\left(v_{j,t}^2\vert\mathcal{F}_{t-q}\right)-\E v_{j,t}^2}\right]
+ 15\frac{\delta_T}{Th}\sum\limits_{t=1}^T\E\left[\abs{v_{j,t}^2}\mathds{1}_{\left\lbrace\abs{v_{j,t}^2}>\kappa_T\right\rbrace}\right]\\
    &:=R_{(\text{i})}+R_{(\text{ii})}+R_{(\text{iii})}. 
\end{split}\end{equation*}
For the second term, note by the proof of Lemma B.1
%\cref{ass:nodewisedgp} 
that $\left\lbrace v_{j,t}\right\rbrace$ is $L_{2m}$-NED on $\left\lbrace\boldsymbol{s}_{T,t}\right\rbrace$ of size 
$-d$. By Assumption 4,
%\cref{ass:statandvmoments},
$\left\lbrace v_{j,t}^2\right\rbrace$ is $L_{\bar m}$-bounded, and by Theorem 17.9 of \cite{Davidson02}, it is $L_{m}$-NED on $\left\lbrace\boldsymbol{s}_{T,t}\right\rbrace$ of size 
$-d$. By Theorem 17.5 $\left\lbrace v_{j,t}^2-\E v_{j,t}^2\right\rbrace$ is then an $L_m$-mixingale of size
$-d$. It then follows that $\E\left[\abs{\E\left(v_{j,t}^2\vert\mathcal{F}_{t-q}\right)-\E v_{j,t}^2}\right]\leq c_t\psi_q\leq Cq^{-d}$, and 
\begin{equation*}
    \sum\limits_{j\in H}R_{(\text{ii})}\leq\sum\limits_{j\in H}6\frac{\delta_T}{Th}\sum\limits_{t=1}^{T}Cq^{-d}=C\frac{\delta_T}{q^d}.
\end{equation*}
For the third term, we have by H\"older's and Markov's inequalities
\begin{equation*}\begin{split}
    \E\left[\abs{v_{j,t}^2}\mathds{1}_{\left\lbrace\abs{v_{j,t}^2}>\kappa_T\right\rbrace}\right]\leq
    C\kappa_T^{1-m}.
\end{split}\end{equation*}
and therefore
\begin{equation*}
    \sum\limits_{j\in H}R_{\text{(iii)}}\leq\sum\limits_{j\in H} 15\frac{\delta_T}{Th}\sum\limits_{t=1}^TC\kappa_T^{1-m}=C\frac{\delta_T}{\kappa_T^{m-1}}.
\end{equation*}
We jointly bound all three terms by a sequence $\eta_T\to 0$.
\begin{equation*}
\text{(1)} \quad Cqh \exp\left(-\frac{Th^2} {q^2\kappa_T^2 \delta_T^2}\right) \leq \eta_T,
\qquad \text{(2)} \quad C\frac{\delta_T}{q^d}\leq \eta_T, 
\qquad \text{(3)} \quad C\frac{\delta_T}{\kappa_T^{m-1}}\leq \eta_T.
\end{equation*}

For the steps below, we assume that $\frac{\eta_T}{h}\leq\frac{1}{e}\implies \sqrt{-\ln(\eta_T/(hq))}\geq1$.
Isolate $\kappa_T$ in (1) and (2),
\begin{equation*}
Cqh\exp\left(\frac{-Th^2}{ q^2\kappa_T^2\delta_T^2}\right)\leq \eta_T \Longleftrightarrow \kappa_T\leq C\frac{\sqrt{T}h}{q\delta_T},
\end{equation*}
\begin{equation*}
    C\frac{\delta_T}{\kappa_T^{m-1}}\leq\eta_T \Longleftrightarrow \kappa_T\geq C\left(\frac{\delta_T}{\eta_T}\right)^{1/(m-1)}.
\end{equation*}
Combining both bounds on $\kappa_T$,
\begin{equation*}\begin{split}
C_1\left(\frac{\delta_T}{\eta_T}\right)^{1/(m-1)} &\leq C_2\frac{\sqrt{T}h}{q\delta_T} \qquad \Longleftrightarrow \quad q\leq C\sqrt{T}h\eta_T^{1/(m-1)}\delta_T^{-m/(m-1)}.
\end{split}\end{equation*}
Isolating $q$ from (2), gives
\begin{equation*}
C \delta_T q^{-d}\leq \eta_T \qquad \Longleftrightarrow \quad q\geq C\eta_T^{-1/d}\delta_T^{1/d}.
\end{equation*}
Combining both bounds on $q$,
\begin{equation*}\begin{split}
& C_1\sqrt{T}h\eta_T^{1/(m-1)}\delta_T^{-m/(m-1)}\geq C_2\delta_T^{1/d}\eta_T^{-1/d} \qquad
\Longleftrightarrow \quad \delta_T\leq  C\eta_T^\frac{d+m-1}{dm+m-1}(\sqrt{T}h)^{\frac{1}{1/d+m/(m-1)}}.
\end{split}\end{equation*}
When $\delta_T$ satisfies this upper bound, $\sum\limits_{j\in H}(R_{(\text{i})}+R_{(\text{ii})}+R_{(\text{iii})})\leq 3\eta_T$, and $\P\left(\setLL\right)\geq 1-3\eta_T$, which completes the proof.

\end{proof}

\begin{proof}[\bf Proof of Lemma B.6]
Note that $\hat\tau_j^2$ can be rewritten as follows
\begin{equation}\label{eq:tauexpanded}\begin{split}
\hat{\tau}_j^2=&\frac{\norm{ \boldsymbol{x}_j-\boldsymbol{X}_{-j}{\boldsymbol{\gamma}}^0_j}_2^2}{T}+\frac{\norm{\boldsymbol{X}_{-j}\left(\hat{{\boldsymbol{\gamma}}}_j-{\boldsymbol{\gamma}}^0_j\right)}_2^2}{T} \\
&\quad - \frac{2\left(\boldsymbol{x}_j - \boldsymbol{X}_{-j}{\boldsymbol{\gamma}}^0_j \right)'\boldsymbol{X}_{-j}\left(\hat{{\boldsymbol{\gamma}}}_j-{\boldsymbol{\gamma}}^0_j\right)}{T}+\lambda_j\Vert\hat{{\boldsymbol{\gamma}}}_j\Vert_1\\
&= \frac{1}{T}\sum\limits_{t=1}^Tv_{j,t}^2+\frac{\norm{\boldsymbol{X}_{-j}\left(\hat{{\boldsymbol{\gamma}}}_j-{\boldsymbol{\gamma}}^0_j\right)}_2^2}{T} - 
\frac{2\left(\boldsymbol{x}_j-\boldsymbol{X}_{-j}{\boldsymbol{\gamma}}^0_j \right)'\boldsymbol{X}_{-j}\left(\hat{{\boldsymbol{\gamma}}}_j-{\boldsymbol{\gamma}}^0_j\right)}{T}+\lambda_j\Vert\hat{{\boldsymbol{\gamma}}}_j\Vert_1.
\end{split}\end{equation}
Then
\begin{equation*}\begin{split}
\vert\hat\tau_j^2-\tau_j^2\vert
&\leq\left\vert\frac{1}{T}\sum\limits_{t=1}^Tv_{j,t}^2-\tau_j^2\right\vert + \frac{\norm{\boldsymbol{X}_{-j}\left(\hat{{\boldsymbol{\gamma}}}_j-{\boldsymbol{\gamma}}^0_j\right)}_2^2}{T}\\
&\quad +\frac{2\left\vert\left(\boldsymbol{x}_j-\boldsymbol{X}_{-j}{\boldsymbol{\gamma}}^0_j \right)'\boldsymbol{X}_{-j} \left(\hat{{\boldsymbol{\gamma}}}_j-{\boldsymbol{\gamma}}^0_j\right) \right\vert}{T}+ \lambda_j\Vert\hat{{\boldsymbol{\gamma}}}_j\Vert_1\\
&=: R_{(\text{i})}+R_{(\text{ii})}+R_{(\text{iii})}+R_{(\text{iv})}.
\end{split}\end{equation*}
By the set $\setLL$, we have
$R_{(\text{i})}\leq\max\limits_{j\in H}\abs{\frac{1}{T}\sum\limits_{t=1}^Tv_{j,t}^2-\tau_j^2}\leq \boundCT$.		
By Corollary 1
%\cref{cor:separateResults} 
applied to the nodewise regression, it holds that
$R_{(\text{ii})}\leq C_1\lambda_j^{2-r}{s}_{r}^{(j)}\leq C_1\bar{\lambda}^{2-r}\bar{{s}_{r}}$.
By the set $\bigcap\limits_{j\in H}\{\setEP{T}^{(j)}(T\frac{\lambda_j}{4})\}$ and the same error bound, we have
\begin{equation*}\begin{split}
R_{(\text{iii})} =& \frac{2\left\vert\boldsymbol{v}_j'\boldsymbol{X}_{-j}\left(\hat{{\boldsymbol{\gamma}}}_j-{\boldsymbol{\gamma}}^0_j\right) \right\vert}{T}
\leq C_2\lambda_j\left\Vert\hat{{\boldsymbol{\gamma}}}_j-{\boldsymbol{\gamma}}^0_j\right\Vert_1
\leq C_2\bar{\lambda}^{2-r}\smaxN.
\end{split}\end{equation*}
By the triangle inequality
$R_{(\text{iv})}\leq\lambda_j\Vert{{\boldsymbol{\gamma}}}^0_j\Vert_1+\lambda_j\Vert\hat{{\boldsymbol{\gamma}}}_j-{\boldsymbol{\gamma}}^0_j\Vert_1.$ Using the weak sparsity index for the nodewise regressions $S_{\lambda,j}=\{k\neq j:\vert\gamma_{j,k}\vert>\lambda_j\}$, write $\Vert{\boldsymbol{\gamma}}^0_j\Vert_1=\left\Vert({{\boldsymbol{\gamma}}^0_j})_{S^c_{\lambda,j}}\right\Vert_1+\left\Vert({{\boldsymbol{\gamma}}^0_j})_{S_{\lambda,j}}\right\Vert_1.$
These terms can then be bounded as follows
\begin{equation*}\begin{split}
\left\Vert({{\boldsymbol{\gamma}}^0_j})_{S^c_{\lambda,j}}\right\Vert_1=&\sum\limits_{k\neq j}\mathds{1}_{\{\vert\gamma^0_{j,k}\vert\leq\lambda_j\}}\vert\gamma^0_{j,k}\vert \leq\lambda_{j}^{1-r}{s}_{r}^{(j)}\leq\bar{\lambda}^{1-r}\smaxN.
\end{split}\end{equation*}
Bounding the $L_1$ norm by the $L_2$ norm, we get
\begin{equation*}\begin{split}
\left\Vert({{\boldsymbol{\gamma}}^0_j})_{S_{\lambda,j}}\right\Vert_1^2\leq&\vert S_{\lambda,j}\vert\Vert{\boldsymbol{\gamma}}^0_j\Vert_2^2 \leq\underset{\bar{}}{\lambda}^{-r}\smaxN\Vert{\boldsymbol{\gamma}}^0_j\Vert_2^2,
\end{split}\end{equation*} 
To further bound $\Vert{\boldsymbol{\gamma}}^0_j\Vert_2^2$, consider the matrix ${\boldsymbol{\Theta}}={\boldsymbol{\Sigma}}^{-1}=\left(\frac{1}{T}\sum_{t=1}^T\E\left[\boldsymbol{x}_t\boldsymbol{x}_t'\right]\right)^{-1}$ and the partitioning 
\begin{equation*}
{\boldsymbol{\Sigma}}=\begin{bmatrix}
\frac{1}{T}\sum_{t=1}^T\E\left(x_{j,t}^2\right) & \frac{1}{T}\sum_{t=1}^T\E\left(x_{j,t}\boldsymbol{x}_{-j,t}'\right)\\
\frac{1}{T}\sum_{t=1}^T\E\left(\boldsymbol{x}_{-j,t}x_{j,t}\right) & \frac{1}{T}\sum_{t=1}^T\E\left(\boldsymbol{x}_{-j,t}\boldsymbol{x}_{-j,t}'\right)
\end{bmatrix}.
\end{equation*}
By blockwise matrix inversion, we can write the $j$th row of $\boldsymbol{\Theta}$ as
\begin{equation}\label{eq:rowTheta}
{\boldsymbol{\Theta}}_{j} =\left[\frac{1}{\tau_j^2}, -\frac{1}{\tau_j^2}\frac{1}{T}\sum_{t=1}^T\E\left(x_{j,t}\boldsymbol{x}_{-j,t}'\right)\left[\frac{1}{T}\sum_{t=1}^T\E\left(\boldsymbol{x}_{-j,t}\boldsymbol{x}_{-j,t}'\right)\right]^{-1}\right]=\frac{1}{\tau_j^2}\left[1,({\boldsymbol{\gamma}}^0_j)'\right].
\end{equation}
It then follows that
\begin{equation*}
\Vert{\boldsymbol{\gamma}}^0_j\Vert_2^2=\sum\limits_{k\neq j}(\gamma^0_{j,k})^2\leq 1+\sum\limits_{k\neq j}(\gamma^0_{j,k})^2=\tau^4_j{\boldsymbol{\Theta}}^{}_j{\boldsymbol{\Theta}}_j'\leq\frac{\tau^4_j}{\Lambda_{\min}^2},
\end{equation*}
as $\frac{1}{\Lambda_{\min}}$ is the largest eigenvalue of ${\boldsymbol{\Theta}}$. For a bound on $\tau_j^2$, by the definition of ${\boldsymbol{\gamma}}^0_j$ from (7)
%\cref{eq:populationGammaj}
and Assumption 5(ii), 
%\cref{ass:nodewise}\ref{ass:nodewiseCompatibility},
it follows that
\begin{equation*}\begin{split}
\tau^2_j=&\min\limits_{{\boldsymbol{\gamma}}_j}\left\lbrace \E\left[\frac{1}{T}\sum_{t=1}^T\left(x_{j,t}-\boldsymbol{x}'_{-j,t}{\boldsymbol{\gamma}}_j\right)^2\right]\right\rbrace\\
&\leq \E\left[\frac{1}{T}\sum_{t=1}^T\left(x_{j,t}-\boldsymbol{x}_{-j,t}'\boldsymbol{0}\right)^2\right]=\frac{1}{T}\sum_{t=1}^T\E\left[x_{j,t}^2\right]={\Sigma}_{j,j} \leq 
C.
\end{split}\end{equation*}
Similar arguments can be used to  bound $\tau_j^2$ from below. By the proof of Lemma 5.3 in \cite{vandeGeer14}, $\tau_j^2=\frac{1}{\Theta_{j,j}}$, and therefore $\tau_j^2 \geq \Lambda_{\min}$. It then follows from Assumption 5(ii)
%\cref{ass:nodewise}\ref{ass:nodewiseCompatibility} 
that
\begin{equation*}
\frac{1}{C} \leq \tau^2_j \leq C,\text{ uniformly over } j\in 1,\dots,N.
\end{equation*} 
We therefore have $\Vert{\boldsymbol{\gamma}}^0_j\Vert_2\leq\frac{\tau_j^2}{\Lambda_{\min}}\leq C^2,$ such that we can bound the fourth term as
\begin{equation*}\begin{split}
R_{(\text{iv})}&\leq\lambda_j\Vert{{\boldsymbol{\gamma}}}^0_j\Vert_1+\lambda_j\Vert\hat{{\boldsymbol{\gamma}}}_j-{\boldsymbol{\gamma}}^0_j\Vert_1
=\lambda_j\left\Vert({{\boldsymbol{\gamma}}^0_j})_{S^c_{\lambda,j}}\right\Vert_1+\lambda_j\left\Vert({{\boldsymbol{\gamma}}^0_j})_{S_{\lambda,j}}\right\Vert_1+\lambda_j\Vert\hat{{\boldsymbol{\gamma}}}_j-{\boldsymbol{\gamma}}^0_j\Vert_1\\
&\leq \bar{\lambda}^{2-r}\bar{s}_r+\bar{\lambda}\sqrt{\underset{\bar{}}{\lambda}^{-r}\smaxN}C_1^2+C_2\bar{\lambda}^{2-r}\smaxN
\end{split}\end{equation*}
Combining all bounds, we have 
\begin{equation*}\begin{split}
\vert\hat\tau_j^2-\tau_j^2\vert\leq& \boundCT+C_1\bar{\lambda}^{2-r}\bar{{s}_{r}}+C_2\bar{\lambda}^{2-r}\bar{{s}_{r}}+\bar{\lambda}^{2-r}\smaxN+\sqrt{\bar{\lambda}^2\underset{\bar{}}{\lambda}^{-r}\smaxN}C_3^2+C_4\bar\lambda^{2-r}\smaxN\\
=&\boundCT +C_5\bar{\lambda}^{2-r}\smaxN+C_6\sqrt{\bar{\lambda}^2\underset{\bar{}}{\lambda}^{-r}\smaxN}.
\end{split}\end{equation*}
For the second statement in Lemma B.6, 
%\cref{lma:hatTauConsistency}, 
we have by the triangle inequality and (B.1)
%\cref{eq:trueTauBounded}
that 
\begin{align*}
\left\vert\frac{1}{\hat\tau_j^2}-\frac{1}{\tau_j^2}\right\vert &\leq 
\frac{\vert\hat\tau_j^2-\tau_j^2\vert}{\tau_j^4-\tau^2_j\vert\hat\tau_j^2-\tau_j^2\vert} \leq
\frac{\vert\hat\tau_j^2-\tau_j^2\vert}{\frac{1}{C^2}-C\vert\hat\tau_j^2-\tau_j^2\vert}\\
&\leq \frac{ \boundCT +C_5\bar{\lambda}^{2-r}\smaxN+C_6\sqrt{\bar{\lambda}^2\underset{\bar{}}{\lambda}^{-r}\smaxN}}{C_7-C_8\left( \boundCT +C_5\bar{\lambda}^{2-r}\smaxN+C_6\sqrt{\bar{\lambda}^2\underset{\bar{}}{\lambda}^{-r}\smaxN}\right)}.  \qedhere
\end{align*}
\end{proof}

\begin{proof}[\bf Proof of Lemma B.7]
First, note that since $\hat{\boldsymbol{\Sigma}}$ is a symmetric matrix
\begin{equation*}
\max\limits_{j\in H}\left\lbrace\Vert \boldsymbol{e}'_j-\hat{\boldsymbol{\Theta}}_j\hat{\boldsymbol{\Sigma}}\Vert_{\infty}\right\rbrace=\max\limits_{j\in H}\left\lbrace\Vert \hat{\boldsymbol{\Sigma}}\hat{\boldsymbol{\Theta}}'_j-\boldsymbol{e}_j\Vert_{\infty}\right\rbrace.
\end{equation*} 
By the extended KKT conditions (see Section 2.1.1 of \citealp{vandeGeer14}),
we have that
$\max\limits_{j\in H}\left\lbrace\Vert\hat{{\boldsymbol{\Sigma}}}\hat{{\boldsymbol{\Theta}}}'_{ j}-\boldsymbol{e}_j\Vert_{\infty}\right\rbrace\leq\max\limits_{j\in H}\left\lbrace\frac{\lambda_j}{\hat{\tau}_j^2}\right\rbrace
\leq\frac{\bar{\lambda}}{\min\limits_{j\in H}\left\lbrace\hat\tau_j^2\right\rbrace}$.
For a lower bound on $\min\limits_{j\in H}\left\lbrace\hat\tau_j^2\right\rbrace$, note that by \cref{eq:tauexpanded}, $\hat\tau_j^2$ can be rewritten as
\begin{equation*}
\begin{split}
\hat{\tau}_j^2 &=\frac{\Vert \boldsymbol{x}_j-\boldsymbol{X}_{-j}{\boldsymbol{\gamma}}^0_j\Vert_2^2}{T}+\frac{\Vert\boldsymbol{X}_{-j}\left(\hat{{\boldsymbol{\gamma}}}_j-{\boldsymbol{\gamma}}^0_j\right)\Vert_2^2}{T}
-\frac{2\left(\boldsymbol{x}_j-\boldsymbol{X}_{-j}{\boldsymbol{\gamma}}^0_j \right)'\boldsymbol{X}_{-j}\left(\hat{{\boldsymbol{\gamma}}}_j-{\boldsymbol{\gamma}}^0_j\right)}{T}+\lambda_j\Vert\hat{{\boldsymbol{\gamma}}}_j\Vert_1.
\end{split}
\end{equation*}
With $\frac{\Vert\boldsymbol{X}_{-j}\left(\hat{{\boldsymbol{\gamma}}}_j-{\boldsymbol{\gamma}}^0_j\right)\Vert_2^2}{T}\geq0$ and $\lambda_j\Vert\hat{\boldsymbol{\gamma}}_j\Vert_1\geq0$ by definition for all $j$, we have
\begin{equation*}\begin{split}
\hat{\tau}_j^2&\geq\frac{\Vert \boldsymbol{x}_j-\boldsymbol{X}_{-j}{\boldsymbol{\gamma}}^0_j\Vert_2^2}{T}-\frac{2\left(\boldsymbol{x}_j-\boldsymbol{X}_{-j}{\boldsymbol{\gamma}}^0_j \right)'\boldsymbol{X}_{-j}\left(\hat{{\boldsymbol{\gamma}}}_j-{\boldsymbol{\gamma}}^0_j\right)}{T}
=\frac{\sum\limits_{t=1}^T v_{j,t}^2}{T}-\frac{2\boldsymbol{v}_j '\boldsymbol{X}_{-j}\left(\hat{{\boldsymbol{\gamma}}}_j-{\boldsymbol{\gamma}}^0_j\right)}{T}.
\end{split}\end{equation*}
The dual norm inequality in combination with the triangle inequality then gives
\begin{equation*}
\begin{split}
\hat{\tau}_j^2&\geq\tau_j^2 -\left\vert\frac{1}{T}\sum\limits_{t=1}^T v_{j,t}^2-\tau_j^2\right\vert -\frac{2}{T}\max\limits_{k\neq j}\left\lbrace\vert\boldsymbol{v}_j'\boldsymbol{x}_k\vert\right\rbrace\Vert\hat{{\boldsymbol{\gamma}}}_j-{\boldsymbol{\gamma}}^0_j\Vert_{1},\\
&\geq\frac{1}{C} -\max\limits_{j}\left\lbrace\left\vert\frac{1}{T}\sum\limits_{t=1}^T v_{j,t}^2-\tau_j^2\right\vert\right\rbrace -\frac{2}{T}\max\limits_{k\neq j}\left\lbrace\vert\boldsymbol{v}_j'\boldsymbol{x}_k\vert\right\rbrace\Vert\hat{{\boldsymbol{\gamma}}}_j-{\boldsymbol{\gamma}}^0_j\Vert_{1},
\end{split}
\end{equation*}
where the second line follows from (B.1).
%\cref{eq:trueTauBounded}.
Then, on the sets $\setLL$ and $\setEP{T}^{(j)}(T\frac{\lambda_j}{4})$
\begin{equation*}\begin{split}
\hat{\tau}_j^2\geq C_1 - \boundCT -\frac{\lambda_{j}}{2}\Vert\hat{{\boldsymbol{\gamma}}}_j-{\boldsymbol{\gamma}}^0_j\Vert_{1} \geq C_1 - \boundCT-C_2\lambda_j^{2-r}s_{r}^{(j)}
\geq C_1 -\boundCT -C_2\bar{\lambda}^{2-r}\smaxN,
\end{split}\end{equation*}
where Corollary 1 
%\cref{cor:separateResults} 
yields the second inequality. As $\bar{\lambda}^{2-r}\smaxN \to 0$, for a large enough $T$ we have that
\begin{equation*}\begin{split}
\min_j\frac{1}{\hat\tau_j^2}\leq&\frac{1}{C_1 - \boundCT -C_2\bar{\lambda}^{2-r}\smaxN}
\end{split}\end{equation*}
from which the result follows.
\end{proof}

\begin{proof}[\bf Proof of Lemma B.8] 
Note that the $j$th row of the matrix $I-\hat{\boldsymbol{\Theta}}\hat{\boldsymbol{\Sigma}}$ is $\boldsymbol{e}'_j-\hat{\boldsymbol{\Theta}}_j\hat{\boldsymbol{\Sigma}}$, where $\hat{\boldsymbol{\Theta}}_j$ is the $j$th row of $\hat{\boldsymbol{\Theta}}$.
Plugging in the definition of $\Delta$, we have
\begin{equation*}\begin{split}
\max\limits_{j\in H} \vert\Delta_j\vert=\sqrt{T}\max\limits_{j\in H}\left\vert\left(\boldsymbol{e}'_j-\hat{\boldsymbol{\Theta}}_j\hat{\boldsymbol{\Sigma}}\right)\left(\hat{\boldsymbol{\beta}}-{\boldsymbol{\beta}}^0\right)\right\vert
\leq\sqrt{T}\max\limits_{j\in H}\left\lbrace\Vert \boldsymbol{e}'_j-\hat{\boldsymbol{\Theta}}_j\hat{\boldsymbol{\Sigma}}\Vert_{\infty}\right\rbrace\Vert\hat{\boldsymbol{\beta}}-{\boldsymbol{\beta}}^0\Vert_1.
\end{split}\end{equation*}
By Lemma A.7, 
%\cref{lma:errorBoundonSets}, 
under Assumptions 2 and 5(ii), %\cref{ass:sparsity,ass:nodewise}\ref{ass:nodewiseCompatibility},
on the sets $\setEP{T}(T\frac{\lambda}{4})\cap \setCC(S_\lambda)$, we have
\begin{equation}\label{eq:improvedInitialOracle}
\frac{\Vert\boldsymbol{X}(\hat{{\boldsymbol{\beta}}}-{\boldsymbol{\beta}}^0)\Vert_2^2}{T}+\lambda\Vert\hat{{\boldsymbol{\beta}}}-{\boldsymbol{\beta}}^0\Vert_1\leq C\lambda^{2-r}s_{r},
\end{equation}
from which it follows that
$\Vert\hat{\boldsymbol{\beta}}-{\boldsymbol{\beta}}^0\Vert_1\leq C\lambda^{1-r}{s}_{r}.$
Combining this bound with Lemma B.7
%\cref{lma:inverseQuality}
gives
\begin{align*}
\max\limits_{j\in H} \vert\Delta_j\vert\leq&\sqrt{T}\lambda^{1-r}{s}_{r}\frac {\bar{\lambda}}{C_1-\boundCT -C_2\bar{\lambda}^{2-r}\smaxN}. \qedhere
\end{align*}
\end{proof}

\begin{proof}[\bf Proof of Lemma B.9]
Starting from the nodewise regression model, write
\begin{equation*}\begin{split}
\frac{1}{\sqrt{T}}\left\vert\hat{\boldsymbol{v}}_j'\boldsymbol{u}-\boldsymbol{v}_j'\boldsymbol{u}\right\vert
=\frac{1}{\sqrt{T}}\left\vert\boldsymbol{u}'\boldsymbol{X}_{-j}\left({\boldsymbol{\gamma}}^0_j-\hat{\boldsymbol{\gamma}}_j\right)\right\vert
\leq\frac{1}{\sqrt{T}} \norm{\boldsymbol{u}'\boldsymbol X}_{\infty} \norm{\hat{\boldsymbol{\gamma}}_j-{\boldsymbol{\gamma}}^0_j}_1.
\end{split}\end{equation*}
By the set $\setEP{T}(T\lambda)$ and Corollary 1,
%\cref{cor:separateResults},
\begin{align*}
\sqrt{T}\frac{\max\limits_{j}\left\lbrace\left\vert\boldsymbol{u}'X_{j}\right\vert\right\rbrace}{T}\left\Vert\hat{\boldsymbol{\gamma}}_j-{\boldsymbol{\gamma}}^0_j\right\Vert_1\leq&\sqrt{T}\lambda\left\Vert\hat{\boldsymbol{\gamma}}_j-{\boldsymbol{\gamma}}^0_j\right\Vert_1
\leq C\sqrt{T}\lambda \lambda_j^{1-r}{s}_{r}^{(j)}
\leq C\sqrt{T}\lambda_{\max}^{2-r}\smaxN, 
\end{align*}
where the upper bound is uniform over $j\in H$.
\end{proof}
	
\begin{proof}[\bf Proof of Lemma B.10]
By the union bound
\begin{equation*}
\P\left(\bigcap\limits_{j\in H}\left\lbrace\max\limits_{s\leq T}\abs{
\sum\limits_{t=1}^{s}v_{j,t}u_t}\leq x\right\rbrace\right)
\geq 1-\sum\limits_{j\in H}\P\left(\max\limits_{s\leq T}\abs{
\sum\limits_{t=1}^{s}v_{j,t}u_t}> x\right).
\end{equation*}
By the Markov inequality, Lemma B.2
%\cref{ass:CLT}
and the mixingale concentration inequality of \cite[Lemma 2]{Hansen91},
\begin{equation*}\begin{split}
\P\left(\max\limits_{ s\leq T}\abs{ \sum\limits_{t=1}^{s}v_{j,t}u_t}> x\right)
\leq\frac{\E\left(\max\limits_{ s\leq T}\abs{\sum\limits_{t=1}^{s}v_{j,t}u_t}^{m}\right)}{ x^{m}}
\leq
\frac{C_1^{m}\left(\sum\limits_{t=1}^T\left(c_{t}^{(j)}\right)^2\right)^{m/2}}{  x^{m}}
=\frac{C T^{m/2}
}{x^{m}},
\end{split}\end{equation*}
from which the result follows.
\end{proof}

\begin{proof}[\bf Proof of Lemma B.11]
Start by writing
\begin{equation*}\begin{split}
&\left\vert\frac{1}{\sqrt{T}}\frac{\hat{\boldsymbol{v}}_{j}'\boldsymbol{u}}{\hat\tau_j^2}-\frac{1}{\sqrt{T}}\frac{\boldsymbol{v}_{j}'\boldsymbol{u}}{\tau_j^2}\right\vert
\leq
\frac{1}{\sqrt{T}}\left\vert\frac{\left(\hat{\boldsymbol{v}}_j'\boldsymbol{u}-\boldsymbol{v}_j'\boldsymbol{u}\right)}{\hat\tau_j^2}\right\vert+
\left\vert\frac{1}{\hat\tau_j^2}-\frac{1}{\tau_j^2}\right\vert\left\vert\frac{\boldsymbol{v}_j'\boldsymbol{u}}{\sqrt{T}}\right\vert
=:R_{(\text{i})} + R_{(\text{ii})}. \\
\end{split}\end{equation*}
For the first term, we can bound from above using Lemmas B.6, B.9 and equation (B.1), %\cref{lma:hatTauConsistency,lma:vuConsistency,eq:trueTauBounded},
all providing bounds uniform over $j\in H$.
We then get 
\begin{equation*}\begin{split}
R_{(\text{i})}
\leq&\frac{\vert\hat{\boldsymbol{v}}_j'\boldsymbol{u}-\boldsymbol{v}_j'\boldsymbol{u}\vert}{\sqrt{T}}\frac{1}{\vert\tau_j^2\vert-\vert\hat\tau_j^2-\tau_j^2\vert}
\leq\frac{C_5\sqrt{T}\lambda_{\max}^{2-r}\smaxN}{1/C_6-\left(\boundCT +C_1\bar{\lambda}^{2-r}\smaxN+C_2\sqrt{\bar{\lambda}^2\underset{\bar{}}{\lambda}^{-r}\smaxN}\right)}.
\end{split}\end{equation*}
For the second term, we can bound from above using Lemma B.6
%\cref{lma:hatTauConsistency} 
and the set $\bigcap\limits_{j\in H}\setEPvuj{T}(h^{1/m}T^{1/2}\eta_T^{-1})$ to get the uniform bound 
\begin{equation*}\begin{split}
R_{(\text{ii})}
\leq&\frac{ h^{1/m}\eta_{T}^{-1}\boundCT +C_7\bar{\lambda}^{2-r}\smaxN h^{1/m}\eta_T^{-1} + C_8\sqrt{\bar{\lambda}^2\underset{\bar{}}{\lambda}^{-r}\smaxN} h^{1/m}\eta_T^{-1}} {C_9-C_{10}\left( \boundCT +C_1\bar{\lambda}^{2-r}\smaxN+C_2\sqrt{\bar{\lambda}^2\underset{\bar{}}{\lambda}^{-r}\smaxN}\right)}. 
\end{split}\end{equation*}
Combining both bounds gives 
\begin{equation*}\begin{split}
&R_{(\text{i})} + R_{(\text{ii})}\leq \frac{h^{1/m}\eta_T^{-1}\boundCT+C_1h^{1/m}\eta_T^{-1}\sqrt{T}\lambda_{\max}^{2-r}\smaxN+C_2h^{1/m}\eta_T^{-1}\sqrt{\bar{\lambda}^2\underset{\bar{}}{\lambda}^{-r}\smaxN}}{C_3-C_4\left(\boundCT +C_1\bar{\lambda}^{2-r}\smaxN+C_2\sqrt{\bar{\lambda}^2\underset{\bar{}}{\lambda}^{-r}\smaxN}\right)}
\end{split}\end{equation*}
from which the result follows.
\end{proof}

\begin{proof}[\bf Proof of Lemma B.12]
The result follows directly from the Markov inequality
\begin{equation*}
\P \left(\norm{\boldsymbol{d}}_{\infty} > x \right) \leq x^{-p} \E \left[\max_t \abs{d_t}^p \right] \leq x^{-p} T \max_t \E \abs{d_t}^p \leq C x^{-p} T.\qedhere
\end{equation*}
\end{proof}

\begin{proof}[\bf Proof of Lemma B.13]
We can write
\begin{equation*}
\begin{split}
&\abs{\frac{1}{T}\sum_{t=l+1}^{T} \left(\hat{w}_{j,t}\hat{w}_{k,t-l} - w_{j,t} w_{k,t-l}\right)} \leq \abs{\frac{1}{T} \sum_{t=l+1}^{T} \left(\hat{w}_{j,t} - w_{j,t}\right) \left(\hat{w}_{k,t-l} - w_{k,t-l} \right)} \\
&\quad\quad + \abs{\frac{1}{T} \sum_{t=l+1}^{T} \left(\hat{w}_{j,t} - w_{j,t}\right) w_{k,t-l}} + \abs{\frac{1}{T}\sum_{t=l+1}^{T} w_{j,t} \left(\hat{w}_{k,t-l} - w_{k,t-l}\right)} \\
&\qquad=: \frac{1}{T} \left[R_{(\text{i})} + R_{(\text{ii})} + R_{(\text{iii})} \right].
\end{split}
\end{equation*}

Take $R_{(\text{i})}$ first. Using that $\hat{w}_{j,t-q} = \hat{u}_{t-q} \hat{v}_{j,t-q}$, straightforward but tedious calculations show that
\begin{equation*}
\begin{split}
&R_{(\text{i})} \leq \abs{\sum_{t=l+1}^T \left(\hat{u}_{t} - u_{t} \right) \left(\hat{u}_{t-l} - u_{t-l} \right) \left(\hat{v}_{j,t} - v_{j,t} \right) \left(\hat{v}_{k,t-l} - v_{k,t-l}\right)} \\
&\quad + \abs{\sum_{t=l+1}^T \left(\hat{u}_{t} - u_{t} \right) \left(\hat{u}_{t-l} - u_{t-l} \right) \left(\hat{v}_{j,t} - v_{j,t} \right)  v_{k,t-l}} + \abs{ \sum_{t=l+1}^T \left(\hat{u}_{t} - u_{t} \right) u_{t-l} \left(\hat{v}_{j,t} - v_{j,t} \right) \left(\hat{v}_{k,t-l} - v_{k,t-l}\right)} \\
&\quad + \abs{\sum_{t=l+1}^T \left(\hat{u}_{t} - u_{t} \right) \left(\hat{u}_{t-l} - u_{t-l} \right) v_{j,t} \left(\hat{v}_{k,t-l} - v_{k,t-l}\right)} + \abs{\sum_{t=l+1}^T \left(\hat{u}_{t} - u_{t} \right) \left(\hat{u}_{t-l} - u_{t-l} \right) v_{j,t} v_{k,t-l}} \\
&\quad + \abs{\sum_{t=l+1}^T \left(\hat{u}_{t} - u_{t} \right) u_{t-l} v_{j,t} \left(\hat{v}_{k,t-l} - v_{k,t-l}\right)} + \abs{\sum_{t=l+1}^T u_{t}  \left(\hat{u}_{t-l} - u_{t-l} \right) \left(\hat{v}_{j,t} - v_{j,t} \right) \left(\hat{v}_{k,t-l} - v_{k,t-l}\right)} \\
&\quad + \abs{\sum_{t=l+1}^T u_{t} \left(\hat{u}_{t-l} - u_{t-l} \right) \left(\hat{v}_{j,t} - v_{j,t} \right) v_{k,t-l}} + \abs{\sum_{t=l+1}^T u_{t} u_{t-l} \left(\hat{v}_{j,t} - v_{j,t} \right) \left(\hat{v}_{k,t-l} - v_{k,t-l}\right)} 
=: \sum_{i=1}^9 R_{\text{(i)},i}.
\end{split}
\end{equation*}
Using that $\norm{\hat{\boldsymbol{v}}_{j} - \boldsymbol{v}_{j}}_2 = \norm{\boldsymbol{X}_{-j} \left(\hat{\boldsymbol{\gamma}}_0 - \boldsymbol{\gamma}_j^0 \right)}_2 \leq C \sqrt{T \bar{\lambda}^{2-r} \smaxN}$ on the set $\setNWcons$ by Corollary 1, 
%\cref{cor:separateResults},
and $\norm{\hat{\boldsymbol{u}} - \boldsymbol{u}}_2 = \norm{\boldsymbol{X} \left(\hat{\boldsymbol{\beta}} - \boldsymbol{\beta}^0 \right)}_2 \leq C \sqrt{T \lambda^{2-r} s_r}$ on the set $\setILcons$ by Corollary 1,
%\cref{cor:separateResults},
we can use the Cauchy--Schwarz inequality to conclude that
\begin{equation*}
R_{\text{(i)},1} \leq \norm{\hat{\boldsymbol{u}} - \boldsymbol{u}}_2^2 \norm{\hat{\boldsymbol{v}}_j - \boldsymbol{v}_j}_2 \norm{\hat{\boldsymbol{v}}_k - \boldsymbol{v}_k}_2 \leq C T^2 \lambda^{2-r} s_r \bar{\lambda}^{2-r} \smaxN \leq C T^2 \left[\lambda_{\max}^{2-r} s_{r,\max}\right]^2.
\end{equation*}

On the set $\setTail{u}{T^{1/2m}} \bigcap\limits_{j\in H} \setTail{v_{j}}{T^{1/2m}}$, we have that $\norm{\boldsymbol{u}}_{\infty}\leq C T^{1/2m}$, and \\  $\norm{\boldsymbol{v}_j}_{\infty} \leq C (hT)^{1/2m}$, uniformly over $j\in H$. Then we can use this, plus the previous results to find that
\begin{equation*}
\begin{split}
R_{\text{(i)},2} 
&\leq \norm{\boldsymbol{v}_k}_{\infty} \sum_{t=l+1}^T \abs{\hat{u}_{t} - u_{t}} \abs{\hat{u}_{t-l} - u_{t-l} } \abs{\hat{v}_{j,t} - v_{j,t}} \\
&\leq \norm{\boldsymbol{v}_k}_{\infty} \norm{\hat{\boldsymbol{u}} - \boldsymbol{u}}_2^2 \norm{\hat{\boldsymbol{v}}_j - \boldsymbol{v}_j}_2 \leq C (hT)^{\frac{1}{2m}} T^{3/2} \left[\lambda_{\max}^{2-r} s_{r,\max}\right]^{3/2}.
\end{split}
\end{equation*}
We then find in the same way that
\begin{equation*}
\begin{split}
R_{\text{(i)},3} &\leq \norm{\boldsymbol{u}}_{\infty} \norm{\hat{\boldsymbol{u}} - \boldsymbol{u}}_2 \norm{\hat{\boldsymbol{v}}_j - \boldsymbol{v}_j}_2 \norm{\hat{\boldsymbol{v}}_k - \boldsymbol{v}_k}_2 \leq C T^{\frac{1}{2m}} T^{3/2} \left[\lambda_{\max}^{2-r} s_{r,\max}\right]^{3/2},\\
R_{\text{(i)},4} &\leq \norm{\hat{\boldsymbol{u}} - \boldsymbol{u}}_2^2 \norm{\boldsymbol{v}_j}_{\infty} \norm{\hat{\boldsymbol{v}}_k - \boldsymbol{v}_k}_2 \leq C (hT)^{\frac{1}{2m}} T^{3/2} \left[\lambda_{\max}^{2-r} s_{r,\max}\right]^{3/2},\\
R_{\text{(i)},5} &\leq \norm{\hat{\boldsymbol{u}} - \boldsymbol{u}}_2^2 \norm{\boldsymbol{v}_j}_{\infty}\norm{\boldsymbol{v}_k}_{\infty} \leq C (hT)^{\frac{1}{m}} T \lambda_{\max}^{2-r} s_{r,\max}. 
\end{split}
\end{equation*}
Defining $\tilde{\boldsymbol{w}}_{j,l} = (u_{1} v_{k,l+1}, \ldots, u_{T} v_{j,T})'$, $\tilde{\boldsymbol{w}}_{k,-l} = (u_{l+1} v_{k,1}, \ldots, u_{T} v_{k,T})'$ and $\tilde{\boldsymbol{u}}_{l} = (u_{1} u_{l+1}, \ldots, u_{T} u_{T})'$, all with $\bar m$ bounded moments, we find on the set 
\begin{equation*}
\setTail{u}{T^{1/2m}} \cap \setTail{\tilde{u}_l}{T^{1/m}} \bigcap\limits_{j\in H} \setTail{\tilde{w}_{j,l}}{T^{1/m}} \bigcap\limits_{k\in H} \setTail{\tilde{w}_{k,-l}}{T^{1/m}}
\end{equation*}
that
\begin{equation*}
\begin{split}
R_{\text{(i)},6} &\leq \norm{\tilde{\boldsymbol{w}}_{j,l}}_{\infty}
\norm{\hat{\boldsymbol{u}} - \boldsymbol{u}}_2 \norm{\hat{\boldsymbol{v}}_k - \boldsymbol{v}_k}_2 \leq C (hT)^{\frac{1}{m}} T \lambda_{\max}^{2-r} s_{r,\max}, \\
R_{\text{(i)},7} &\leq \norm{\boldsymbol{u}}_{\infty} \norm{\hat{\boldsymbol{u}} - \boldsymbol{u}}_2 \norm{\hat{\boldsymbol{v}}_j - \boldsymbol{v}_j}_2 \norm{\hat{\boldsymbol{v}}_k - \boldsymbol{v}_k}_2 \leq C T^{\frac{1}{2m}} T \left[\lambda_{\max}^{2-r} s_{r,\max}\right]^{3/2}, \\
R_{\text{(i)},8} &\leq \norm{\tilde{\boldsymbol{w}}_{k,-l}}_{\infty}
\norm{\hat{\boldsymbol{u}} - \boldsymbol{u}}_2 \norm{\hat{\boldsymbol{v}}_j - \boldsymbol{v}_j}_2 \leq C (hT)^{\frac{1}{m}} T \lambda_{\max}^{2-r} s_{r,\max}, \\
R_{\text{(i)},9} &\leq \norm{\tilde{\boldsymbol{u}}_l}_{\infty}^2 \norm{\hat{\boldsymbol{v}}_j - \boldsymbol{v}_j}_2 \norm{\hat{\boldsymbol{v}}_k - \boldsymbol{v}_k}_2 \leq C T^{\frac{1}{m}} T \lambda_{\max}^{2-r} s_{r,\max}. \\
\end{split}
\end{equation*}
It then follows that
\begin{equation*}
\begin{split}
\frac{1}{T} R_{\text{(i)}} &\leq C_1 T\left[\lambda_{\max}^{2-r} s_{r,\max}\right]^2+C_2 h^{1/2m}T^{(m+1)/2m}\left[\lambda_{\max}^{2-r} s_{r,\max}\right]^{3/2}\\
&\quad+C_3h^{1/m}T^{1/m}\lambda_{\max}^{2-r} s_{r,\max}.
\end{split}
\end{equation*}

For $R_{\text{(ii)}}$ we get analogously on the set $\setTail{u}{T^{1/2m}} \bigcap\limits_{j\in H} \setTail{v_j}{(hT)^{1/2m}} \bigcap\limits_{j\in H} \setTail{w_j}{(hT)^{1/m}}$
\begin{equation*}
\begin{split}
R_{\text{(ii)}} &\leq \abs{\frac{1}{T} \sum_{t=l+1}^{T} \left(\hat{u}_{t} - u_{t}\right) \left(\hat{v}_{j,t} - v_{j,t} \right) w_{k,t-l}} \\
&\quad + \abs{\frac{1}{T} \sum_{t=l+1}^{T} \left(\hat{u}_{t} - u_{t}\right) v_{j,t} w_{k,t-l}} + \abs{\frac{1}{T} \sum_{t=l+1}^{T} u_{t} \left(\hat{v}_{j,t} - v_{j,t} \right) w_{k,t-l}} \\
&\leq \norm{\hat{\boldsymbol{u}} - \boldsymbol{u}}_2 \norm{\hat{\boldsymbol{v}}_j - \boldsymbol{v}_j}_2 \norm{\boldsymbol{w}_k}_{\infty} + \norm{\hat{\boldsymbol{u}} - \boldsymbol{u}}_2 \norm{\boldsymbol{v}_j}_{\infty} \norm{\boldsymbol{w}_k}_{\infty} + \norm{\boldsymbol{u}}_{\infty} \norm{\hat{\boldsymbol{v}}_j - \boldsymbol{v}_j}_2 \norm{\boldsymbol{w}_k}_{\infty},\\
&\leq C_1 (hT)^{\frac{1}{m}} T \lambda_{\max}^{2-r} s_{r,\max} + C_2 (hT)^{\frac{3}{2m}} T^{1/2} \sqrt{\lambda_{\max}^{2-r} s_{r,\max}} + C_3h^{\frac{1}{m}} T^{\frac{3}{2m}} T^{1/2} \sqrt{\lambda_{\max}^{2-r} s_{r,\max}}.
\end{split}
\end{equation*}
It then follows that $\frac{1}{T}R_{\text{(ii)}}\leq C_1h^{1/m}T^{1/m}\lambda_{\max}^{2-r} s_{r,\max}+C_2h^{3/2m}T^{(3-m)/2m}\sqrt{\lambda_{\max}^{2-r} s_{r,\max}}$.
Finally, $R_{\text{(iii)}}$ follows identically to $R_{\text{(ii)}}$.

Collect all sets in the set
\begin{equation*}
\begin{split}
\mathcal{E}_{T,uvw}^{(j,k)} &:= \setTail{u}{T^{1/2m}} \bigcap\limits_{j\in H} \setTail{v_j}{(hT)^{1/2m}} \\
&\qquad \cap \setTail{\tilde{u}}{T^{1/m}} \bigcap\limits_{j\in H} \setTail{\tilde{w}_{j,l}}{(hT)^{1/m}} \bigcap\limits_{k\in H} \setTail{\tilde{w}_{k,-l}}{(hT)^{1/m}}.
\end{split}
\end{equation*}
Now note that by application of Lemma B.12,
%\Cref{lma:tailbounds},
we can show that all sets, and by extension their intersection, have a probability of at least $1 - CT^{-c}$ for some $c>0$. Take for instance the sets with $x = T^{1/m}$. In that case we can apply Lemma B.12
%\Cref{lma:tailbounds}
with $p = \bar m$ moments to obtain a probability of $1 - C \left(T^{1/m} \right)^{-\bar m} T = 1 - C T^{1-\bar m/m}$, so $c=\bar{m}/m-1>0$. The sets for $p=2\bar m$ moments can be treated similarly. For the sets involving intersections over $j\in H$, Lemma B.12
%\cref{lma:tailbounds} 
can be used with an additional union bound argument: $\P\left(\bigcap\limits_{j\in H}\setTail{d}{x}\right)\geq1-Cx^{-p}hT$. These sets therefore hold with probability at least $1-C(hT)^{-c}$. Since $h$ is non-decreasing, this probability converges no slower than $1-CT^{-c}$.
\end{proof}

\begin{proof}[\bf Proof of Lemma B.14]
Consider the set $\left\lbrace\max\limits_{(j,k)\in H^2}\left\vert\frac{1}{T}\sum\limits_{t=l+1}^{T}\left({w}_{j,t}{w}_{k,t-l}-\E{w}_{j,t}{w}_{k,t-l}\right)\right\vert\leq h^2\chi_T\right\rbrace$. As in Lemma A.3
%\cref{lma:CovarianceCloseness}, 
we use the Triplex inequality \citep{Jiang09Triplex} to show under which conditions this set holds with probability converging to 1. By the union bound,
\begin{equation*}\begin{split}
    &\P\left(\max\limits_{(j,k)\in H^2}\left\vert\frac{1}{T}\sum\limits_{t=l+1}^{T}\left({w}_{j,t}{w}_{k,t-l}-\E{w}_{j,t}{w}_{k,t-l}\right)\right\vert
    \leq h^2\chi_T\right)\\
    &\geq 1-\sum\limits_{(j,k)\in H^2}\P\left(\left\vert\frac{1}{T}\sum\limits_{t=l+1}^{T}\left({w}_{j,t}{w}_{k,t-l}-\E{w}_{j,t}{w}_{k,t-l}\right)\right\vert> h^2\chi_T\right).
\end{split}\end{equation*}
Let $z_{t}={w}_{j,t}{w}_{k,t-l}$:
\begin{equation*}
\begin{split}
&\P\left(\left\vert\sum\limits_{t=l+1}^{T}[z_{t}-\E z_{t}]\right\vert>h^2\chi_T(T)\right)\leq 2q \exp\left( \frac{-Th^4\chi_T^2}{288q^2\kappa_T^2} \right)\\
&\quad + \frac{6}{h^2T\chi_T} \sum\limits_{t=1}^{T} \E\left\vert \E\left(z_t\left\vert\mathcal{F}_{t-q}\right.\right)-\E(z_t) \right\vert	+\frac{15}{h^2T\chi_T}\sum\limits_{t=1}^{T} \E\left[\left\vert z_t \right\vert\boldsymbol{1}_{\left\lbrace \vert z_t\vert>\kappa_T\right\rbrace }\right]\\
&=:R_{(\text{i})}+R_{(\text{ii})}+R_{(\text{iii})}.
\end{split}
\end{equation*}

We treat the first term last, as we first need to establish the restrictions put on $\chi_T$, $q$ and $\kappa_T$ from $R_{\text{(ii)}}$ and $R_{\text{(iii)}}$. For the second term, by Lemma B.2(iii)
%\cref{ass:CLT}\ref{ass:CLTmixingales2}
\begin{equation*}\begin{split}
\E\left\vert \E\left(z_t\left\vert\mathcal{F}_{t-q}\right.\right)-\E(z_t) \right\vert\leq c_{t}\psi_{q}\leq C{\psi}_{q} \leq C_1 q^{-d},
\end{split}\end{equation*}
such that $R_{(\text{ii})} \leq C h^{-2}\chi_T^{-1} q^{-d}$. Hence we need $\chi_T^{-1} q^{-1} \to 0$ as $T\to \infty$, such that $\sum\limits_{(j,k)\in H^2}R_{(\text{ii})}\to0$.

For the third term, we have by H\"older's and Markov's inequalities
\begin{equation*}\begin{split}
&\E\left[\left\vert z_t \right\vert\boldsymbol{1}_{\left\lbrace \vert z_t\vert>\kappa_T\right\rbrace }\right]
\leq 
\kappa_T^{1-m/2} \E \abs{z_t}^{m/2}
\end{split}\end{equation*}
so $R_{(\text{iii})} \leq C h^{-2}\chi_T^{-1} \kappa_T^{1-m/2}$. Hence we know that we need to take $\kappa_T$ and $\chi_T$ such that $\chi_T^{-1} \kappa_T^{1-m/2} \to 0$ as $T \to \infty$, giving $\sum\limits_{(j,k)\in H^2}R_{(\text{iii})}\to0$.

Our goal is to minimize $\chi_T$ while ensuring all conditions are satisfied. 
We jointly bound all three terms by a sequence $\eta_T\to 0$: 
\begin{equation*}
\text{(1)} \quad \sum\limits_{(j,k)\in H^2}R_{(\text{i})}\leq Cqh^2 \exp\left( \frac{-Th^4\chi_T^2}{q^2\kappa_T^2} \right)\leq \eta_T, \qquad \text{(2)} \quad C\chi_T^{-1}q^{-d}\leq \eta_T,
\qquad \text{(3)} \quad C\chi_T^{-1} \kappa_T^{1-m/2}\leq \eta_T.
\end{equation*}
For the steps below, we assume that $\frac{\eta_T}{h^2}\leq\frac{1}{e}\implies \sqrt{-\ln(\eta_T/(qh^2))}\geq1$.  First, isolate $\kappa_T$ in (1) and (2),
\begin{equation*}
Cqh^2 \exp\left( \frac{-Th^4\chi_T^2}{q^2\kappa_T^2} \right)\leq \eta_T \qquad\Longleftrightarrow \quad\kappa_T\leq C\frac{\sqrt{T}h^2\chi_T}{q}.
\end{equation*}
\begin{equation*}
C\chi_T^{-1} \kappa_T^{1-m/2}\leq \eta_T \qquad \Longleftrightarrow \quad\kappa_T\geq C\left(\frac{1}{\chi_T\eta_T}\right)^{2/(m-2)}.
\end{equation*}
Combining both bounds,
\begin{equation*}\begin{split}
&C_1\left(\frac{1}{\chi_T\eta_T}\right)^{2/(m-2)} \leq C_2\frac{\sqrt{T}h^2\chi_T}{q}\qquad \Longleftrightarrow \quad q \leq C\sqrt{T}h^2\chi_T^{m/(m-2)}\eta_T^{2/(m-2)},
\end{split}\end{equation*}
Isolating $q$ from (2),
\begin{equation*}
C \chi_T^{-1}q^{-d}\leq \eta_T \qquad\Longleftrightarrow \quad q\geq C\left(\frac{1}{\eta_T\chi_T}\right)^{1/d}.
\end{equation*}
Satisfying both bounds on $q$,
\begin{equation*}\begin{split}
C_1\sqrt{T}h^2\chi_T^{m/(m-2)}\eta_T^{2/(m-2)}\geq C_2\left(\frac{1}{\eta_T\chi_T}\right)^{1/d}\qquad
\Longleftrightarrow \quad\chi_T\geq  C\eta_T^{-\frac{2d+m-2}{dm+m-2}} (\sqrt{T}h^2)^{-\frac{1}{1/d+m/(m-2)}}.
\end{split}\end{equation*}
When $\chi_T$ satisfies this lower bound, $\sum\limits_{(j,k)\in H^2}(R_{(\text{i})}+R_{(\text{ii})}+R_{(\text{iii})})\leq 3\eta_T$, and
\begin{equation*}
   \P\left(\max\limits_{(j,k)\in H^2}\left\vert\frac{1}{T}\sum\limits_{t=l+1}^{T}\left({w}_{j,t}{w}_{k,t-l}-\E{w}_{j,t}{w}_{k,t-l}\right)\right\vert\leq h^2\chi_T\right)\geq 1-3\eta_T,
\end{equation*}
Which completes the proof. 
\end{proof}

\begin{proof}[\bf Proof of Lemma B.15]
By the definition of $\hat{\boldsymbol{\Theta}}$, it follows directly that $\hat{\boldsymbol{\Theta}}\boldsymbol{X}' = \hat{\boldsymbol{\Upsilon}}^{-2}\hat{\boldsymbol{V}}'$, where $\hat{\boldsymbol{V}} = (\hat{\boldsymbol{v}}_1, \ldots, \hat{\boldsymbol{v}}_N)$, 
such that $\hat{\boldsymbol{\Theta}}\boldsymbol{X}'\boldsymbol{u}/\sqrt{T}=\hat{\boldsymbol{\Upsilon}}^{-2} \hat{\boldsymbol{V}}'\boldsymbol u / \sqrt{T}$.

The proof will now proceed by showing that $\max\limits_{1\leq p\leq P}\abs{\boldsymbol{r}_{N,p} \left(\hat{\boldsymbol{\Theta}}\boldsymbol{X}'\boldsymbol{u} - \boldsymbol{\Upsilon}^{-2} \boldsymbol{V}'\boldsymbol u\right)}/\sqrt{T} \xrightarrow{p} 0$ and $\max\limits_{1\leq p\leq P}\abs{\boldsymbol{r}_{N,p} \Delta} \xrightarrow{p} 0$.
By Lemma B.8, 
%\cref{lma:DeltaNegligible},
it holds that 
\begin{equation*}
\max\limits_{j\in H} \vert\Delta_j\vert\leq\sqrt{T}\lambda^{1-r}{s}_{r}\frac{\bar{\lambda}}{C_1-\eta_T-C_2\bar{\lambda}^{2-r}\smaxN} =: U_{\Delta,T},
\end{equation*}
on the set $\setILcons \cap \setNWcons \cap \setLL$. 
First note that $U_{\Delta,T}\to0$ as the assumption $\lambda_{\max}^2\lambda_{\min}^{-r}\leq\eta_T\left[h^{2/m}\sqrt{T}s_{r,\max}\right]^{-1}$ implies that  $\sqrt{T}\bar{\lambda}\lambda^{1-r}{s}_{r}\to0$ and $\bar{\lambda}^{2-r}\smaxN\to0$.
Regarding $\setILcons \cap \setNWcons \cap \setLL$, it follows from Lemma A.4
%\Cref{lma:empiricalProcess} 
that
$\P\left(\setEP{T}(T\lambda/4)\right)\geq 1-C\frac{N}{T^{m/2}\lambda^m}$, and 
from Lemma B.4 
%\Cref{lma:nodewiseEmpiricalProcess} 
that
$\P\left(\bigcap\limits_{j\in H}\left\lbrace\setEP{T}^{(j)}(T\frac{\lambda_j}{4})\right\rbrace\right)\geq 1-C\frac{hN}{T^{m/2}\underset{\bar{}}{\lambda}^{m}}$; both of these probabilities converge to 1 when $\lambda_{\min}\geq \eta_T^{-1}\frac{(hN)^{1/m}}{\sqrt{T}}$.
By Lemma B.3,
%\cref{lma:nodewiseCovarianceCloseness}, 
$\P\left(\setCC{(S_\lambda)}\bigcap\limits_{j\in H}\setCC{(S_{\lambda,j})}\right)\geq
1-3(1+h)\eta_T^\prime\to1
$  when $h{\eta_T^\prime}\to0$ and 
    \begin{equation*}
    \lambda_{\min}^{-r}s_{r,\max}\leq C\eta_T^{\frac{d+m-1}{dm+m-1}}\left[\frac{\sqrt{T}}{N^{\left(\frac{2}{d}+\frac{2}{m-1}\right)}}\right]^{\frac{1}{\frac{1}{d}+\frac{m}{m-1}}}.
\end{equation*}
For the former condition, we may let
$h\eta_T^\prime\leq\eta_T\implies\eta_T^\prime\leq\eta_Th^{-1}$ and $\eta_T^{\prime -1}\geq\eta_T^{-1} h$, and combining this with the latter condition we require that
\begin{equation*}
    \lambda_{\min}^{-r}s_{r,\max}\leq C\eta_T^{\frac{d+m-1}{dm+m-1}}\left[\frac{\sqrt{T}}{\left(hN\right)^{\left(\frac{2}{d}+\frac{2}{m-1}\right)}}\right]^{\frac{1}{\frac{1}{d}+\frac{m}{m-1}}},
\end{equation*}
which we assume in this lemma. Note that this bound makes redundant the previous bound $\lambda_{\min}\geq \eta_T^{-1}\frac{(hN)^{1/m}}{\sqrt{T}}$ when $0<r<1$, by arguments similar to those in the proof of Theorem 1.
%\cref{thm:ourContribution}.
The probability of $\setLL$ converges to 1 by Lemma B.5
%\cref{ass:tailBound} 
when $\delta_T\leq  C\eta_{T,1}(\sqrt{T}h)^{\frac{1}{1/d+m/(m-1)}}$. We may therefore let $\delta_T=  C\eta_{T,1}(\sqrt{T}h)^{\frac{1}{1/d+m/(m-1)}}$, where $\eta_{T,1}$ will be addressed later in the proof.  
We assume that $\max\limits_{1\leq p\leq P}\norm{\boldsymbol{r}_{N,p}}_1<C$, from which it follows that $\max\limits_{1\leq p\leq P}\abs{\boldsymbol{r}_{N,p}\Delta} \leq \norm{\boldsymbol{r}_{N,p}}_1 \max\limits_{j\in H} \vert\Delta_j\vert \rightarrow 0$. Similarly
\begin{equation*}
    \max\limits_{1\leq p\leq P}\abs{\boldsymbol{r}_{N,p} \left(\hat{\boldsymbol{\Theta}}\boldsymbol{X}'\boldsymbol{u} - \boldsymbol{\Upsilon}^{-2} \boldsymbol{V}'\boldsymbol u\right)}/\sqrt{T}
    \leq \max\limits_{1\leq p\leq P}\norm{\boldsymbol{r}_{N,p}}_1\max\limits_{j\in H}\frac{1}{\sqrt{T}} \left\vert\frac{\hat{\boldsymbol{v}}_{j}'\boldsymbol{u}}{\hat\tau_j^2}-\frac{\boldsymbol{v}_{j}'\boldsymbol{u}}{\tau_j^2}\right\vert. 
\end{equation*}
By Lemma B.11,
%\Cref{lma:nodewiseConsistency}, 
on the set
\begin{equation*}
\mathcal{E}_{V,T} := \setEP{T}(T\lambda/4)\cap \setNWcons\cap\setLL\bigcap\limits_{j\in H} 
\setEPvuj{T}(h^{1/m}T^{1/2}\eta_T^{-1})
\end{equation*}
it holds that
\begin{equation*}\begin{split}
\max\limits_{j\in H}\frac{1}{\sqrt{T}} \left\vert\frac{\hat{\boldsymbol{v}}_{j}'\boldsymbol{u}}{\hat\tau_j^2}-\frac{\boldsymbol{v}_{j}'\boldsymbol{u}}{\tau_j^2}\right\vert
\leq \frac{h^{1/m}\eta_{T,2}^{-1}\boundCT+C_1h^{1/m}\eta_{T,2}^{-1}\sqrt{T}\lambda_{\max}^{2-r}\smaxN+C_2h^{1/m}\eta_{T,2}^{-1}\sqrt{\bar{\lambda}^2\underset{\bar{}}{\lambda}^{-r}\smaxN}}{C_3-C_4\left(\boundCT +C_1\bar{\lambda}^{2-r}\smaxN+C_2\sqrt{\bar{\lambda}^2\underset{\bar{}}{\lambda}^{-r}\smaxN}\right)}=: U_{V,T}.
\end{split}\end{equation*}
Plugging in our choice of $\delta_T$ into the first term in the numerator, 
\begin{equation*}
    h^{1/m}\eta_{T,2}^{-1}\boundCT=C(\eta_{T,1}\eta_{T,2})^{-1} h^{1+1/m}(\sqrt{T}h)^{-\frac{1}{1/d+m/(m-1)}}=C(\eta_{T,1}\eta_{T,2})^{-1} \left(\frac{h^{\frac{m+1}{dm}+\frac{2}{m-1}}}{\sqrt{T}}\right)^{\frac{1}{1/d+m/(m-1)}}.
\end{equation*} 
We may choose $\eta_{T,1}$ and $\eta_{T,2}$ such that $(\eta_{T,1}\eta_{T,2})^{-1}$ grows arbitrarily slowly. Therefore, this term converges to 0 when $\frac{h^{\frac{m+1}{dm}+\frac{2}{m-1}}}{\sqrt{T}}\to0$. The two other terms in the numerator then converge to 0 when $\lambda_{\max}^2\lambda_{\min}^{-r}\leq\eta_T\left[h^{2/m}\sqrt{T}s_{r,\max}\right]^{-1}$. Under these rates the denominator then converges to $C_3$, which gives  $U_{V,T} \to 0$.
The only new set appearing in $\mathcal{E}_{V,T}$ is $\bigcap\limits_{j\in H}\setEPvuj{T}(h^{1/m}T^{1/2}\eta_T^{-1})$, whose probability converges to 1 by Lemma B.10.
%\Cref{lma:vuempirical}.
It follows directly that 
\begin{equation*}
\abs{\boldsymbol{R}_N \left(\hat{\boldsymbol{\Theta}}\boldsymbol{X}'\boldsymbol{u} - \boldsymbol{\Upsilon}^{-2} \boldsymbol{V}'\boldsymbol u\right)}/\sqrt{T}\xrightarrow{p} 0. \qedhere
\end{equation*}
\end{proof}

\begin{proof}[\bf Proof of Lemma B.16]
The following bounds on $R^{\Omega}_{N,T}$ and $R^{\beta}_{N,T}$ hold on the set
\begin{equation*}
\setILcons \cap \setNWcons \cap \setLL \cap \mathcal{E}_{T,uvw} \cap \setEP{T,ww}\left(\eta_T^{-1}h^2\left(\sqrt{T}h^2\right)^{-\frac{1}{1/d+m/(m-2)}}\right),
\end{equation*}
which holds with probability converging to 1 when 
$\lambda_{\max}^2\lambda_{\min}^{-r}\leq\eta_T\left[h^{2/m}\sqrt{T}s_{r,\max}\right]^{-1}$, $\frac{h^{\frac{m+1}{dm}+\frac{2}{m-1}}}{\sqrt{T}}\to0$,
$\lambda_{\min}^{-r}s_{r,\max}\leq C\eta_T^{\frac{d+m-1}{dm+m-1}}\left[\frac{\sqrt{T}}{\left(hN\right)^{\left(\frac{2}{d}+\frac{2}{m-1}\right)}}\right]^{\frac{1}{\frac{1}{d}+\frac{m}{m-1}}}$,
and, if $r=0$, $\lambda_{\min}\geq \eta_T^{-1}\frac{(hN)^{1/m}}{\sqrt{T}}$, see the proof of Theorem 3 
%\cref{thm:LRVconsistency}
for details. Under Assumption 6,
%\cref{ass:HDCLT},
{$m$ and $d$ may be arbitrarily large}, and assuming polynomial growth rates allows us to  simplify these conditions to the following:
\begin{equation*}
    \begin{split}
    0<r<1:&\ \frac{1/2+b}{2-r}<\ell<\frac{1/2-b}{r},\\
    r=0:&\ \frac{1/2+b}{2-r}<\ell<1/2.
    \end{split}
\end{equation*}
These bounds are feasible when $b<\frac{1-r}{2}$.
By (B.2)
%\cref{eq:def_deltaT}
\begin{equation*} 
R_{N,T}^{\Omega}\leq C_1\Delta\tau \left[1 + \Delta \tau + \Delta \tau \Delta \omega \right] + C_2 Q_T^{1-d-\delta},
\end{equation*}
where $\delta>0$,
\begin{equation*}
    \Delta\tau=\max\limits_{j\in H}\abs{\frac{1}{\hat{\tau}_j^2}-\frac{1}{{\tau}_j^2}}\leq \frac{ \boundCT +C_1\bar{\lambda}^{2-r}\smaxN+C_2\sqrt{\bar{\lambda}^2\underset{\bar{}}{\lambda}^{-r}\smaxN}}{C_3-C_4\left( \boundCT +C_1\bar{\lambda}^{2-r}\smaxN+C_2\sqrt{\bar{\lambda}^2\underset{\bar{}}{\lambda}^{-r}\smaxN}\right)},
\end{equation*}
with $\delta_T=  C\eta_{T,1}(\sqrt{T}h)^{\frac{1}{1/d+m/(m-1)}}$, and
\begin{equation*}
\begin{split}
\Delta\omega=\max\limits_{(j,k)\in H^2}\abs{\hat\omega_{j,k} - \omega_{j,k}^{N,Q_T}} 
&\leq \left(2 Q_T + 1 \right) \left[ C_1 \left[T^{1/2} \lambda_{\max}^{2-r} s_{r,\max}\right]^2
+ C_2h^{\frac{1}{m}} T^{\frac{1}{m}} \lambda_{\max}^{2-r} s_{r,\max}\right.\\
&\quad + C_3 \sqrt{h^{\frac{3}{m}}T^{\frac{3-m}{m}} \lambda_{\max}^{2-r} s_{r,\max}}+C_4\left[h^{\frac{1}{3m}}T^{\frac{m+1}{3m}}\lambda_{\max}^{2-r} s_{r,\max}\right]^{\frac{3}{2}}\\
&\quad \left. +  C_5 \eta_T^{-1}h^2\left(\sqrt{T}h^2\right)^{-\frac{1}{1/d+m/(m-2)}} \right].
\end{split}
\end{equation*}

$Q_T^{1-d-\delta}$ is dominated by the term $ C_1\Delta\tau \left[1 + \Delta \tau + \Delta \tau \Delta \omega \right]$, since $d$ may be arbitrarily large, and we can limit the analysis to $\Delta\tau$ and $\Delta\omega$. 

For $\Delta\tau$, we first consider the numerator of the upper bound
\begin{equation*}\begin{split}
    \boundCT +C_1\bar{\lambda}^{2-r}\smaxN+C_2\sqrt{\bar{\lambda}^2\underset{\bar{}}{\lambda}^{-r}\smaxN}=&O\left(T^{\mcH-(\mcH+1/2)\frac{1}{1/d+m/(m-1)}}+T^{b-\ell(2-r)}+T^{\frac{1}{2}(b-\ell(2-r))}\right)\\
    =&O\left(T^{\epsilon-1/2}+T^{b-\ell(2-r)}+T^{\frac{1}{2}(b-\ell(2-r))}\right),
\end{split}\end{equation*}
for some arbitrarily small $\epsilon>0$. From the earlier conditions, $\frac{1/2+b}{2-r}<\ell\implies b-\ell(2-r)<-1/2$, which implies that the numerator converges to 0, and that it converges at the rate of $O\left(T^{\frac{1}{2}(b-\ell(2-r))}\right)$, since the two other terms have a smaller exponent of $T$. The same expression from the numerator also appears in the denominator, so the latter 
converges to a non-zero constant, and $\Delta\tau=O\left(T^{\frac{1}{2}(b-\ell(2-r))}\right)$. 

For $\Delta\omega$, we may simplify the upper bound as follows
\begin{equation*}
\begin{split}
& \left(2 Q_T + 1 \right) \left[ C_1 \left[T^{1/2} \lambda_{\max}^{2-r} s_{r,\max}\right]^2
+ C_2h^{\frac{1}{m}} T^{\frac{1}{m}} \lambda_{\max}^{2-r} s_{r,\max}\right.\\
&\quad + C_3 \sqrt{h^{\frac{3}{m}}T^{\frac{3-m}{m}} \lambda_{\max}^{2-r} s_{r,\max}}+C_4\left[h^{\frac{1}{3m}}T^{\frac{m+1}{3m}}\lambda_{\max}^{2-r} s_{r,\max}\right]^{\frac{3}{2}}
\left. +  C_5 \eta_T^{-1}h^2\left(\sqrt{T}h^2\right)^{-\frac{1}{1/d+m/(m-2)}} \right]\\
&=O\left(T^{\mathcal{Q}}\left[T^{2(1/2+b-\ell(2-r))}+T^{\epsilon+b-\ell(2-r)}+T^{\epsilon+\frac{1}{2}(-1+b-\ell(2-r))}+T^{\epsilon+\frac{3}{2}\left(1/3+b-\ell(2-r)\right)}+T^{\epsilon-1/2}\right]\right)\\
&=O\left(T^{\mathcal{Q}+2(1/2+b-\ell(2-r))}+T^{\mathcal{Q}+\epsilon-1/2}\right).
\end{split}
\end{equation*}
Since $\Delta\tau\to0$,
\begin{equation*}\begin{split}
    \Delta\tau \left[1 + \Delta \tau + \Delta \tau \Delta \omega \right]=&O\left(\Delta\tau+[\Delta\tau]^2\Delta\omega\right)\\
    =&O\left(T^{\frac{1}{2}(b-\ell(2-r))}+T^{\mathcal{Q}+1+3(b-\ell(2-r))}+T^{\mathcal{Q}-1/2+(b-\ell(2-r))}\right). 
\end{split}\end{equation*}
When $\mathcal{Q}<\min\left\lbrace-1-\frac{5}{6}(b-\ell(2-r)), \frac{1}{2}-\frac{1}{2}(b-\ell(2-r))\right\rbrace$, the first term dominates the others, and $R^{\Omega}_{N,T}=O\left(T^{\frac{1}{2}(b-\ell(2-r))}\right)$. Note that since $b-\ell(2-r)<-1/2$, this bound on $\mathcal{Q}$ is satisfied when $\mathcal{Q}<2/3$.
Following the proof of Lemma B.15,
%\cref{lma:CLTpreliminary}, 
\begin{equation}\begin{split}
R_{N,T}^\beta:= \max\limits_{1\leq p\leq P}\abs{\boldsymbol{r}_{N,p}\left(\frac{\hat{\boldsymbol{\Theta}}\boldsymbol{X}'\boldsymbol{u}}{\sqrt{T}}+\Delta-\frac{\boldsymbol{\Upsilon}^{-2} \boldsymbol{V}'\boldsymbol u}{\sqrt{T}}\right)}\leq U_{\Delta,T}+U_{V,T},
\end{split}\end{equation}
where
\begin{equation*}
U_{\Delta,T}=\sqrt{T}\lambda^{1-r}{s}_{r}\frac{\bar{\lambda}}{C_1-\eta_T-C_2\bar{\lambda}^{2-r}\smaxN}, 
\end{equation*}
and 
\begin{equation*}
    U_{V,T}=\frac{h^{1/m}\eta_{T}^{-1}\boundCT+C_1h^{1/m}\eta_{T}^{-1}\sqrt{T}\lambda_{\max}^{2-r}\smaxN+C_2h^{1/m}\eta_{T}^{-1}\sqrt{\bar{\lambda}^2\underset{\bar{}}{\lambda}^{-r}\smaxN}}{C_3-C_4\left(\boundCT +C_1\bar{\lambda}^{2-r}\smaxN+C_2\sqrt{\bar{\lambda}^2\underset{\bar{}}{\lambda}^{-r}\smaxN}\right)},
\end{equation*}
with $\delta_T=  C\eta_{T,1}(\sqrt{T}h)^{\frac{1}{1/d+m/(m-1)}}$.
For $U_{\Delta,T}$, the numerator is of order $O\left(T^{1/2+b-\ell(2-r)}\right)$, and the denominator of order $O\left(1+T^{b-\ell(2-r)}\right)=O(1)$, so $U_{\Delta,T}=O\left(T^{1/2+b-\ell(2-r)}\right)$. For $U_{V,T}$, note that each term in the numerator is multiplied by $h^{1/m}\eta_T^{-1}$, which we can take to be $O(T^{\epsilon})$ for an arbitrarily small $\epsilon>0$. The remainder of the numerator is then 
\begin{equation*}\begin{split}
    \boundCT +C_1\sqrt{T}{\lambda}_{\max}^{2-r}\bar{s_{r}}+C_2\sqrt{\bar{\lambda}^2\underset{\bar{}}{\lambda}^{-r}\smaxN}=&O\left(T^{\mcH-(\mcH+1/2)\frac{1}{1/d+m/(m-1)}}+T^{1/2+b-\ell(2-r)}+T^{\frac{1}{2}(b-\ell(2-r))}\right)\\
    =&O\left(T^{\epsilon-1/2}+T^{1/2+b-\ell(2-r)}+T^{\frac{1}{2}(b-\ell(2-r))}\right),\\
    =&O\left(T^{\epsilon-1/2}+T^{1/2+b-\ell(2-r)}\right).\\
\end{split}\end{equation*}
Since the denominator contains the same expression as $\Delta\tau$, it 
converges to a non-zero constant, and $U_{V,T}=O\left(T^{\epsilon}\left[T^{-1/2}+T^{\frac{1}{2}(b-\ell(2-r))}\right]\right)$. Combining these terms,
\begin{equation*}
    R^{\beta}_{N,T}=O\left(T^{1/2+b-\ell(2-r)}+T^{\epsilon}\left[T^{-1/2}+T^{\frac{1}{2}(b-\ell(2-r))}\right]\right)=O\left(T^{\epsilon-1/2}+T^{1/2+b-\ell(2-r)}\right).
\end{equation*}
Finally, as mentioned at the start of the proof, these results hold on a set whose probability converges to 1. We therefore replace $O(\cdot)$ with $O_p(\cdot)$ and the proof is complete.
\end{proof}

\subsection{Illustration of conditions for Corollary 1}\label{sec:simplifiedRatesLasso}

\begin{example}\label{cor:simplifiedRatesLasso}
    The requirements of Corollary 1 are satisfied when 
    $N\sim T^{a}$ for $a>0$, $s_{r}\sim T^{b}$ for $b>0$, and $\lambda\sim T^{-\ell}$ for \begin{equation*}
    \begin{split}
    0<r<1:&\quad \frac{b}{1-r}<\ell<\frac{1}{r(\frac{1}{d}+\frac{m}{m-1})}\left[\frac{1}{2}-b\left(\frac{1}{d}+\frac{m}{m-1}\right)-2a\left(\frac{1}{d}+\frac{1}{m-1}\right)\right],\\
    r=0:&\quad\frac{b}{1-r}<\ell<\frac{1}{2}-\frac{a}{m}.
    \end{split}
\end{equation*}
    This choice of $\ell$ is feasible when 
    \begin{equation}\label{eq:exampleRequirementsLasso}
        \left(\frac{2b}{1-r}\right)\left(\frac{1}{d}+\frac{m}{m-1}\right)+4a\left(\frac{1}{d}+\frac{1}{m-1}\right)<1.
    \end{equation}

\Cref{fig:momentsRequiredLasso} demonstrates which values of $a$, $b$, $m$, $d$, and $r$ are feasible, as well as how many moments $m$ are required for different combinations of the other parameters.
\end{example}
\begin{figure}[ht]
\caption{Required moments $m$ implied by \cref{eq:exampleRequirementsLasso}. Contours mark intervals of 10 moments, and values above $m=100$ are truncated to 100. Non-shaded areas indicate infeasible regions.}\label{fig:momentsRequiredLasso}
\centering
\includegraphics[width=\textwidth]{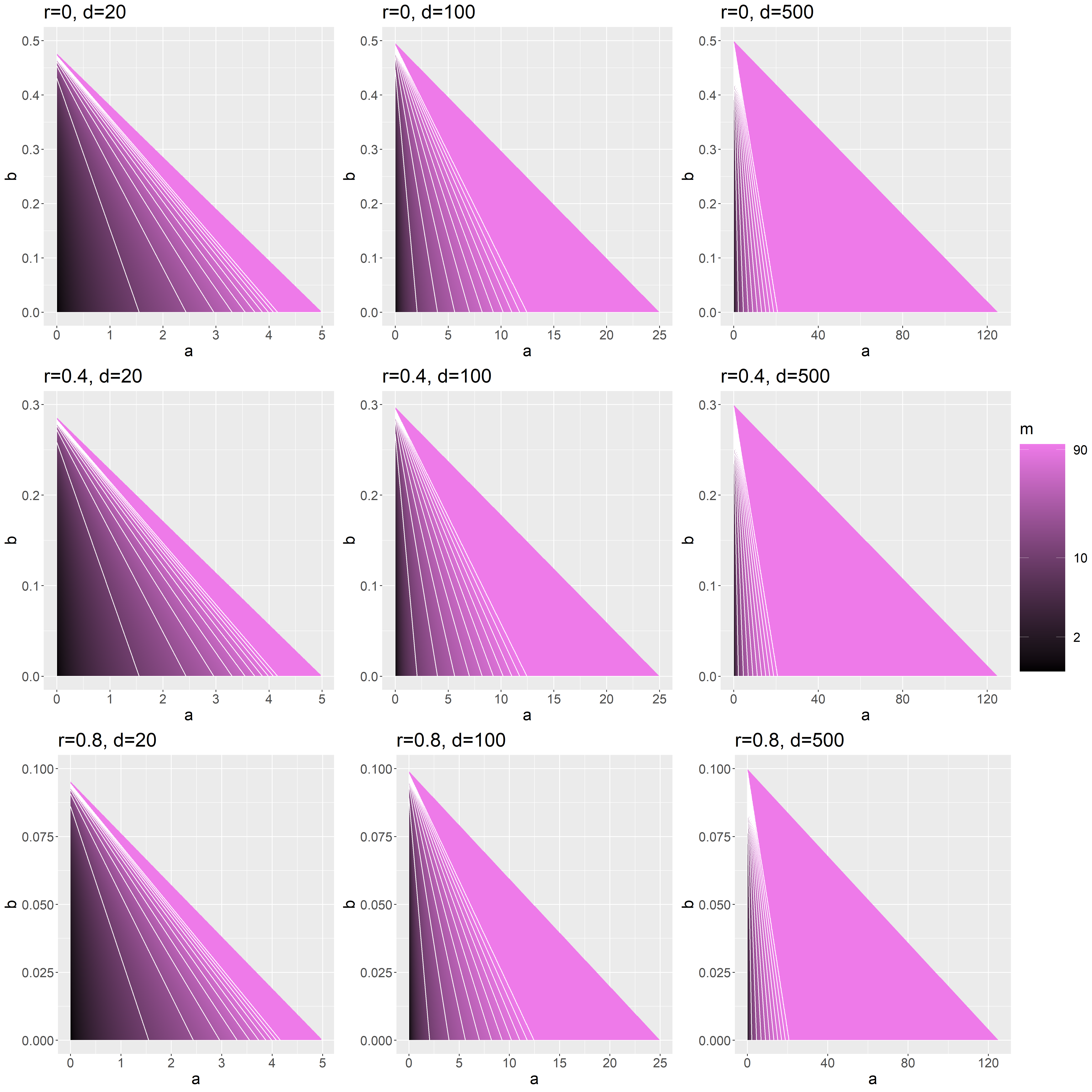}
\end{figure}

\subsection{Additional notes on Examples 5 and 6}\label{sec:supplS4examples}
We start with a lemma on useful properties of the matrix pseudo-norm induced by the $\norm{\cdot}_r$ pseudo-norm. Its proof is omitted, but available upon request. We then provide further details on Examples 5 and 6.

\begin{lemma}\label{lma:inducedNormProp}
    For matrices $\bA,\bB\in \mathds{R}^{n\times m}$ with column vectors $\ba_j$ and $\bb_j$,  define the induced matrix pseudo-norm $\norm{\bA}_r=\max\limits_{\bx\neq\bzero}\frac{\norm{\bA\bx}_r}{\norm{\bx}_r}=\max\limits_{\norm{\bx}_r=1}\norm{\bA\bx}_r$, where for a vector $\bx$ the pseudo-norm $\norm{\bx}_r=\left(\sum_{j}\abs{x_j}^r\right)^{1/r}$. For $0<r<1$, the following hold
    \begin{enumerate}[label=(\arabic*)]
        \item $\norm{c A}_r=\abs{c}\norm{A}_r$,
        \item $\norm{\bA}_r=\max_j\norm{\ba_j}_r$,
        \item $\norm{\bA\bB}_r\leq\norm{\bA}_r\norm{\bB}_r$,
        \item $\norm{\bA+\bB}_r^r\leq\norm{\bA}_r^r+\norm{\bB}_r^r$.
        \item $m^{1/2-1/r} \norm{\bA}_2 \leq \norm{\bA}_r\leq n^{1/r-1/2}\norm{\bA}_2$
    \end{enumerate}
    For $r=0$, the induced matrix pseudo-norm $\norm{\bA}_0=\max\limits_{\bx\neq\bzero}\frac{\norm{\bA\bx}_0}{\norm{\bx}_0}$, with $\norm{\bx}_0=\sum_{j}\mathds{1}_{\{\abs{x_j}>0\}}$. The above properties also hold for this norm, except property (5), and (1), which is replaced by 
    \begin{enumerate}
        \item[(1')] $\norm{c \bA}_0=\norm{\bA}_0$, for $c\neq 0$.
    \end{enumerate}
\end{lemma}

\bigskip
\subsubsection{\label{supp:SparseFactor} Example 5: Sparse factor model}
\noindent Recall the factor model
\begin{equation*}\begin{split}
y_t&=\boldsymbol{\beta^0}'\boldsymbol{x}_{t}+u_t,\ u_t\sim IID(0,1)\\
\boldsymbol{x}_t&=\underset{N\times k}{\bLambda}\underset{k\times 1}{\bbf_t}+\boldsymbol{\nu}_t,\ \boldsymbol{\nu}_t \sim IID(\boldsymbol{0},\bSigma_{\bnu}),\qquad
\bbf_t \sim IID(\bzero,\bSigma_{\bbf}),
\end{split}\end{equation*}
where $\bLambda$ has bounded elements, $\bSigma_{\bbf}$ and $\bSigma_{\bnu}$ are positive definite with bounded eigenvalues, and $\bnu_t$ and $\bbf_t$ uncorrelated. We make the following assumptions on the factor loadings:
\begin{equation}
C_1 N^a \leq \lambda_{\min} (\bLambda' \bLambda) \leq \lambda_{\max} (\bLambda' \bLambda) \leq C_2 N^{b}, \qquad 0 < a \leq b \leq 1.
\end{equation}
These assumptions imply that the $k$ largest eigenvalues of $\bSigma=\bLambda\bSigma_{\bbf}\bLambda^\prime+\bSigma_{\bnu}$ diverge at rates between $N^a$ and $N^b$, while the remaining $N-k+1$ eigenvalues do not diverge. This holds as we can bound the largest eigenvalue $\lambda_{\max} (\bSigma)$ from above by
\begin{equation*}
\begin{split}
\lambda_{\max} (\bSigma) &\leq \lambda_{\max} (\bLambda \bSigma_{\bbf} \bLambda') + \lambda_{\max} (\bSigma_{\bnu})
\leq \lambda_{\max} (\bSigma_{\bbf}) \lambda_{\max} (\bLambda' \bLambda) + \lambda_{\max} (\bSigma_{\bnu}) \leq C_1  N^b + C_2.
\end{split}
\end{equation*}
Similarly, we can bound bound the $k$-th largest eigenvalue $\lambda_{k} (\bSigma)$ using Weyl's inequality and the min-max theorem from below by
\begin{equation*}
\begin{split}
\lambda_{k} (\bSigma) &\geq \lambda_k (\bLambda \bSigma_{\bbf} \bLambda') + \lambda_{\min} (\bSigma_{\bnu}) 
= \max_{\mathcal{U}} \left\{\left.\min_{\bx\in\mathcal{U}\setminus\bzero} \frac{\bx'\bLambda \bSigma_{\bbf} \bLambda' \bx}{\bx'\bx} \right| \dim(\mathcal{U}) = N-k+1 \right\} + \lambda_{\min} (\bSigma_{\bnu})
\\
&\geq \lambda_{\min} (\bSigma_{\bbf}) \max_{\mathcal{U}} \left\{\left.\min_{\bx\in\mathcal{U}\setminus\bzero} \frac{\bx'\bLambda \bLambda' \bx}{\bx'\bx} \right| \dim(\mathcal{U}) = N-k+1 \right\} + \lambda_{\min} (\bSigma_{\bnu}) \\
&= \lambda_{\min} (\bSigma_{\bbf}) \lambda_{k} (\bLambda \bLambda') + \lambda_{\min} (\bSigma_{\bnu}) = \lambda_{\min} (\bSigma_{\bbf}) \lambda_{\min} (\bLambda \bLambda') + \lambda_{\min} (\bSigma_{\bnu}) \geq C_1 N^a + C_2,
\end{split}
\end{equation*}
where we used that $ \lambda_{k} (\bLambda \bLambda') = \lambda_{k} (\bLambda' \bLambda) = \lambda_{\min} (\bLambda' \bLambda)$.

Therefore, this assumption generates a weak factor model if $b<1$, while if $b=1$ but $a<1$ some factors, but not all, are weak; see e.g.~\cite{uematsu2022estimation, uematsu2022inference} and the references therein.\footnote{Our setup corresponds to the framework with factors of varying strength as proposed by \citet{uematsu2022estimation, uematsu2022inference} by setting $\lambda_{j} (\bLambda' \bLambda) \sim N^{a_j}$ where $b = a_1 \geq \ldots \geq a_k = a$.} If $a=b=1$ we have the standard strong factor model with dense loadings.

Sparse factor loadings satisfy these assumptions. In particular, from \cref{lma:inducedNormProp}(5) we find that $\lambda_{\max} (\bLambda' \bLambda)=\norm{\bLambda}_2^2 \leq k^{2/r - 1} \norm{\bLambda}_r^2$; thus, with a fixed number $k$ of factors, the sparsity of $\bLambda$ provides an upper bound for the strength of divergence of the largest eigenvalues.\footnote{This bound only holds for $r>0$. \citet{uematsu2022estimation} consider the case $r=0$.} Sparse factor models may provide accurate descriptions of various economic and financial datasets. For example, \citet{uematsu2022inference} find strong evidence of sparse factor loadings in the FRED-MD macroeconomic dataset \citep{mccracken2016fred}, as well as of firm-level excess returns of the S\&P500 beyond the market return factor. \citet{freyaldenhoven2021identification} uses sparsity in the loadings to identify the factors, motivating the sparsity empirically through the presence of ``local'' factors in economic and financial data. Further empirical evidence for sparse factor models is reviewed in \citet{uematsu2022estimation}.

We now derive the sparsity bound of Example 5.
%\cref{ex:SparseFactor}. 
We bound $\norm{\bgamma_j^0}_r^r$ based on the fact that $\bTheta=\bUpsilon^{-2}\bGamma$, where $\bUpsilon^{-2}=\diag(1/\tau_1^2,\dots, 1/\tau_N^2)$, and 
\begin{equation*}\label{ChatDef} 
{\boldsymbol{\Gamma}}:=\begin{bmatrix}
1      & -{\gamma}_{1,2} & \dots & -{\gamma}_{1,N} \\
-{\gamma}_{2,1}      & 1 & \dots & -{\gamma}_{2,N} \\
\vdots & \vdots &\ddots & \vdots\\
-{\gamma}_{N,1}      & -{\gamma}_{N,2}  & \dots & 1
\end{bmatrix}.
\end{equation*}
This result follows from the definition of $\bgamma_j^0$ as linear projection coefficients, and the block matrix inverse identity for $\bTheta$. Then
\begin{equation*}\begin{split}
\max_j\norm{\bgamma_j^0}^r_r\leq& 1+\max_j\norm{\bgamma_j^0}^r_r=\max_j\norm{(1,\bgamma_j^{0\prime})^\prime}_r^r=\max_j\norm{(1,-\bgamma_j^{0\prime})^\prime}_r^r=\norm{\bGamma}_r^r\\
=&\norm{(\bUpsilon^{-2})^{-1}\bTheta}_r^r\leq\norm{(\bUpsilon^{-2})^{-1}}_r^r\norm{\bTheta}_r^r\leq \max_j \tau_j^{2r} \norm{\bTheta}_r^r\leq C\norm{\bTheta}_r^r,
\end{split}\end{equation*}
where  $\max_j \tau_j^{2r}\leq C$ follows from (B.1)
%\cref{eq:trueTauBounded}. 
Note that when $r=0$, these steps follow similarly, noting that $\norm{(\bUpsilon^{-2})^{-1}}_0=1$, and therefore $C=1$.

By the Woodbury matrix identity
\begin{equation*}
\begin{split}
\bTheta &= \bSigma_{\bnu}^{-1}- \bSigma_{\bnu}^{-1}\bLambda / N^a \left(\bSigma_{\bbf}^{-1}/N^a + \bLambda^\prime \bSigma_{\bnu}^{-1} \bLambda/N^a \right)^{-1} \bLambda^{\prime}\bSigma_{\bnu}^{-1}.
\end{split}
\end{equation*}
Then
\begin{equation*}
\begin{split}
\norm{\bTheta}_r^r &\leq \norm{\bSigma_{\bnu}^{-1}}_r^r + \norm{\bSigma_{\bnu}^{-1}}_r^r \norm{\bLambda / N^a}_r^r \norm{\left(\bSigma_{\bbf}^{-1}/N^a + \bLambda^\prime \bSigma_{\bnu}^{-1} \bLambda/N^a \right)^{-1}}_r^r \norm{\bLambda^{\prime}}_r^r \norm{\bSigma_{\bnu}^{-1}}_r^r.
\end{split}
\end{equation*}
As for positive semidefinite symmetric matrices $\bA$ and $\bB$ we have that 
\begin{equation*}
\norm{(\bA+\bB)^{-1}}_2 \leq \frac{1}{\lambda_{\min}(\bA + \bB) } \leq \frac{1} {\lambda_{\min} (\bA) + \lambda_{\min} (\bB)} \leq \frac{1}{\lambda_{\min} (\bB)},
\end{equation*}
it follows that
\begin{equation*}
\begin{split}
\norm{\left(\bSigma_{\bbf}^{-1}/N^a + \bLambda^\prime \bSigma_{\bnu}^{-1} \bLambda/N^a \right)^{-1}}_2 \leq \frac{1}{\lambda_{\min} \left(\bLambda^\prime \bSigma_{\bnu}^{-1} \bLambda/N^a \right)}\leq \frac{1}{\lambda_{\min} (\bSigma_{\bnu}^{-1}) \lambda_{\min} \left(\bLambda^\prime  \bLambda/N^a \right)}.
\end{split}
\end{equation*}
As $\lambda_{\min} (\bSigma_{\bnu}^{-1}) = 1/\lambda_{\max} (\bSigma_{\bnu}) \geq 1/C$, it follows from our assumptions that $\lambda_{\min} \left(\bLambda^\prime  \bLambda/N^a \right) \geq C$ and therefore $\norm{\left(\bSigma_{\bbf}^{-1}/N + \bLambda^\prime \bSigma_{\bnu}^{-1} \bLambda/N \right)^{-1}}_2 \leq C$. It then also follows from \cref{lma:inducedNormProp}(5) that $\norm{\left(\bSigma_{\bbf}^{-1}/N^a + \bLambda^\prime \bSigma_{\bnu}^{-1} \bLambda/N^a \right)^{-1}}_r^r \leq C k^{1-r/2}$ and
\begin{equation} \label{eq:boundSparseFactor}
\begin{split}
\norm{\bTheta}_r^r &\leq \norm{\bSigma_{\bnu}^{-1}}_r^r + C k^{1-r/2} \norm{\bSigma_{\bnu}^{-1}}_r^r \norm{\bLambda / N^a}_r^r  \norm{\bLambda^{\prime}}_r^r \norm{\bSigma_{\bnu}^{-1}}_r^r.
\end{split}
\end{equation}
With $\norm{\bLambda^{\prime}}_r^r \leq C k$, we then find the bound
\begin{equation*}
\begin{split}
\norm{\bTheta}_r^r &\leq \norm{\bSigma_{\bnu}^{-1}}_r^r + C k^{2-r/2} N^{-r a} \norm{\bSigma_{\bnu}^{-1}}_r^{2r} \norm{\bLambda}_r^r.
\end{split}
\end{equation*}

We provide two examples of $\bSigma_{\bnu}$ such that $\bSigma_{\bnu}^{-1}$ is sparse. For block diagonal structures, this follows trivially, since the inverse maintains the same block diagonal structure. For a Toeplitz structure $\bSigma_{\bnu,i,j}=\rho^{\abs{i-j}}$, by Section 8.8.4 of \cite{gentle2007matrix},  
\begin{equation*}
    \bSigma_{\bnu}^{-1}=\frac{1}{1-\rho^2}\begin{bmatrix}
    1 & -\rho & 0 & \dots & 0 \\
    -\rho & 1+\rho^2 & -\rho & \dots & 0\\
    0 & -\rho & 1+\rho^2 & \dots & 0\\
    \vdots & \vdots & \vdots & \vdots & \vdots\\
    0 & 0 & 0 & \dots & 1
    \end{bmatrix},
\end{equation*}
and we can bound 
\begin{equation*}
    \norm{\bSigma_{\bnu}^{-1}}_r^r=\max\limits_{j}\norm{\bSigma_{\bnu,\cdot,j}^{-1}}_r^r=\norm{\bSigma^{-1}_{\bnu,\cdot,\Ntwo}}_r^r=\frac{\abs{1+\rho^2}^{r}+2\abs{\rho}^r}{\abs{1-\rho^2}^r}\leq C,
\end{equation*}
or simply $\max\limits_{j}\norm{\bSigma_{\bnu,\cdot,j}^{-1}}_0=3$ for $r=0$.

Note that a (potentially weak) factor model without sparse loadings does not yield a sufficiently sparse matrix $\bTheta$ for all values of $r$. In \cref{eq:boundSparseFactor} we may try to bound $\norm{\bLambda}_r^r$ directly using \cref{lma:inducedNormProp}(5) to bound $
\norm{\bLambda / N^a}_r^r \leq  N^{1 + (b- 2a - 1)r/2} \left[\lambda_{\max} (\bLambda' \bLambda/ N^b) \right]^{r/2}$, such that
$\norm{\bTheta}_r^r \leq \norm{\bSigma_{\bnu}^{-1}}_r^r \left(1 + C k^{2-r/2} N^{1 + (b- 2a - 1)r/2} \right)$. This is not a tight enough bound to guarantee sparsity of $\bTheta$. To illustrate, for the standard dense factor model with $a=b=1$ and $k$ fixed, we get $\norm{\bTheta}_r^r \leq C N^{1-r}$. Weaker divergence of the eigenvalues even increases the power of $N$.

\subsubsection{\label{supp:sparseVAR1} Example 6: Sparse VAR(1)}
Recall the sparse VAR(1) model
    \begin{equation*}
        \bz_t=\bPhi \bz_{t-1}+\bu_t,\ \E\bu_t\bu_t^\prime:=\bOmega,\ \E\bu_t\bu_{t-l}^\prime=\bzero,\ \forall l\neq0, 
    \end{equation*}
    with our regression of interest being 
$y_t=\bphi_1\bz_{t-1}+u_{1,t}$.
For Example 6(a)
%\cref{ex:sparseVAR1}(a) 
with a symmetric block-diagonal coefficient matrix $\bPhi$ and the error covariance matrix $\bOmega$ being the identity, we can simplify $\bSigma=\sum\limits_{q=0}^{\infty}\bPhi^{q}\bOmega\bPhi^{\prime q}=\sum\limits_{q=0}^{\infty}\bPhi^{2q}=\left(\bI-\bPhi^2\right)^{-1}$, where $\bPhi^{0}=\bPhi^{\prime 0}=\bI$, and $\bTheta=\bSigma^{-1}=\bI-\bPhi^2$. Note that $\bI-\bA$ is invertible iff 1 is not an eigenvalue of $\bA$. Since the eigenvalues of $\bPhi^2$ are between (and not including) 0 and 1, $\bSigma$ exists. $\bI-\bPhi^2$ inherits the block diagonal structure of $\bPhi$, so we may bound $\max\limits_{j}\norm{\bgamma_j^0}_r^r\leq C \norm{\bTheta}_r^r\leq Cb$.

This result can be extended to the case where
$\bOmega$ has the same block diagonal structure as the VAR coefficient matrix $\bPhi$. 
While the simplified expression for $\bSigma$ provided above no longer holds, both 
$\bSigma$ and $\bSigma^{-1}$ remain block diagonal when $\bOmega$ and $\bPhi$ share  the same block structure.
%Also 
%$\bOmega_{\scalebox{.7}{-j,-j}}$ and therefore $\bOmega_{\scalebox{.7}{-j,-j}}^{-1}$ have the same block diagonal structure as $\bPhi_{\scalebox{.7}{-j}}$. 
As a result, the nonzero structure of $\bgamma_j^0$ remains unaltered.

For Example 6(b)
%\cref{ex:sparseVAR1}(b)
with a diagonal $\bPhi$ and Toeplitz $\bOmega$, we can simplify $ \bSigma=\sum\limits_{q=0}^{\infty}\bPhi^{q}\bOmega\bPhi^{\prime q}=\sum\limits_{q=0}^{\infty}\phi^{2q}\bOmega=\frac{1}{1-\phi^2}\bOmega$ and by similar arguments to \cref{supp:SparseFactor},
%Section 8.8.4 of \cite{gentle2007matrix}, 
\begin{equation*}
    \bTheta=\frac{1-\phi^2}{1-\rho^2}\begin{bmatrix}
    1 & -\rho & 0 & \dots & 0 \\
    -\rho & 1+\rho^2 & -\rho & \dots & 0\\
    0 & -\rho & 1+\rho^2 & \dots & 0\\
    \vdots & \vdots & \vdots & \vdots & \vdots\\
    0 & 0 & 0 & \dots & 1
    \end{bmatrix}
\end{equation*}.
The precision matrix is clearly sparse in this case, and $\max\limits_{j}\norm{\bgamma_j^0}_r^r\leq C \norm{\bTheta}_r^r\leq C$.

Finally, we numerically investigated the extension where the VAR coefficient matrix also has a Toeplitz structure, namely
$\Phi_{i,j}=0.4^{1+\abs{i-j}}$. 
We vary the sample size between $N=10$ and $N=1000$ and display the boundedness in $r$-norm of the parameter vector in the nodewise regressions in Figure \ref{fig:NumericalSparsity} for different values of $r$.
We use a log-scale since this sparsity grows by orders of magnitude for decreasing $r$.
\begin{figure}[ht]
\caption{Example 6(b): We display  $\ln\left(\max\limits_{j}\norm{\bgamma_j^0}_r^r\right)$ for $N$ between 10 and 1000, and $r$ between 0.1 and 0.9.}
\centering
\includegraphics[width=\textwidth]{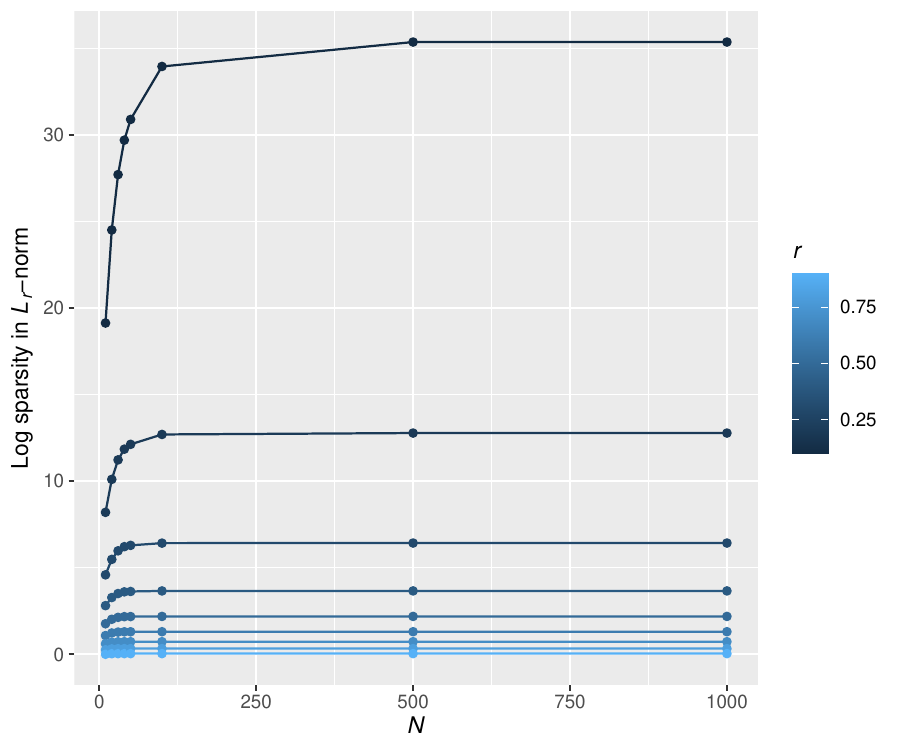}\label{fig:NumericalSparsity}
\end{figure}
\clearpage
\subsection{Algorithmic details for choosing the lasso tuning parameter  
} \label{sec:supplS4.3}
\begin{algorithm}[ht!]
\caption{Plug-in choice of $\lambda$}
\BlankLine
At $k=0$, initialize $\lambda^{(0)}\leftarrow\norm{\bX^\prime \by}_{\infty}/T$ and $\hat\bu^{(0)}\leftarrow\by-\frac{1}{T}\sum_{t=1}^Ty_t$\;
\While{$1\leq k\leq K$}{
Obtain the estimated long-run covariance matrix $\hat\bOmega^{(k)}$ as in (9),
%\cref{eq:LRC},
with $\hat\bXi(l)=\frac{1}{T-l}\sum\limits_{t=l+1}^{T}\bx_t\hat u^{(k-1)}_t\hat u^{(k-1)}_{t-l}\bx_{t-l}^\prime$\;
\While{$1\leq b \leq B$}{
Draw $\hat\bg^{(b)}$ from $N\left(\boldsymbol{0}, \hat\bOmega^{(k)}\right)$\; 
$m_b \leftarrow \norm{\hat\bg^{(b)}}_{\infty}$\;
}
$\lambda^{(k)}\leftarrow c\frac{1}{\sqrt{T}} q_{(1-\alpha)}$, where $q_{(1-\alpha)}$ is the $(1-\alpha)$-quantile of $m_1, \ldots, m_B$\;
\eIf{$\abs{\lambda^{(k)}-\lambda^{(k-1)}}/\lambda^{(k-1)}< \epsilon$} {$\lambda\leftarrow\lambda^{(k)}$\;
\textbf{break}\;
}{
Estimate $\hat\bbeta^{(k)}$ with the lasso using $\lambda^{(k)}$ as the tuning parameter\; $\hat\bu^{(k)}\leftarrow\by-\bX\hat\bbeta^{(k)}$\;
}
}
$\lambda\leftarrow\lambda^{(K)}$\;
\end{algorithm}

\bigskip

We set $K=15$, $\epsilon=0.01$,
$B=1000$, 
$\alpha=0.05$, and $c=0.8$ throughout the simulation study. 
\newpage
\subsection{Additional simulation details}\label{supp:additionalsim} 
\begin{figure}[ht!]
\caption{Model A, $\rho$ heat map coverage: Contours mark the coverage thresholds at 5\% intervals, from 75\% to the nominal 95\%, from dark green to white respectively. Units on the axes are not proportional to the $\lambda$-value but rather its position in the grid. The value of $\lambda$ is $(10T)^{-1}$ at 0, and increases exponentially to a value that sets all parameters to zero at 50. Plots are based on 100 replications, with colored dots representing combinations of $\lambda$'s selected by PI (purple), AIC (red), BIC (blue), EBIC (yellow).} 
\centering
\includegraphics[width=\textwidth]{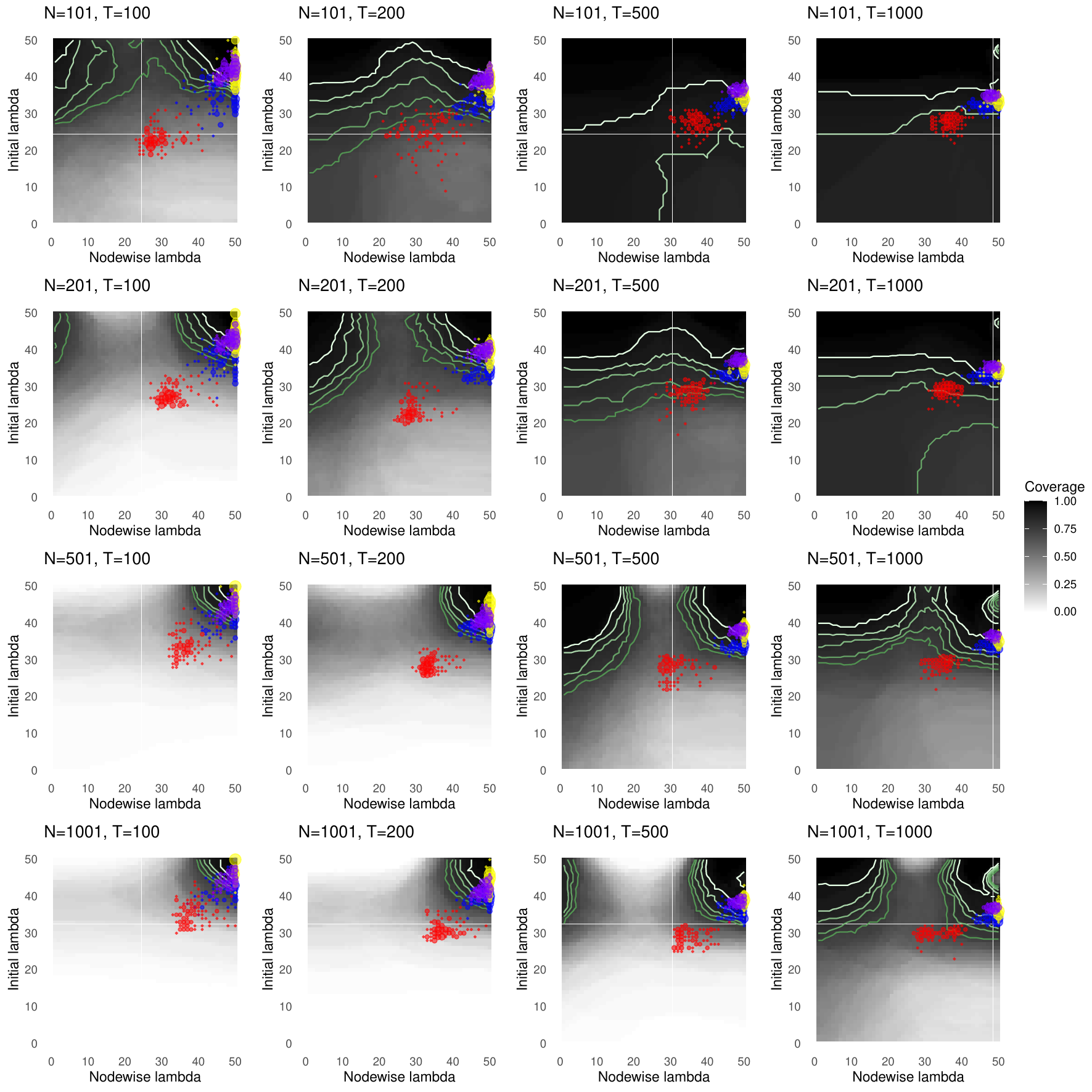}\label{fig:iidrhocoverage}
\end{figure}

\begin{figure}[ht!]
\caption{Model A, $\beta_1$ heat map coverage: Contours mark the coverage thresholds at 5\% intervals, from 75\% to the nominal 95\%, from dark green to white respectively. Units on the axes are not proportional to the $\lambda$-value but rather its position in the grid. The value of $\lambda$ is $(10T)^{-1}$ at 0, and increases exponentially to a value that sets all parameters to zero at 50. Plots are based on 100 replications, with colored dots representing combinations of $\lambda$'s selected by PI (purple), AIC (red), BIC (blue), EBIC (yellow).}
\centering
\includegraphics[width=\textwidth]{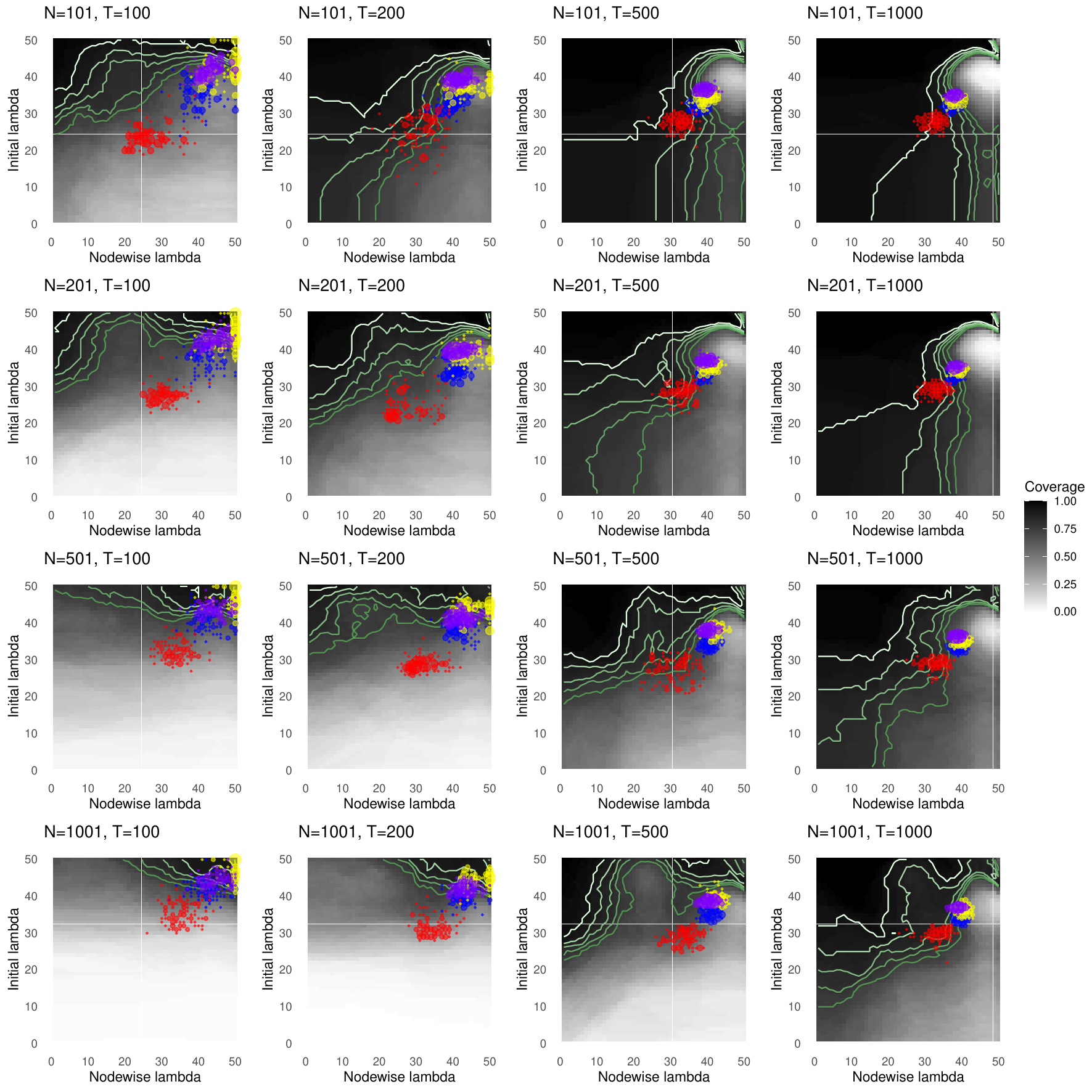}\label{fig:iidbeta1coverage}
\end{figure}

\begin{figure}[ht!]
\caption{Model B, $\rho$ heat map coverage: Contours mark the coverage thresholds at 5\% intervals, from 75\% to the nominal 95\%, from dark green to white respectively. Units on the axes are not proportional to the $\lambda$-value but rather its position in the grid. The value of $\lambda$ is $(10T)^{-1}$ at 0, and increases exponentially to a value that sets all parameters to zero at 50. Plots are based on 100 replications, with colored dots representing combinations of $\lambda$'s selected by PI (purple), AIC (red), BIC (blue), EBIC (yellow).}
\centering
\includegraphics[width=\textwidth]{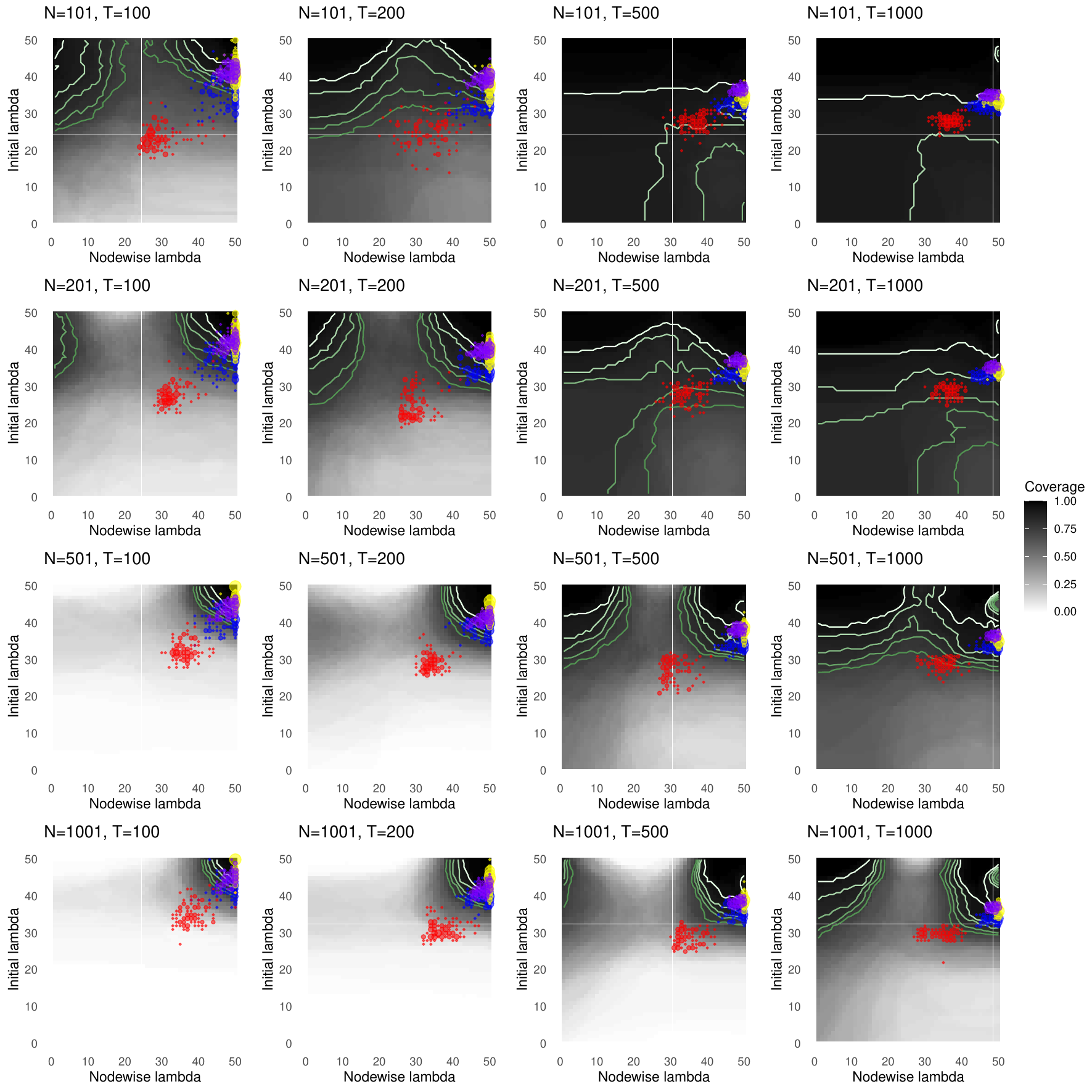}\label{fig:garchrhocoverage}
\end{figure}

\begin{figure}[ht!]
\caption{Model B, $\beta_1$ heat map coverage: Contours mark the coverage thresholds at 5\% intervals, from 75\% to the nominal 95\%, from dark green to white respectively. Units on the axes are not proportional to the $\lambda$-value but rather its position in the grid. The value of $\lambda$ is $(10T)^{-1}$ at 0, and increases exponentially to a value that sets all parameters to zero at 50. Plots are based on 100 replications, with colored dots representing combinations of $\lambda$'s selected by PI (purple), AIC (red), BIC (blue), EBIC (yellow).}
\centering
\includegraphics[width=\textwidth]{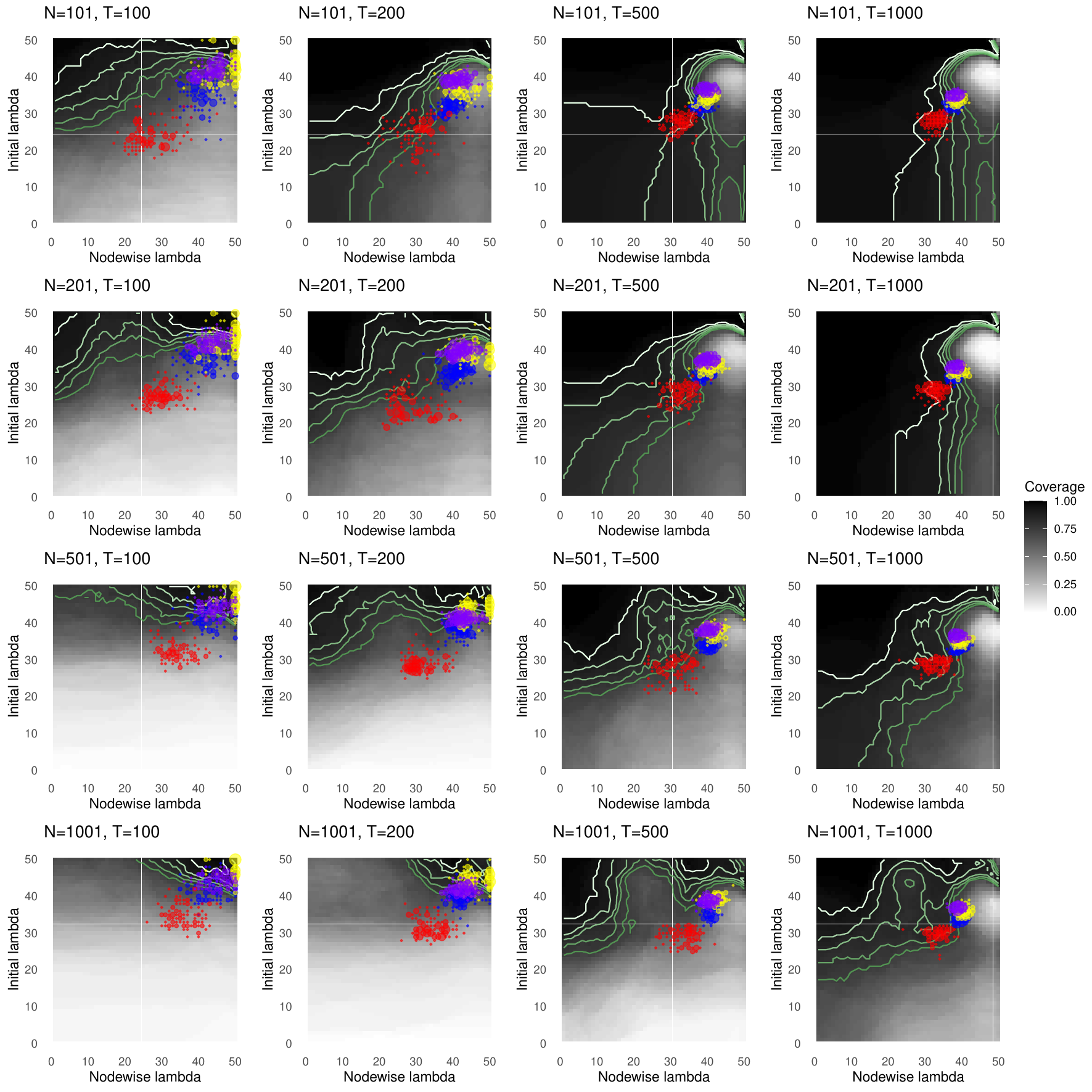}\label{fig:garchbeta1coverage}
\end{figure}

\begin{figure}[ht!]
\caption{Model C, $\rho$ heat map coverage: Contours mark the coverage thresholds at 5\% intervals, from 75\% to the nominal 95\%, from dark green to white respectively. Units on the axes are not proportional to the $\lambda$-value but rather its position in the grid. The value of $\lambda$ is $(10T)^{-1}$ at 0, and increases exponentially to a value that sets all parameters to zero at 50. Plots are based on 100 replications, with colored dots representing combinations of $\lambda$'s selected by PI (purple), AIC (red), BIC (blue), EBIC (yellow).}
\centering
\includegraphics[width=\textwidth]{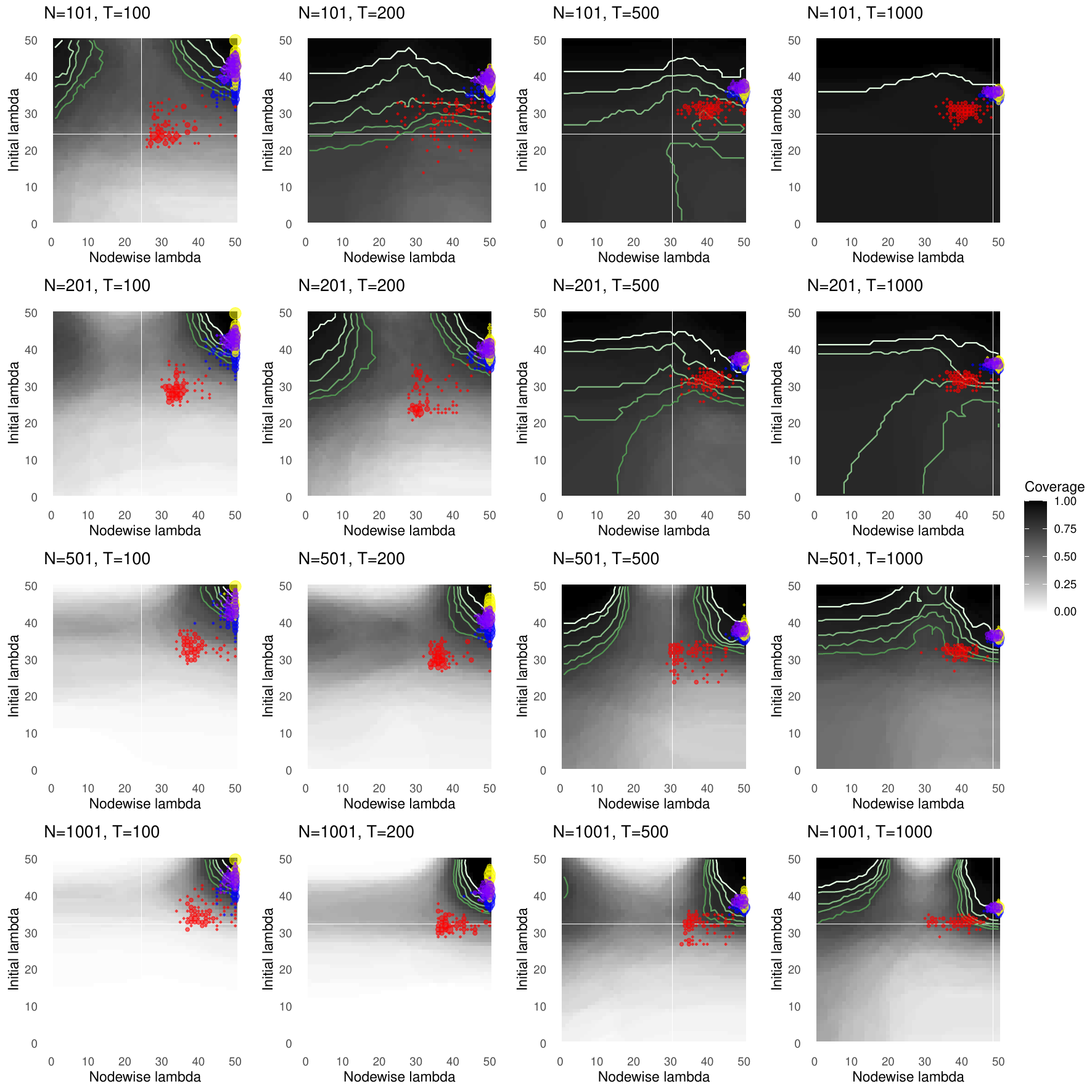}\label{fig:corrrhocoverage}
\end{figure}

\begin{figure}[ht!]
\caption{Model C, $\beta_1$ heat map coverage: Contours mark the coverage thresholds at 5\% intervals, from 75\% to the nominal 95\%, from dark green to white respectively. Units on the axes are not proportional to the $\lambda$-value but rather its position in the grid. The value of $\lambda$ is $(10T)^{-1}$ at 0, and increases exponentially to a value that sets all parameters to zero at 50. Plots are based on 100 replications, with colored dots representing combinations of $\lambda$'s selected by PI (purple), AIC (red), BIC (blue), EBIC (yellow).}
\centering
\includegraphics[width=\textwidth]{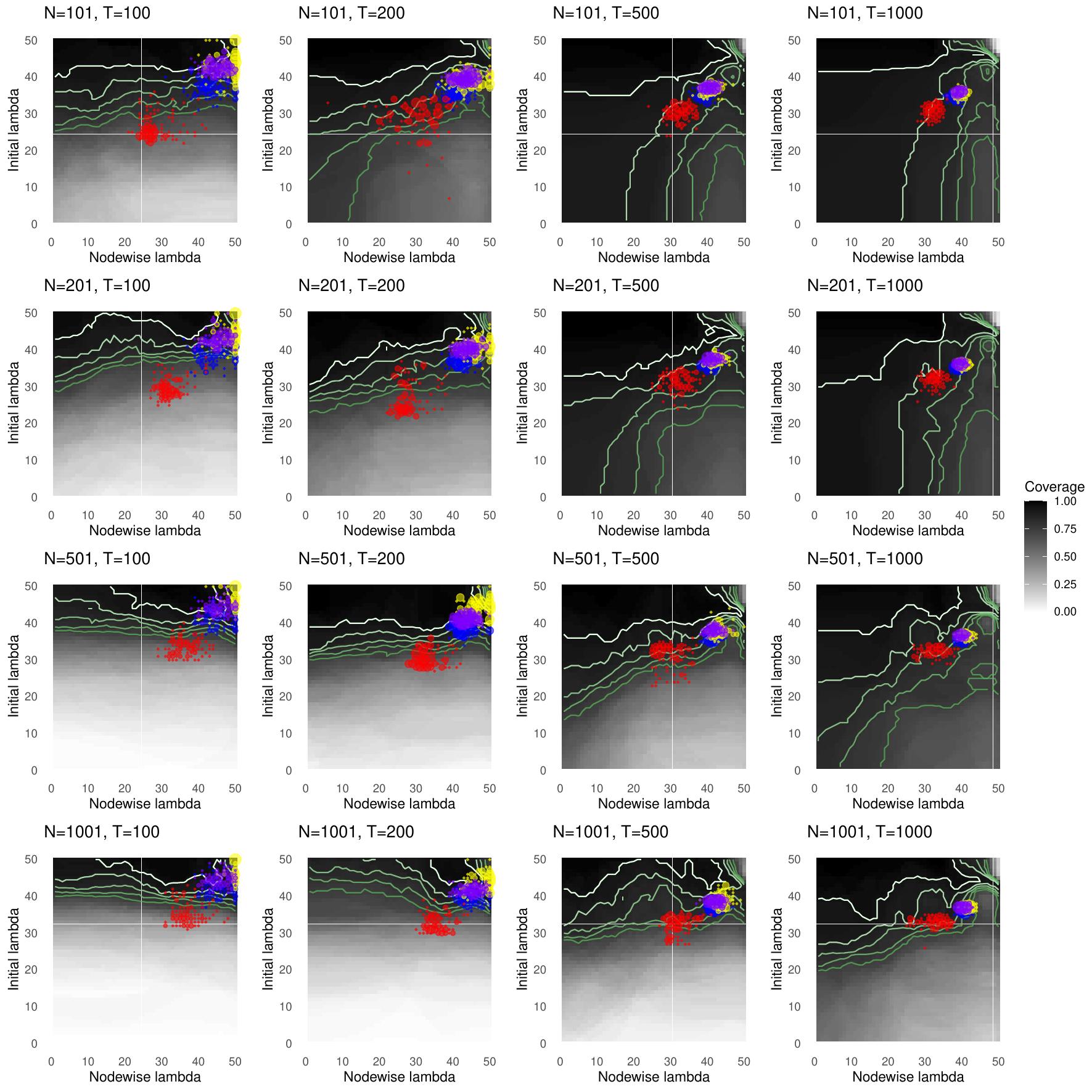}\label{fig:corrbeta1coverage}
\end{figure}

\end{appendices}

\end{document}